\documentclass[11pt]{amsart}
\usepackage{amsthm,amsfonts,amssymb,euscript}
\usepackage{graphics}
\usepackage{xcolor} 
\usepackage{hyperref}

\usepackage{graphicx}
\usepackage{color}
\usepackage{transparent}

\usepackage{fullpage} 

\definecolor{green}{rgb}{0,0.8,0} 

\newtheorem{theorem}{Theorem}[section]
\newtheorem{corollary}[theorem]{Corollary}
\newtheorem{lemma}[theorem]{Lemma}
\newtheorem{proposition}[theorem]{Proposition}
\theoremstyle{definition}
\newtheorem{definition}[theorem]{Definition}

\theoremstyle{remark}
\newtheorem{remark}[theorem]{Remark}
\newtheorem{conjecture}{Conjecture}
\numberwithin{equation}{section}

\newcommand{\abs}[1]{\vert#1\vert}

\newcommand{\set}[1]{\{#1\}}

\newcommand{\ep}{\epsilon}
\def\beaa{\begin{eqnarray*}}
\def\eeaa{\end{eqnarray*}}
\def\bea{\begin{eqnarray}}
\def\eea{\end{eqnarray}}
\def\be{\begin{equation}}
\def\ee{\end{equation}}

\newcommand{\aeq}{\sim}

\newcommand{\ud}{\mathrm{d}}
\newcommand{\rd}{\partial}
\newcommand{\nb}{\nabla}

\newcommand{\bb}{\Big}

\newcommand{\gmm}{\gamma}
\newcommand{\Gmm}{\Gamma}
\newcommand{\dlt}{\delta}

\newcommand{\eps}{\epsilon}

\newcommand{\lmb}{\lambda}
\newcommand{\Lmb}{\Lambda}

\newcommand{\omg}{\omega}
\newcommand{\om}{\omega}
\newcommand{\Omg}{\Omega}

\newcommand{\bfg}{{\bf g}}

\newcommand{\bfR}{{\bf R}}

\newcommand{\bfT}{{\bf T}}

\newcommand{\bbR}{\mathbb R}
\newcommand{\bbS}{\mathbb S}

\newcommand{\calA}{\mathcal A}
\newcommand{\calB}{\mathcal B}

\newcommand{\calD}{\mathcal D}

\newcommand{\calI}{\mathcal I}

\newcommand{\calM}{\mathcal M}

\newcommand{\calQ}{\mathcal Q}
\newcommand{\calR}{\mathcal R}

\newcommand{\calT}{\mathcal T}

\newcommand{\covD}{\nb}
\newcommand{\Ric}{\bfR}
\newcommand{\scal}{R}
\newcommand{\met}{\bfg}
\newcommand{\EM}{\bfT}

\newcommand{\mfd}{\calM}
\newcommand{\PD}{\calQ}

\newcommand{\cpt}{\PD_{\mathrm{cpt}}}
\newcommand{\extr}{\PD_{\mathrm{ext}}}
\newcommand{\intr}{\PD_{\mathrm{int}}}

\newcommand{\uC}{\underline{C}}

\newcommand{\dur}{\nu}
\newcommand{\dvr}{\lmb}

\newcommand{\trpR}{\calT}
\newcommand{\appH}{\calA}
\newcommand{\regR}{\calR}

\newcommand{\TV}{\mathrm{T.V.}}

\newcommand{\pfstep}[1]{\vspace{.5em} {\it #1.}}

\newcommand{\Err}{\mathrm{Err}}

\newcommand{\fparagraph}[1]{\paragraph{\bfseries #1}}

\begin{document}

\title[]{Quantitative decay rates for dispersive solutions to the Einstein-scalar field system in spherical symmetry}
\author{Jonathan Luk}
\address{Department of Mathematics, MIT, Cambridge, MA, 02139}
\email{jluk@math.mit.edu}

\author{Sung-Jin Oh}%
\address{Department of Mathematics, UC Berkeley, Berkeley, CA, 94720}%
\email{sjoh@math.berkeley.edu}%


\begin{abstract}
In this paper, we study the future causally geodesically complete solutions of the spherically symmetric Einstein-scalar field system. Under the \emph{a priori assumption} that the scalar field $\phi$ scatters locally in the scale-invariant bounded-variation (BV) norm, we prove that $\phi$ and its derivatives decay polynomially. Moreover, we show that the decay rates are sharp. In particular, we obtain sharp quantitative decay for the class of global solutions with small BV norms constructed by Christodoulou. As a consequence of our results, for every future causally geodesically complete solution with sufficiently regular initial data, we show the dichotomy that either the sharp power law tail holds or that the spacetime blows up at infinity in the sense that some scale invariant spacetime norms blow up.
\end{abstract}

\maketitle
\section{Introduction}
In this paper, we study the quantitative long time dynamics for the spherically symmetric dispersive spacetimes satisfying the Einstein-scalar field equations. More precisely, these are spherically symmetric solutions $(\mfd,\met,\phi)$ to the Einstein-scalar field system, where $\met$ is a Lorentzian metric and $\phi$ is a real valued function on a $3+1$ dimensional manifold $\mfd$, such that $(\mfd,\met)$ is future causally geodesically complete and $\phi$ scatters locally in the scale-invariant bounded-variation (BV) norm. For these spacetimes, we establish a Price-law type decay for the scalar field $\phi$, the Christoffel symbols associated to $\met$ and all of their derivatives. To obtain the decay results, we do not need to assume any smallness of the initial data. Moreover, we show that the decay rates in this paper are sharp.

The spherically symmetric Einstein-scalar field system, being one of the simplest model of self-gravitating matter in this symmetry class, has been studied extensively both numerically and mathematically. In a seminal series of papers by Christodoulou \cite{Christodoulou:1987ta}, \cite{Christodoulou:1991}, \cite{Christodoulou:1993bt}, \cite{Christodoulou:1994}, \cite{Christodoulou:1999}, he achieved a complete understanding of the singularity structure of spherically symmetric spacetime solutions to this system. The culmination of the results shows that generic\footnote{in the BV class, i.e., the initial data for $\partial_v(r\phi)$ has bounded variation. More precisely, Christodoulou showed that the non-generic set of initial data has co-dimension at least two in the BV topology.} spherically symmetric initial data with one asymptotically flat end give rise to a spacetime whose global geometry is either dispersive (with a Penrose diagram represented by Figure 1) or contains a black hole region $\mathcal B\mathcal H$ which terminates in a spacelike curvature singularity $\mathcal S$ (with a Penrose diagram represented by Figure 2). In particular, in either of these generic scenarios, the spacetime possesses a complete null infinity $\mathcal I^+$ and thus obeys the weak cosmic censorship conjecture. Moreover, in either case the maximal Cauchy development of the data is inextendible with a $C^2$ Lorentzian metric and therefore also verifies the strong cosmic censorship conjecture. We refer the readers to \cite{Kommemi} for a comprehensive discussion on general singularity structures for spherically symmetric spacetimes.

\begin{figure}[htbp] \label{fig.disp}
\begin{center}
 
\includegraphics{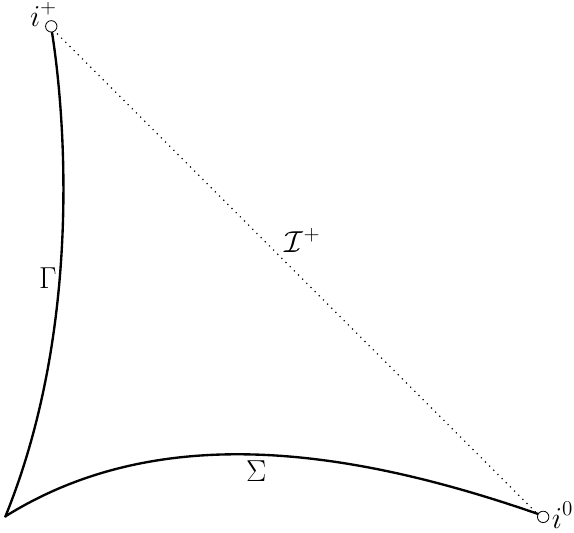}
 
\caption{The dispersive case}
\end{center}
\end{figure}

\begin{figure}[htbp]
\begin{center}
 
\includegraphics{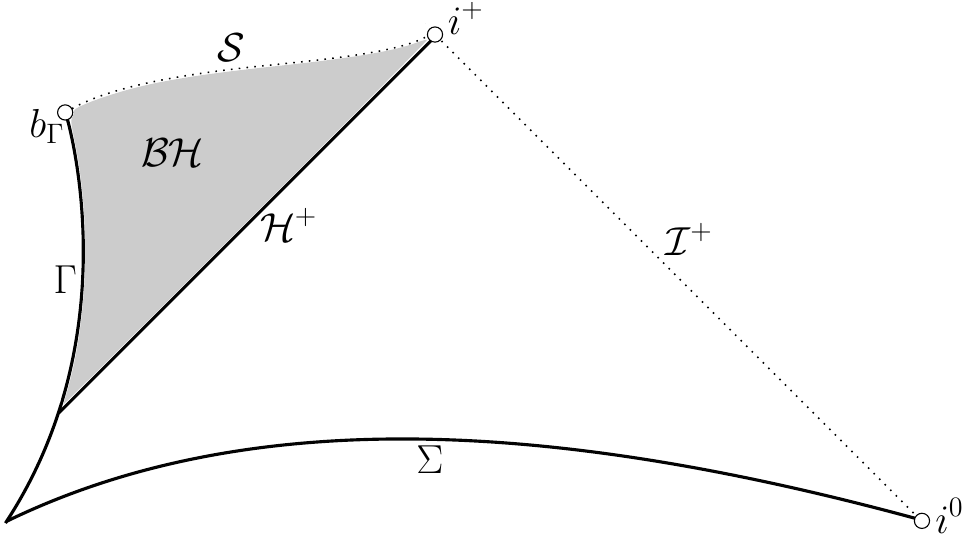}
 
\caption{The black hole case}
\end{center}
\end{figure}

The remarkable resolution of the cosmic censorship conjectures however gives very little information on the long time dynamics for these spacetimes except for the small data\footnote{i.e., when the initial data is close to that of Minkowski space.} case  \cite{Christodoulou:1993bt}. In particular, not much is known about the asymptotic decay of the scalar field as measured by a far-away observer at null infinity. In the dispersive case, Christodoulou showed that the Bondi mass decays to zero along null infinity without an explicit decay rate. In the black hole case, he showed that the Bondi mass approaches the mass of the black hole, from which one can infer the non-quantitative decay for the scalar field along null infinity \cite{Christodoulou:1993bt}.
 
The long time dynamics in the case where the spacetime settles to a black hole was subsequently studied in the seminal work\footnote{In fact, they considered the more general Einstein-Maxwell-scalar field equations.} of Dafermos-Rodnianski \cite{DR} . They proved a quantitative decay rate for the scalar field (and its derivatives) in the spacetime including along null infinity $\mathcal I^+$ and the event horizon $\mathcal H^+$. The proof is based on the local conservation of energy, which is subcritical, together with techniques exploiting the conformal geometry of the spacetime and the celebrated red-shift effect along the event horizon. The result in particular justified, in a nonlinear setting, the heuristics of Price \cite{Price}. It turns out that the quantitative decay rates, when combined with the results of \cite{D}, also have interesting consequences for the strong cosmic censorship conjecture in the context of the spherically symmetric Einstein-Maxwell-scalar field system.

In this paper, we study the other generic scenerio, i.e., spherically symmetric dispersive spacetime solutions to the Einstein-scalar field system. Unlike in the black hole case, the monotonic Hawking mass is \emph{supercritical} and provides no control over the dynamics of the solution. We thus do not expect to be able to obtain quantitative decay rates for large solutions without imposing extra assumptions. Instead, we assume \emph{a priori} the non-quantitative decay of a \emph{critical} quantity - the BV norm\footnote{Solutions of bounded variation have been first studied by Christodoulou \cite{Christodoulou:1993bt} and plays an important role in the proof of the cosmic censorship conjectures \cite{Christodoulou:1999}.} - but only locally in a region where the area of the orbit of the symmetry group $SO(3)$ remains uniformly bounded. Under this assumption of local BV scattering, we show that the scalar field and all its derivatives decay with a quantitative rate, reminescent of the Price law decay rates in the black hole case. (We refer the readers to the statement of the main theorems in Section \ref{sec.main.thm} for the precise rates that we obtain.) We prove, in particular, a quantitative decay rate for the scalar field along null infinity.

Our results apply in particular to the class of solutions arising from initial data with small BV norm. Christodoulou \cite{Christodoulou:1993bt} showed that these spacetimes are future causally geodesically complete. One can easily deduce from \cite{Christodoulou:1993bt} that in fact these spacetimes satisfy the BV scattering assumption and therefore the solutions obey the quantitative decay estimates of our main theorem (see Theorem \ref{thm:smallData}). On the other hand, our results do not require any smallness assumptions on the initial data. We conjecture that indeed our class of spacetimes contains those arising from large data:
\begin{conjecture}\label{large.sol.conj}
There exists initial data of arbitrarily large BV norm whose maximal global development scatters locally in the BV norm.
\end{conjecture}

In addition to the upper bounds that we obtain in our main theorem, we also construct examples where we prove lower bounds for the solutions with the same rates as the upper bounds. In particular, there exists a class of initial data with compactly supported scalar field whose future development saturates the decay estimates in the main theorem. This shows that the decay rates are sharp. We note that the decay rate is also consistent with the numerical study of Biz\'on-Chmaj-Rostworowski \cite{BCR}.

As a corollary of the main result on decay, we show the following dichotomy: either the quantitative decay rates are satisfied or the solution blows up at infinity. The latter are solutions such that some scale-invariant spacetime norms become infinite (see precise definition in Definition \ref{def.blow.up.infty}).

The decay result in this paper easily implies that the locally BV scattering solutions that we consider are stable against small, regular, \emph{spherically symmetric} perturbations. More ambitiously, one may conjecture that
\begin{conjecture}\label{stab.conj}
Spherically symmetric locally BV scattering dispersive solutions to the Einstein-scalar field equations are stable against \emph{non-spherically symmetric} perturbations.
\end{conjecture}

Conjecture \ref{stab.conj}, if true, will generalize the monumental theorem on the nonlinear stability of Minkowski spacetime of Christodoulou-Klainerman \cite{CK} (see also a simpler proof in \cite{LR}). For nonlinear wave equations satisfying the null condition, it is known \cite{Alinhac}, \cite{Yang} that \emph{large} solutions decaying sufficiently fast are globally stable against small perturbations. On the other hand, our main theorem shows a quantitative decay rate for spherically symmetric locally BV scattering dispersive spacetimes. Conjecture \ref{stab.conj} can therefore be viewed as an attempt to generalize the results in \cite{Alinhac}, \cite{Yang} to the Einstein-scalar field system. We will address both Conjectures \ref{large.sol.conj} and \ref{stab.conj} in future works.
\\
\\
\noindent{\bf Acknowledgements.} The authors would like to thank Mihalis Dafermos and Igor Rodnianski for valuable discussions. We also thank Jonathan Kommemi for providing the Penrose diagrams. 

J. Luk is supported by the NSF Postdoctoral Fellowship DMS-1204493. S.-J. Oh is a Miller Research Fellow, and thanks the Miller Institute at UC Berkeley for the support.

\subsection{Outline of the paper}
We outline the remainder of the paper. In Section \ref{sec.setup}, we discuss the set-up of the problem and in particular define the class of solutions considered in the main theorem, i.e., the locally BV scattering solutions. In Section \ref{sec.main.thm}, we state the main theorems in the paper (Theorems \ref{main.thm.1} and \ref{main.thm.2}), their consequences and additional theorems on the optimality of the decay rates. In the same section, we outline the main ideas of the proof. In Sections \ref{sec.anal.prop} and \ref{sec.geom}, we explain the main analytic features of the equations and the geometry of the class of spacetimes that we consider.
 
Sections \ref{sec.decay1} and \ref{sec.decay2} consist of the main content of this paper. In Section \ref{sec.decay1}, we prove the decay estimates for $\phi$, $\rd_v(r\phi)$ and $\rd_u(r\phi)$, i.e., the first main theorem (Theorem \ref{main.thm.1}). In Section \ref{sec.decay2}, using in particular the results in Section \ref{sec.decay1}, we derive the decay bounds for the second derivatives for $r\phi$ and the metric components, i.e., the second main theorem (Theorem \ref{main.thm.2}). 

In the remaining sections of the paper, we turn to other theorems stated in Section \ref{sec.main.thm}. In Section \ref{sec.dichotomy}, we give a proof of the dichotomy alluded to above, i.e., either the quantitative decay holds or the spacetime blows up at infinity. In Section \ref{sec:smallData}, we sketch a proof of a refinement of the conclusions of the main theorems in the small data case. Finally, in Section \ref{sec.opt}, we prove optimality of the decay rates asserted by the main theorems.

\section{Set-up}\label{sec.setup}
In this section, we define the set-up, formulate the equations in a double null coordinate system and explain the characteristic initial value problem. This will allow us to state the main theorem in the next section.

\subsection{Spherically Symmetric Einstein-Scalar-Field System (SSESF)} \label{subsec:derivation}
We begin with a brief discussion on the derivation of the Spherically Symmetric Einstein-Scalar-Field System \eqref{eq:SSESF} from the $(3+1)$-dimensional Einstein-scalar-field system. 

Solutions to the Einstein-scalar field equations can be represented by a triplet $(\mfd, \met_{\mu \nu},\phi)$, where $(\mfd, \met_{\mu \nu})$ is a $(3+1)$-dimensional Lorentzian manifold and $\phi$ a real-valued function on $\mfd$. The spacetime metric $\met_{\mu \nu}$ and the scalar field $\phi$ satisfy the Einstein-scalar-field system:
\begin{equation} \label{eq:ES}
\left\{
\begin{aligned}
	\Ric_{\mu \nu} - \frac{1}{2} \met_{\mu \nu} \scal =& 2 \EM_{\mu \nu}, \\
	\covD^{\mu} \rd_{\mu} \phi =& 0.
\end{aligned}
\right.
\end{equation}
where $\Ric_{\mu \nu}$ is the Ricci curvature of $\met_{\mu \nu}$, $\scal$ is the scalar curvature, and $\covD_{\mu}$ is the covariant derivative given by the Levi-Civita connection on $(\mfd, \met)$. The energy-momentum tensor $\EM_{\mu \nu}$ is given by the scalar field $\phi$, i.e.
\begin{equation}\label{eq:T} 
	\EM_{\mu\nu} = \rd_{\mu} \phi \rd_{\nu} \phi - \frac{1}{2} \met_{\mu \nu} \rd^{\lmb} \phi \rd_{\lmb} \phi.
\end{equation}

Assume that the solution $(\mfd, \met_{\mu \nu}, \phi)$ is spherically symmetric, i.e., the group $\mathrm{SO}(3)$ of three dimensional rotations acts smoothly and isometrically on $(\mfd, \met)$, where each orbit is either a point or is isometric to $\bbS^{2}$ with a round metric. The scalar field $\phi$ is required to be constant on each of the orbits. These assumptions are propagated by \eqref{eq:ES}; hence, if $(\mfd, \met_{\mu \nu}, \phi)$ is a Cauchy development, then it suffices to assume spherical symmetry only on the initial data.

The quotient $\calM / \mathrm{SO}(3)$ gives rise to a (1+1)-dimensional Lorentzian manifold with boundary, which we denote by $(\PD, g_{ab})$. The boundary $\Gmm$ consists of fixed points of the group action. We define the \emph{area radius function} $r$ on $\PD$ to be
\begin{equation*}
	 r := \sqrt{\frac{\mbox{Area of symmetry sphere}}{4 \pi}}.
\end{equation*}
and $r=0$ at $\Gmm$. Note that each component of $\Gmm$ is a timelike geodesic.

We assume that $\Gmm$ is non-empty and connected, and moreover that there exists a \emph{global double null coordinates} $(u,v)$, i.e. a coordinate system $(u,v)$ covering $\PD$ in which the metric takes the form
\begin{equation} \label{eq:defn4Met}
	g_{ab} \ud x^{a} \cdot \ud x^{b} = - \Omg^{2} \ud u \cdot \ud v
\end{equation}
for some $\Omg > 0$. We remark that both assumptions are easily justified if $(\calM, \bfg)$ is a Cauchy development of a spacelike hypersurface homeomorphic to $\bbR^{3}$.

The metric $\bfg_{\mu \nu}$ of $\calM$ is characterized by $\Omg$ and $r$ and takes the form
\begin{equation}
	\bfg_{\mu \nu} \ud x^{\mu} \cdot \ud x^{\nu} = - \Omg^{2} \ud u \cdot \ud v + r^{2} \ud s^{2}_{\bbS^{2}}\label{metric}
\end{equation}
where $\ud s^{2}_{\bbS^{2}}$ is the standard line element on the unit sphere $\bbS^{2}$. Therefore, we may reformulate the \emph{spherically symmetric Einstein-scalar-field system} \eqref{eq:SSESF} in terms of the triplet $(\phi, r, \Omg)$ as 
\begin{equation} \label{eq:SSESF} \tag{SSESF}
\left\{
\begin{aligned}
	r \rd_{u} \rd_{v} r =& - \rd_{u} r \rd_{v} r - \frac{1}{4} \Omg^{2}, \\
	r^{2} \rd_{u} \rd_{v} \log \Omg =& \, \rd_{u} r\rd_{v} r + \frac{1}{4} \Omg^{2} -  r^{2} \rd_{u} \phi \rd_{v} \phi, \\
	r \rd_{u} \rd_{v} \phi =& - \rd_{u} r \rd_{v} \phi - \rd_{v} r \rd_{u} \phi, \\
	2 \Omg^{-1} \rd_{u} r\,  \rd_{u} \Omg =& \, \rd^{2}_{u} r + r (\rd_{u} \phi)^{2}, \\
	2 \Omg^{-1} \rd_{v} r\,  \rd_{v} \Omg =& \, \rd^{2}_{v} r + r (\rd_{v} \phi)^{2},
\end{aligned}
\right.
\end{equation}
with the boundary condition $r=0$ along $\Gmm$. 

\subsection{Basic assumptions, notations and conventions}
In this subsection, we introduce the basic assumptions on the base manifold $\PD$, as well as some notations and conventions that will be used in the rest of the paper.

\subsubsection*{Definition of $\PD$ and $\mfd$}
Denote by $\bbR^{1+1}$ the (1+1)-dimensional Minkowski space, with the standard double null coordinates $(u,v)$. Let $\PD$ be a (1+1)-dimensional Lorentzian manifold which is conformally embedded into $\bbR^{1+1}$ with $\ud s^{2}_{\PD} = - \Omg^{2} \ud u \cdot \ud v$. Given a non-negative function $r$ on $\PD$, we define the set $\Gmm := \set{(u,v) \in \PD : r(u,v) = 0}$, called the \emph{axis of symmetry}. We also define $(\mfd, \met_{\mu \nu})$ to be the (1+3)-dimensional Lorentzian manifold with $\mfd = \PD \times \bbS^{2}$ and $\met_{\mu \nu}$ given by \eqref{eq:defn4Met}; this is to be thought of as the full spacetime before the symmetry reduction. (We refer to \S \ref{subsec:derivation} for the full interpretation.)

\subsubsection*{Assumptions on the conformal geometry of $\PD$}
We assume that $\Gmm \subset \PD$ is a connected set, which is the image of a future-directed timelike curve emanating from the point $(1,1)$.
We also assume that $C_{1} \subset \PD$, where
\begin{equation*}
C_{1} = \set{(u,v) \in \bbR^{1+1} : u=1, \, 1 \leq v < \infty}.
\end{equation*}
 
Furthermore, $\PD$ is assumed to be the domain of dependence of $\Gmm$ and $C_{1}$ to the future, in the sense that every causal curve in $\PD$
has its past endpoint on either $\Gmm$ or $C_{1}$. 

\subsubsection*{Notations for the conformal geometry of $\PD$}
Denote by $C_{u}$ (resp. $\uC_{v}$) the constant $u$ (resp. $v$) curve in $\PD$. Note that these are null curves in $\PD$.

Given $(u_{0}, v_{0}) \in \PD$, we define the \emph{domain of dependence} of the line segment $C_{u_{0}} \cap \set{v \leq v_{0}}$, denoted by $\calD(u_{0}, v_{0})$, to be the set of all points $p \in \PD$ such that all past-directed causal curves passing $p$ intersects $\Gmm \cup (C_{u_{0}} \cap \set{v \leq v_{0}})$, plus the line segment $(C_{u_{0}} \cap \set{v \leq v_{0}})$ itself. 

Also, we define the \emph{future null infinity} $\calI^{+}$ to be the set of ideal points $(u, +\infty)$ such that $\sup_{C_{u}} r = \infty$. 

\subsubsection*{Integration over null curves}
Whenever we integrate over a subset of $C_{u}$ (resp. $\uC_{v}$), we will use the standard line element $\ud v$ ($\ud u$) for integrals over, i.e.,  
\begin{align*}
	\int_{\uC_v\cap\{u_1\leq u\leq u_2\}} f = \int_{u_1}^{u_2} f(u',v) \, \ud u'  \\ 
	\int_{C_u\cap\{v_1\leq v\leq v_2\}} f = \int_{v_1}^{v_2} f(u,v') \, \ud v',
\end{align*} 
respectively.

\subsubsection*{Functions of bounded variation}
Unless otherwise specified, functions of bounded variation (BV) considered in this paper will be assumed to be right-continuous. By convention,
\begin{equation*}
	\rd_{v} f \, \ud v  \hbox{ or } \rd_{v} f
\end{equation*}
will refer to the distributional derivative of $f$, which is a finite signed measure, and
\begin{equation*}
	\abs{\rd_{v} f} \, \ud v  \hbox{ or } \abs{\rd_{v} f}
\end{equation*}
will denote the total variation measure. Unless otherwise specified, these measures will be evaluated on intervals of the form $(v_{1}, v_{2}]$. Thus, according to our conventions,
\begin{align*}
	\int_{v_{1}}^{v_{2}} \rd_{v} f (v) \, \ud v =& f(v_{2}) - f(v_{1}), \\
	\int_{v_{1}}^{v_{2}} \abs{\rd_{v} f (v)} \, \ud v = & \TV_{(v_{1}, v_{2}]} [f]. 
\end{align*}

\subsubsection*{New variables}
We introduce the following notation for the directional derivatives of $r$:
\begin{equation*}
	\dvr := \frac{\rd r}{\rd v}, \quad \dur := \frac{\rd r}{\rd u},
\end{equation*}

The \emph{Hawking mass} $m(u,v)$ is defined by the relation
\begin{equation}\label{mdef}
1 - \frac{2m}{r} = \rd^{a} r \rd_{a} r = - 4 \Omg^{-2} \rd_{u} r \rd_{v} r.
\end{equation}

For a solution to \eqref{eq:SSESF}, the quantity $m$ possesses useful monotonicity properties (see Lemma \ref{lem:mntn4m}), which will be one of the key ingredients of our analysis. We define the \emph{mass ratio} to be
\begin{equation*}
	\mu := \frac{2m}{r}.
\end{equation*}

We also define the \emph{Bondi mass} on $C_{u}$ by $M(u) := \lim_{v \to \infty} m(u, v)$, provided the limit exists. The Bondi mass $M_{i} := M(1) = \lim_{v \to \infty} m(1, v)$ on the initial curve $C_{1}$ is called the \emph{initial Bondi mass}.

\subsection{Refomulation in terms of the Hawking mass}
The Hawking mass as defined in \eqref{mdef} turns out to obey useful monotonicity (See \S \ref{subsec:monotonicity}). We therefore reformulate \eqref{eq:SSESF} in terms of $m$ and eliminate $\Omg$. Notice that by \eqref{metric} and \eqref{mdef}, the metric is determined by $r$ and $m$.

We say that \emph{$(\phi, r, m)$ on $\PD$ is a solution to }(SSESF) if the following equations hold:
\begin{equation} \label{eq:SSESF:dr}
\left\{
\begin{aligned}
\rd_{u} \dvr = & \frac{\mu}{(1-\mu) r} \dvr \dur, \\
\rd_{v} \dur = & \frac{\mu}{(1-\mu) r} \dvr \dur,
\end{aligned}
\right.
\end{equation}
\begin{equation} \label{eq:SSESF:dm}
\left\{
\begin{aligned}
2 \dur \rd_{u} m = &  (1-\mu) r^{2} (\rd_{u} \phi)^{2}, \\
2 \dvr \rd_{v} m = &  (1-\mu) r^{2} (\rd_{v} \phi)^{2},
\end{aligned}
\right.
\end{equation}
\begin{equation} \label{eq:SSESF:dphi}
\rd_{u} \rd_{v} (r \phi) = \frac{\mu \dvr \dur}{(1-\mu)r} \phi,
\end{equation}
and moreover, the following boundary conditions hold:
\begin{equation*}
	r = 0 \hbox{ and } m = 0 \hbox{ along } \Gmm.
\end{equation*}

We remark that using \eqref{eq:SSESF:dr}, the wave equation \eqref{eq:SSESF:dphi} for $\phi$ may be rewritten in either of the following two equivalent forms:
\begin{align} 
\rd_{u} (\rd_{v} (r \phi)) = (\rd_{u} \dvr) \phi, \label{eq:SSESF:dphi'} \tag{\ref{eq:SSESF:dphi}$'$} \\
\rd_{v} (\rd_{u} (r \phi)) = (\rd_{v} \dur) \phi. \label{eq:SSESF:dphi''} \tag{\ref{eq:SSESF:dphi}$''$}
\end{align}

\subsection{Choice of coordinates} \label{subsec:coordSys}
Note that $\PD$ is ruled by the family of null curves $C_{u}$. Since a null curve $C_{u}$ with $u \neq 1$ cannot intersect $C_{1}$, its past endpoint must be on $\Gmm$. Therefore, our assumptions so far impose the following conditions on the double null coordinates $(u,v)$ on $\PD$: $u$ is constant on each future-directed null curve emanating from $\Gmm$ and $v$ is constant on each conjugate null curve. However, these conditions are insufficient to give a unique choice of a coordinate system, as the system \eqref{eq:SSESF} and assumptions so far are invariant under the change of coordinates
\begin{equation*}
	u \mapsto U(u), \quad v \mapsto V(v), \quad U(1) = V(1) = 1,
\end{equation*}
for any strictly increasing functions $U$ and $V$. To remove this ambiguity, we fix the choice of the coordinate system, once and for all, as follows.

We first fix $v$ on $C_{1}$ relating it with the function $r$. Specifically, we will require that $v = 2r + 1$ on $C_1$, which in particular implies that 
\begin{equation} \label{eq:id4dvr}
\dvr(1, v) = \frac{1}{2}.
\end{equation}

Next, in order to fix $u$, we prescribe $u$ such that $\Gmm = \set{(u,v) : u = v}$. To do so, for every outgoing null curve $\uC$ in $\PD$, follow the incoming null curve to the past starting from $\uC\cap \Gmm$ until the point $p_*$ where it intersects the initial curve $C_1$. We then define the $u$-coordinate value for $\uC$ to be the $v$-coordinate value for $p_*$. 

Under such coordinate choice, $\calD(u_{0}, v_{0})$ may be expressed as
\begin{equation*}
	\calD(u_{0}, v_{0}) = \set{(u,v) \in \PD : u \in [u_{0}, v_{0}], v \in [u, v_{0}]}.
\end{equation*}

Moreover, if $r$ and $\phi$ are sufficiently regular functions on $\PD$, then our coordinate choice leads to $\lim_{v \to u+}(\dvr + \dur) (u,v) = \lim_{u \to v-}(\dvr + \dur) (u,v) = 0$ and $\lim_{v \to u+} (\rd_{v} + \rd_{u}) (r \phi)(u,v) = \lim_{u \to v-} (\rd_{v} + \rd_{u}) (r \phi)(u,v) = 0$. These conditions will be incorporated into precise formulations of solutions to \eqref{eq:SSESF} with limited regularity in the following subsection.

\subsection{Characteristic initial value problem}
In this paper, we study the characteristic initial value problem for \eqref{eq:SSESF} with data prescribed on $C_{1}$, under quite general assumptions on the regularity. In this subsection, we give precise formulations of initial data and solutions to \eqref{eq:SSESF} to be considered in this paper.

We begin with a discussion on the constraint imposed by \eqref{eq:SSESF} (more precisely, \eqref{eq:SSESF:dr}--\eqref{eq:SSESF:dphi}) on initial data for $(\phi, r, m)$. In fact, the constraint is very simple, thanks to the fact that they are prescribed on a characteristic (i.e., null) curve $C_{1}$. Once we prescribe $\phi$ on $C_{1}$, the coordinate condition \eqref{eq:id4dvr} dictates the initial values of $r$, and the initial values of $m$ are then determined by the constraint \eqref{eq:SSESF:dm} along the $v$ direction, as well as the boundary condition $m(1, 1) = 0$.  In other words, initial data for $(\phi, r, m)$ possess only one degree of freedom, namely the prescription of a single real-valued function $\phi(1, v)$, or equivalently, $\rd_{v} (r \phi)(1, v)$.

Following Christodoulou \cite{Christodoulou:1993bt}, we say that an initial data set for $(\phi, r, m)$ is of \emph{bounded variation} (BV) if $\rd_{v}(r \phi)(1, \cdot)$ is a (right-continuous) BV function on $[1, \infty)$ with finite total variation on $(1, \infty)$. We also define the notion of \emph{solution of bounded variation} to \eqref{eq:SSESF} as follows.

\begin{definition}[Bounded variation solutions to \eqref{eq:SSESF}] \label{def:BVsolution}
A solution $(\phi, r, m)$ to \eqref{eq:SSESF} is called a \emph{solution of bounded variation} on $\PD$ if on every compact domain of dependence $\calD(u_{0}, v_{0})$, the following conditions hold:
\begin{enumerate}
\item $\sup_{\calD(u_{0}, v_{0})} (-\dur) < \infty$ and $\sup_{\calD(u_{0}, v_{0})} \dvr^{-1} < \infty$.
\item $\dvr$ is BV on each $C_{u} \cap \calD(u_{0}, v_{0})$ uniformly in $u$, and $\dur$ is BV on each $\uC_{v} \cap \calD(u_{0}, v_{0})$ uniformly in $v$.
\item For each $a$ with $(a, a) \in \Gmm$,
\begin{equation*}
	\lim_{\eps \to 0+} (\dur + \dvr)(a, a+\eps) = 0.
\end{equation*}
\item $\phi$ is an absolutely continuous function on each $C_{u} \cap \calD(u_{0}, v_{0})$ with total variation bounded uniformly in $u$, and also an absolutely continuous function on each $\uC_{v} \cap \calD(u_{0}, v_{0})$ with total variation bounded uniformly in $v$.
\item For each $a$ with $(a, a) \in \Gmm$,
\begin{align*}
\lim_{\eps \to 0} \sup_{0 < \dlt \leq \eps} \mathrm{T.V.}_{\set{a-\dlt} \times (a-\dlt, a)} [\phi] =0, 
& \qquad \lim_{\eps \to 0} \sup_{0 < \dlt \leq \eps} \mathrm{T.V.}_{(a-\eps, a-\dlt) \times \set{a-\dlt}} [\phi] =0,  \\
\lim_{\eps \to 0} \sup_{0 < \dlt \leq \eps} \mathrm{T.V.}_{(a, a+\dlt) \times \set{a+\dlt}} [\phi] =0, 
& \qquad \lim_{\eps \to 0} \sup_{0 < \dlt \leq \eps} \mathrm{T.V.}_{\set{a+\dlt} \times (a+\dlt, a+\eps)} [\phi] =0.
\end{align*}
\item $\rd_{v}(r \phi)$ is BV on each $C_{u} \cap \calD(u_{0}, v_{0})$ uniformly in $u$, and $\rd_{u}(r \phi)$ is BV on each $\uC_{v} \cap \calD(u_{0}, v_{0})$ uniformly in $v$.
\item For each $a$ with $(a, a) \in \Gmm$,
\begin{equation*}
	\lim_{\eps \to 0+} \big( \rd_{v}(r \phi) + \rd_{u}(r \phi) \big) (a, a+\eps) = 0.
\end{equation*}
\end{enumerate}
\end{definition}

We also consider more regular data and solutions, as follows. We say that an initial data set for $(\phi, r, m)$ is $C^{1}$ if $\rd_{v}(r \phi)(1, \cdot)$ is $C^{1}$ on $[1, \infty)$ with $\sup_{C_{1}} \abs{\rd_{v}^{2}(r \phi)} < \infty$. In the following definition, we define the corresponding notion of a \emph{$C^{1}$ solution} to \eqref{eq:SSESF}.

\begin{definition}[$C^{1}$ solutions to \eqref{eq:SSESF}] \label{def:C1solution}
A solution $(\phi, r, m)$ to \eqref{eq:SSESF} is called a \emph{$C^{1}$ solution} on $\PD$ if the following conditions hold on every compact domain of dependence $\calD(u_{0}, v_{0})$:
\begin{enumerate}
\item $\sup_{\calD(u_{0}, v_{0})} (-\dur) < \infty$ and $\sup_{\calD(u_{0}, v_{0})} \dvr^{-1} < \infty$.
\item $\dvr$, $\dur$ are $C^{1}$ on $\calD(u_{0}, v_{0})$.
\item For each $a$ with $(a, a) \in \Gmm$,
\begin{equation*}
	\lim_{\eps \to 0+} (\dur + \dvr)(a, a+\eps) = \lim_{\eps \to 0+} (\dur + \dvr)(a-\eps, a) = 0.
\end{equation*}
\item $\rd_{v}(r \phi)$ and $\rd_{u} (r \phi)$ are $C^{1}$ on $\calD(u_{0}, v_{0})$.
\item For each $a$ with $(a, a) \in \Gmm$,
\begin{equation*}
	\lim_{\eps \to 0+} \big( \rd_{v}(r \phi) + \rd_{u}(r \phi) \big) (a, a+\eps) 
	= \lim_{\eps \to 0+} \big( \rd_{v} (r \phi) + \rd_{v}(r \phi) \big) (a - \eps, a)
	= 0.
\end{equation*}
\end{enumerate}
\end{definition}

\begin{remark} \label{rem:wp}
By \cite[Theorem 6.3]{Christodoulou:1993bt}, a BV initial data set leads to a unique BV solution to \eqref{eq:SSESF} on $\set{(u,v) : 1 \leq u \leq 1+\dlt, v \geq u}$ for some $\dlt > 0$.
If the initial data set is furthermore $C^{1}$, then it is not difficult to see that the corresponding solution is also $C^{1}$ (persistence of regularity). In fact, this statement follows from the arguments in Section \ref{sec.decay2} of this paper; see, in particular, the proof of Lemma \ref{lem:decay2:key4nullStr}.
\end{remark}

\subsection{Local scattering in BV and asymptotic flatness}
We are now ready to formulate the precise notion of \emph{locally BV scattering solutions} to \eqref{eq:SSESF}, which is the class of solutions that we consider in this paper. In particular, for this class of solutions, we make a priori assumptions on its global geometry. 

\begin{definition}[Local scattering in BV] \label{def:locBVScat}
We say that a BV solution $(\phi, r, m)$ to \eqref{eq:SSESF} is \emph{locally scattering in the bounded variation norm (BV)}, or a \emph{locally BV scattering solution}, if the following conditions hold:
\begin{enumerate}
\item \emph{Future completeness of radial null geodesics}: Every incoming null geodesic in $\PD$ has its future endpoint on $\Gmm$, and every outgoing null geodesic in $\PD$ is infinite towards the future  in the affine parameter. 
Moreover, there exists a global system of null coordinates $(u,v)$ and $\PD$ is given by
\begin{equation} \label{eq:globalCoords}
	\PD = \set{(u,v) : u \in [1, \infty), v \in [u, \infty)}.
\end{equation}

\item \emph{Vanishing final Bondi mass}: The final Bondi mass vanishes, i.e.,
\begin{equation} \label{eq:zeroMf}
M_{f} := \lim_{u \to \infty} M(u) = 0.
\end{equation}

\item \emph{Scattering in BV in a compact $r$-region}: There exists $R > 0$ such that for $\cpt$ defined to be the region $\set{(u,v) \in \PD : r(u,v) \leq R}$, we have 
\begin{equation} \label{eq:locBVScat}
	\int_{C_{u} \cap \cpt} \abs{\rd_{v}^{2} (r \phi)} \to 0, \quad
	\int_{C_{u} \cap \cpt} \abs{\rd_{v} \log \dvr} \to 0
\end{equation}
as $u \to \infty$.
\end{enumerate}
\end{definition}

Several remarks concerning Definition~\ref{def:locBVScat} are in order.
\begin{remark} 
In fact, the condition \eqref{eq:globalCoords} is a consequence of future completeness of radial null geodesics and the preceding assumptions. To see this, first recall our assumption that $C_{1} = \set{(u,v) : u = 1, v \in [1, \infty)}$. Hence from our choice of the coordinate $u$ and future completeness of incoming radial null geodesics, it follows that the range of $u$ must be $[1, \infty)$. Furthermore, for each $u \in [1, \infty)$, the range of $v$ on $C_{u}$ is $[u, \infty)$ by future completeness of outgoing radial null geodesics and Definition~\ref{def:BVsolution}. More precisely, future completeness of $C_{u}$ implies that it can be continued past $\set{u} \times [u, v_{0}]$ as long as $\int_{u}^{v_{0}} \Omg^{2} \, \ud v < \infty$, and Definition~\ref{def:BVsolution} implies\footnote{We refer to the proof of Proposition~\ref{prop:geomLocBVScat} below for details of estimating $\frac{-\dur}{1-\mu}$ in terms of assumptions on $\phi$, $\rd_{v}(r \phi)$ and $\dvr$.} that $\Omg^{2} = - \frac{4 \lmb \nu}{1-\mu}$ indeed remains bounded on $\set{u} \times [u, v_{0}]$ for every finite $v_{0}$.
\end{remark}

\begin{remark}  \label{rem:FCGC}
For more regular (e.g., $C^{1}$) asymptotically flat solutions, the conditions $(1)$ and $(2)$ in Definition \ref{def:locBVScat} may be replaced by a single equivalent condition, namely requiring the full spacetime $(\mfd, \met)$ to be \emph{future casually geodesically complete} as a (1+3)-dimensional Lorentzian manifold. In particular, (2) follows from the deep work \cite{Christodoulou:1987ta} of Christodoulou, in which it was proved that if $M_{f} > 0$ then the space-time necessarily contains a black-hole and thus is not future causally geodesically complete.
\end{remark}

\begin{remark}\label{rmk.unif.int}
As we will see in the proof, there exists a universal $\tilde{\ep}_0$ such that (3) in Definition \ref{def:locBVScat} can be replaced by the weaker requirement that there exists $R>0$ and $U>0$ such that
\begin{equation*}
	\int_{C_{u} \cap \cpt} \abs{\rd_{v}^{2} (r \phi)} \leq \tilde{\ep}_0, \quad
	\int_{C_{u} \cap \cpt} \abs{\rd_{v} \log \dvr} \leq \tilde{\ep}_0
\end{equation*}
for $u\geq U$. To simplify the exposition, we will omit the proof of this improvement.
\end{remark}

\begin{remark} 
For a sufficiently regular, asymptotically flat solution to \eqref{eq:SSESF}, Condition $(1)$ in Definition \ref{def:locBVScat} is equivalent to requiring that the conformal compactification of $\PD$ is depicted by a Penrose diagram as in Figure \ref{fig.disp} (in the introduction). For more discussion on Penrose diagrams, we refer the reader to \cite[Appendix C]{DR} and \cite{Kommemi}. In fact, this equivalence follows easily from the classification of all possible Penrose diagrams for the system \eqref{eq:SSESF} given in the latter reference.
\end{remark}

We also define the precise notion of \emph{asymptotic flatness} for initial data with BV or $C^{1}$ regularity. As we shall see soon in the main theorems, the rate of decay for the initial data, measured in $r$, is directly related to the rate of decay of the corresponding solution, in both $u$ and $r$.

\begin{definition}[Asymptotic flatness of order $\omg'$ in BV or $C^{1}$] \label{def:AF}
For $\omg' >1$, we make the following definition.

\begin{enumerate}
\item We say that an initial data set is \emph{asymptotically flat of order $\omg'$ in BV} if $\rd_{v} (r \phi)(1, \cdot) \in \mathrm{BV}[1, \infty)$ and there exists $\calI_{1} > 0$ such that 
\begin{equation}
	\sup_{C_{1}} (1+r)^{\omg'} \abs{\rd_{v}(r \phi)} \leq \calI_{1} < \infty.
\end{equation}

\item We say that an initial data set is \emph{asymptotically flat of order $\omg'$ in $C^{1}$} if $\rd_{v} (r \phi)(1, \cdot) \in C^{1}[1, \infty)$ and there exist $\calI_{2} > 0$ such that 
\begin{equation}
	\sup_{C_{1}} (1+r)^{\omg'} \abs{\rd_{v}(r \phi)} + \sup_{C_{1}} (1+r)^{\omg'+1} \abs{\rd_{v}^{2}(r \phi)} \leq \calI_{2} < \infty.
\end{equation}
\end{enumerate}
\end{definition}

\begin{remark} 
The initial Bondi mass $M_{i} := \lim_{v \to \infty} m(1, v)$ can be easily bounded by $\leq C \calI_{1}^{2}$; see Lemma \ref{lem:bnd4Mi}. 
\end{remark}

\begin{remark} \label{rem:PhiIsZero}
Observe that both conditions imply that $(r \phi)(1,v)$ tends to a finite limit as $v \to \infty$; in particular, $\lim_{v \to \infty} \phi(1, v) = 0.$
This serves to fix the gauge freedom $(\phi, r, m) \mapsto (\phi + c, r, m)$ for solutions to \eqref{eq:SSESF}.
\end{remark}

\section{Main results}\label{sec.main.thm}
\subsection{Main theorems}
With the definitions of locally BV scattering solutions and asymptotically flat initial data, we now have the necessary means to state the main theorems of this paper. Roughly speaking, these theorems say that locally BV scattering solutions with asymptotically flat initial data exhibits quantitative decay rates, which can be read off from the rate $\omg'$ in Definition \ref{def:AF}. The first theorem is for initial data and solutions in BV.

\begin{theorem}[Main Theorem in BV] \label{main.thm.1}
Let $(\phi, r, m)$ be a locally BV scattering solution to \eqref{eq:SSESF} with asymptotically flat initial data of order $\omg'$ in BV. Then for $\omg := \min \set{\omg', 3}$, there exists a constant $A_{1} > 0$ such that 
\begin{align} 
	\abs{\phi} \leq & A_{1} \min \set{u^{-\omg}, r^{-1} u^{-(\omg-1)}}, \label{eq:decay1:1} \\
	\abs{\rd_{v}(r \phi)} \leq & A_{1} \min \set{u^{-\omg}, r^{-\omg}}, \label{eq:decay1:2} \\
	\abs{\rd_{u} (r \phi)} \leq & A_{1} u^{-\omg}. \label{eq:decay1:3}
\end{align}
\end{theorem}

The second theorem is for initial data and solutions in $C^{1}$. 

\begin{theorem}[Main Theorem in $C^{1}$] \label{main.thm.2}
Let $(\phi, r, m)$ be a locally BV scattering solution to \eqref{eq:SSESF} with asymptotically flat initial data of order $\omg'$ in $C^{1}$. Then, in addition to the bounds \eqref{eq:decay1:1}-\eqref{eq:decay1:3}, there exists a constant $A_{2} > 0$ such that 
\begin{align} 
	\abs{\rd_{v}^{2} (r \phi)} \leq & A_{2} \min \set{u^{-(\omg+1)}, r^{-(\omg+1)}}, \label{eq:decay2:1} \\
	\abs{\rd_{u}^{2} (r \phi)} \leq & A_{2} u^{-(\omg+1)}, \label{eq:decay2:2} \\
	\abs{\rd_{v} \dvr} \leq & A_{2} \min \set{u^{-3}, r^{-3}}, \label{eq:decay2:3} \\
	\abs{\rd_{u} \dur} \leq & A_{2} u^{-3}. \label{eq:decay2:4}
\end{align}
for $\omg := \min \set{\omg', 3}$.
\end{theorem}

Some remarks regarding the main theorems are in order.

\begin{remark}
Notice that in Theorem \ref{main.thm.2}, the decay rates for $\rd_v\lambda$ and $\rd_u\nu$ are independent of the order $\om'$ of asymptotic flatness of the initial data. This is because the scalar field terms enter the equations for $\rd_u\rd_v\log\lambda$ and $\rd_v\rd_u\log\nu$ quadratically (see equations \eqref{eq:eq4dvdvr:normal} and \eqref{eq:eq4dudur:normal}) and thus as long as $\om'>1$, their contributions to the decay rates of $\rd_v\lambda$ and $\rd_u\nu$ are lower order compared to the term involving the Hawking mass.
\end{remark}

\begin{remark}
By Remark \ref{rem:wp}, a $C^{1}$ initial data set gives rise to a $C^{1}$ solution. Hence Remark \ref{rem:FCGC} applies, and the conditions (1)--(2) of Definition \ref{def:locBVScat} may be replaced by a single equivalent condition of \emph{future causal geodesic completeness} of $(\mfd, \bfg)$ in the case of Theorem \ref{main.thm.2}. 
\end{remark}

\begin{remark}
In general, the constants $A_1$ and $A_2$ depend not only on the size of the initial data (e.g., $\calI_{1}$, $\calI_{2}$), but rather on the full profile of the solution. Nevertheless, for the special case of small initial total variation of $\rd_{v}(r \phi)$, $A_{1}$ and $A_{2}$ \emph{do} depend only on the size of the initial data; see \S \ref{sec:mainThm:smallData}.
\end{remark}

\begin{remark}
If the initial data also verify higher derivative estimates, then the techniques in proving Theorems \ref{main.thm.1} and \ref{main.thm.2} also allow us to derive decay bounds for higher order derivatives. The proof of the higher derivative decay estimates is in fact easier than the proofs of the first and second derivative decay bounds since we have already obtained sufficiently strong control of the scalar field and the geometry of the spacetime. We will omit the details.
\end{remark}

\begin{remark}
The decay rates that we obtain in these variables imply as immediate corollaries decay rates for $\rd_v \phi$, $\rd_u \phi$, etc. See Corollaries \ref{cor:decay1} and \ref{cor:decay2}.
\end{remark}

\begin{remark} \label{rem:coord}
The decay rates in the main theorems are measured with respect to the double null coordinates $(u,v)$ normalized at the initial curve and the axis $\Gmm$ as in \S \ref{subsec:coordSys}. To measure the decay rate along null infinity, one can alternatively normalize the $u$ coordinate\footnote{In particular, this normalization is used in \cite{DR} for the black hole case. By changing the null coordinates, we can thus more easily compare the decay rates in our setting and that in \cite{DR}.} by requiring $\rd_{u} r=-\frac 12$ at future null infinity. As we will show in \S \ref{sec.coord}, for the class of spacetimes considered in this paper, the decay rates with respect to this new system of null coordinates are the same up to a constant multiplicative factor.
\end{remark}

\begin{remark}
In view of Remark \ref{rmk.unif.int}, the assumption of local BV scattering can be replaced by the \emph{boundedness} of the \emph{subcritical} quantities
$$\int_{C_{u} \cap \cpt} \abs{\rd_{v}^{2} (r \phi)}^p \leq C, \quad
	\int_{C_{u} \cap \cpt} \abs{\rd_{v} \log \dvr}^p \leq C,\quad\mbox{for }p>1.$$
This is because for fixed $\tilde{\ep}_0$, one can choose $R$ to be sufficiently small (depending on $C$) and apply H\"older's inequality to ensure	that
$$\int_{C_{u} \cap \cpt} \abs{\rd_{v}^{2} (r \phi)} \leq \tilde{\ep}_0, \quad
	\int_{C_{u} \cap \cpt} \abs{\rd_{v} \log \dvr} \leq \tilde{\ep}_0.$$
\end{remark}

\begin{remark}
We also notice that the proof of our main theorem can be carried out in an identical manner for locally BV scattering solutions arising from asymptotically flat \emph{Cauchy data}. More precisely, we can consider initial data given on a Cauchy hypersurface
$$v=f(u),\quad\mbox{with }C^{-1}\leq -f'(u)\leq C$$
and satisfying the constraint equation together with the following bounds on the initial data:
$$(1+r)^{\om'}(|\phi|+|\rd_v(r\phi)|+|\rd_u(r\phi)|+|\lambda-\frac 12|+|\nu+\frac 12|)\leq \tilde{\mathcal I}_1,$$
and
$$(1+r)^{\om'+1}(|\rd_v^2(r\phi)|+|\rd_u^2(r\phi)|+|\rd_v \log\lambda|+|\rd_u\log\nu|)\leq \tilde{\mathcal I}_2.$$ 
Then, if we assume in addition that the spacetime is locally BV scattering to the future, the conclusions of Theorems \ref{main.thm.1} and \ref{main.thm.2} hold.
\end{remark}

\begin{remark}
Our main theorems can also be viewed as results on upgrading qualitative decay to quantitative decay estimates. Such problems have been widely studied in the \emph{linear} setting (without the assumption on spherical symmetry) on nontrapping asymptotically flat Lorentzian manifolds \cite{DR2}, \cite{Ta}, \cite{MTT}, as well as for the obstacle problem on Minkowski space \cite{Morawetz}, \cite{Strauss}. In the \emph{nonlinear} setting, we mention the work of Christodoulou-Tahvildar--Zadeh \cite{CTZ}, who studied the energy critical 2-dimensional spherically symmetric wave map system and proved asymptotic decay for the solution and its derivatives.
\end{remark}

\subsection{BV scattering and the blow-up at infinity scenerio}

The condition of local BV scattering in the main theorems follows if one rules out, a priori, a blow-up at infinity scenario. More precisely, we say that a solution blows up at infinity if some scale-invariant spacetime norms are infinite as follows:

\begin{definition}\label{def.blow.up.infty}
Let $(\phi, r, m)$ be a BV solution to \eqref{eq:SSESF} such that the condition $(1)$ of Definition \ref{def:locBVScat} (future completeness of radial null geodesics) holds. We say that the solution \emph{blows up at infinity} if at least one of the following holds:
\begin{enumerate}
\item $\displaystyle{\sup \lambda_{\Gmm}^{-1} = \infty}$, where $\displaystyle{\lmb_{\Gmm}(u) := \lim_{v \to u+} \dvr(u,v)}$.
\item $\displaystyle{\int_1^{\infty}\int_u^{\infty}|\rd_v\lambda\rd_u\phi-\rd_u\lambda\rd_v\phi| \ud v \ud u =\infty}$,
\item $\displaystyle{\int_1^{\infty}\int_u^{\infty}|\rd_u\phi\rd_v(\nu^{-1}\rd_u(r\phi))-\rd_v\phi\rd_u(\nu^{-1}\rd_u(r\phi))| \ud v \ud u =\infty}$.
\end{enumerate}
\end{definition} 

\begin{remark}
We do not prove in the paper the existence or non-existence of solutions that blow up at infinity. We remark that this is analogous to the blow-up at infinity scenarios which have recently been constructed in some simpler \emph{semilinear}, \emph{critical} wave equations \cite{DK}.
\end{remark}

It follows from our main theorem that if a solution does not blow up at infinity, it obeys quantitative decay estimates. More precisely, we have
\begin{theorem}[Dichotomy between blow-up at infinity and BV scattering] \label{thm.dichotomy}
Let $(\phi, r, m)$ be a BV solution to \eqref{eq:SSESF} such that the condition $(1)$ of Definition \ref{def:locBVScat} (future completeness of radial null geodesics) holds. Assume furthermore that the initial data for $(\phi, r, m)$ obey the condition\footnote{By Remark \ref{rem:PhiIsZero}, note that this is the only condition on $\lim_{v \to \infty} \phi(1, v)$ which is consistent with asymptotic flatness.} $\lim_{v \to \infty} \phi(1,v) = 0$ and
\begin{equation} \label{dichotomy.hyp}
\int_{C_{1}} \abs{\rd_{v}^{2} (r \phi)} \, \ud v + \sup_{C_{1}} \abs{\rd_{v}(r \phi)} < \infty.
\end{equation}

Then either
\begin{enumerate}
\item the solution blows up at infinity; or
\item the solution is globally BV scattering, in the sense that the conditions $(2)$ and $(3)$ of Definition \ref{def:locBVScat} hold with $R =\infty$.
\end{enumerate}
\end{theorem}

This theorem is established in Section \ref{sec.dichotomy}. It then follows from our main theorems (Theorems \ref{main.thm.1} and \ref{main.thm.2}) that if a BV solution does not blow up at infinity and possesses asymptotically flat initial data, then it obeys quantitative decay estimates.

\subsection{Refinement in the small data in BV case} \label{sec:mainThm:smallData}
By a theorem of Christodoulou \cite{Christodoulou:1993bt}, the maximal development of data with small BV norms does not blow up at infinity. The previous theorem applies, and thus the corresponding solution is globally BV scattering, in the sense described in Theorem \ref{thm.dichotomy}. Moreover, a closer inspection of the proof of the main theorems reveals that the following stronger conclusion holds in this case.

\begin{theorem} [Sharp decay for data with small BV norm] \label{thm:smallData}
There exists a universal $\eps_{0} > 0$ such that for $0 < \eps \leq \eps_{0}$, the following statements hold.

\begin{enumerate}
\item If the initial data set is asymptotically flat of order $\omg'$ in BV and
\begin{equation*}
	\int_{C_{1}} \abs{\rd_{v}^{2} (r \phi)} < \eps,
\end{equation*}
then the maximal development $(\phi, r, m)$ is globally BV scattering, in the sense that Definition \ref{def:locBVScat} holds with arbitrarily large $R > 0$. Moreover, it satisfies the estimates \eqref{eq:decay1:1}--\eqref{eq:decay1:3} with $A_{1} \leq C_{\calI_{1}} (\calI_{1} + \eps)$.

\noindent Here (and similarly in (2)), we use the convention that $C_{\calI_{1}}$ depends on $\calI_1$ in a non-decreasing fashion.\footnote{In particular, for $\calI_1$ sufficiently small, we have the estimate $A_1\leq C(\mathcal I_1+\ep)$ for some absolute constant $C$.}

\item If, in addition, the initial data set is asymptotically flat of order $\omg'$ in $C^{1}$, then the maximal development also satisfies \eqref{eq:decay2:1}--\eqref{eq:decay2:4} with $A_{2} \leq C_{\calI_{2}} (\calI_{2} + \eps)$.
\end{enumerate}
\end{theorem}

The point of this theorem is that we only need to know that the initial total variation to be small in order to conclude pointwise decay rates; in particular, $\calI_{1}$, $\calI_{2}$ can be arbitrarily large. In this sense, Theorem \ref{thm:smallData} generalizes both the small BV global well-posedness theorem \cite[Theorem 6.2]{Christodoulou:1993bt} and the earlier small data scattering theorem \cite{Christodoulou:1986ue} for data that are small in a weighted $C^1$ norm. A proof of this theorem will be sketched in Section \ref{sec:smallData}.

\subsection{Optimality of the decay rates}

Our main theorems show upper bounds for the decay rates of the scalar field $\phi$ and its derivatives both towards null infinity (i.e., in $r$) and along null infinity (i.e., in $u$). For $\omega' = \omg<3$, if the decay rate of the initial data towards null infinity satisfies also a lower bound, then we can show that both the $r$ and $u$ decay rates in Theorem \ref{main.thm.1} are saturated. More precisely,

\begin{theorem} [Sharpness of $t^{-\omg}$ tail for $1 < \omg < 3$] \label{thm.opt.1}
Let $1 < \omg < 3$. Suppose, in addition to the assumptions of Theorem \ref{main.thm.1}, that there exists $V\geq 1$ such that the initial data set satisfies the lower bound
\begin{equation*}
	r^{\omg} \rd_{v}(r \phi) (1, v) \geq L > 0,
\end{equation*}
for $v\geq V$.

Then there exists a constant $L_{\omg} >0$ such that
\begin{align*}
	\rd_{v} (r \phi)(u, v) \geq & L_{\omg} \min\set{r^{-\omg}, u^{-\omg}}, \\
	- \rd_{u} (r \phi)(u, v) \geq & L_{\omg} u^{-\omg},
\end{align*}
for $u$ sufficiently large.
\end{theorem}
\begin{remark}
One can also infer the sharpness of the decay of $\phi$ from that of its derivatives. We will omit the details.
\end{remark}

This theorem will be proved in \S \ref{subsec.opt.1}. In fact, the proof of this theorem is similar to the proof of the upper bounds in the first main theorem (Theorem \ref{main.thm.1}). We will show that after restricting to $u$ sufficiently large, the initial lower bound propagates and the nonlinear terms only give lower order contributions. Notice also that the analogous statement is false for $\omega'\geq 3$, since the nonlinear terms may dominate the contribution of the initial data.

For $\omega' \geq 3$, we can show that the decay rates in Theorem \ref{main.thm.1} are sharp in the following sense:
\begin{theorem} [Sharpness of $t^{-3}$ tail] \label{thm.opt.2}
For arbitrarily small $\eps > 0$, there exists a locally BV scattering solution $(\phi, r, m)$ to \eqref{eq:SSESF} which satisfies the following properties:
\begin{enumerate}
\item $\rd_{v} (r \phi)(1, v)$ is smooth, compactly supported in the $v$-variable and has total variation less than $\eps$, i.e.,
\begin{equation*}
	\int_{C_{1}} \abs{\rd_{v}^{2} (r \phi)} < \eps.
\end{equation*}
\item There exists a constant $L_{3} > 0$ such that
\begin{align*}
	\rd_{v} (r \phi) (u, v) \geq & L_{3} \min \set{r^{-3}, u^{-3}}, \\
	- \rd_{u} (r \phi) (u, v) \geq & L_{3} u^{-3},
\end{align*}
for $u$ sufficiently large.
\end{enumerate}
\end{theorem}

To prove Theorem \ref{thm.opt.2}, we will first establish a sufficient condition for the desired lower bounds in terms of (non-vanishing of) a single real number $\mathfrak{L}$, which is computed from information at the null infinity. This result (Lemma \ref{lem:LB}) is proved using the decay rates proved in the main theorems, and we believe it might be of independent interest. In \S \ref{subsec.opt.2}, we will complete the proof of Theorem \ref{thm.opt.2} by constructing an initial data set for which $\mathfrak{L}$ can be bounded away from zero. This can be achieved by showing that the solution is close to that of a corresponding linear problem and controlling the error terms after taking $\eps > 0$ to be sufficiently small and using Theorem \ref{thm:smallData}.

\subsection{Strategy of the proof of the main theorems}

Roughly speaking, the proof of decay of $\phi$ and its derivatives can be split into three steps. In the first two steps, we control the incoming part\footnote{We call these variables `incoming' because they obey a transport equation in the $\rd_u$ direction.} of the derivatives of the scalar field and metric components, i.e., $\rd_v(r\phi)$, $\rd_v^2(r\phi)$ and $\rd_v\dvr$. To this end, we split the spacetime into the exterior region $\extr:=\{(u,v)\in \mathcal Q: v\geq 3u\}$ and the interior region $\intr:=\{(u,v)\in \mathcal Q: v\leq 3u\}$. In the first step, we control the incoming part of the solution in the exterior region. In this region, we have $r \gtrsim v, u$, thus the negative $r$ weights in the equations give the required decay of $\phi$ and its derivatives. We then prove bounds in the interior region in the second step. Here, we exploit certain (non-quantitative) smallness in the spacetimes quantities as $u \to \infty$ given by the assumption of local BV scattering to propagate the decay estimates from the exterior region to the interior region all the way up to the axis. Finally, in the third step, we control the outgoing part of the solution, i.e., $\rd_u(r\phi)$, $\rd_u^2(r\phi)$ and $\rd_u\dur$, by showing that the decay bounds that we have proved along the axis can be propagated in the outgoing direction.

We remind the readers that the above sketch is only a heuristic argument and is not true if taken literally. In particular, in order to carry out this procedue, we need to first show that the local BV scattering assumption provides some control over the spacetime geometry. As we will show below, the estimates are derived in slightly different fashions for the first and the second derivatives of $r\phi$. We note in particular that carrying out this general scheme relies heavily on the analytic structure of the Einstein-scalar field equations, including the montonicity properties as well as the null structure of the (renormalized) equations.

\subsubsection{Estimates for first derivatives of $r\phi$}

To obtain decay bounds for the first derivatives of $r\phi$, we will rely on the wave equation
$$\rd_u \rd_v(r\phi)=\frac{2 m \lambda \nu}{(1-\mu)r^2}\phi.$$
Notice that when we solve for the incoming radiation $\rd_v(r\phi)$ using this as a transport equation in $u$, the right hand side does not depend explicitly on the outgoing radiation $\rd_u(r\phi)$. Instead, the right hand side consists of terms that are either lower order (in terms of derivatives) or satisfy a certain monotonicity property.

In particular, this equations shows that as long as $\phi$ can be controlled, we can estimate $\rd_v(r\phi)$ by integrating along the incoming $u$ direction. On the other hand, we can also control $\phi$ once a bound on $\rd_v(r\phi)$ is known by integrating along the outgoing $v$ direction.

To achieve the desired decay rates for $\phi$, $\rd_v(r\phi)$ and $\rd_u(r\phi)$, we follow the three steps outlined above:

\begin{enumerate}
\item [(1)] Bounds\footnote{The estimates in this region are similar to the corresponding bounds for the black hole case in \cite{DR}. There, it was observed that the quantity $\partial_v(r \phi)$, which Dafermos-Rodnianski called an almost Riemann invariant, verifies an equation such that the right hand side has useful weights in $r$ and give the desired decay rates.} for $\rd_v(r\phi)$ and $\phi$ in $v\geq 3u$: In the exterior region, we have $r\gtrsim u,v$, it is therefore sufficient to prove the decay in $r$. First, we  prove that $\sup_{C_u}(1+r)\phi$ is bounded. This is achieved in a compact region by continuity of the solution\footnote{In particular, since we are simply using compactness, the constants in Theorem \ref{main.thm.1} depend not only on the size of the initial data.} and in the region of large $r$ by integrating $\rd_v(r\phi)$ in the outgoing direction from the compact region. Since $\rd_v(r\phi)$ can in turn be controlled by $\phi$, we get the desired bound. To improve over this bound we define
\begin{equation*}
\calB_{1}(U) := \sup_{u \in [1, U]} \sup_{C_{u}} \bb( u^{\om} \abs{\phi} + r u^{\om-1} \abs{\phi} \bb)
\end{equation*}
and show via the wave equation that
$$r^\om|\rd_v(r\phi)|\leq C(u_1)+\ep(u_1)\mathcal{B}_1(U),$$
where $\ep\to 0$ as $u_1\to \infty$. This gives the optimal decay rate for $\rd_v(r\phi)$ in the exterior region up to an arbitrarily small loss, which can be estimated once $\mathcal{B}_1(U)$ can be controlled.

\item [(2)] Bounds for $\rd_v(r\phi)$ and $\phi$ in $v\leq 3u$: For the decay of the first derivatives, the interior region $\{v\leq 3u\}$ is further divided into the intermediate region $\{r\geq R\}$ and the compact region $\{r\leq R\}$. In these two regions, the $r$ weight in the equation is not sufficient to give the sharp decay rate. Instead, we start from the decay rate $\rd_v(r\phi)$ obtained in the first step in the exterior region and propagate this decay estimate inwards. To achieve this, we need to show that $\int \frac{2 m \lambda \nu}{(1-\mu)r^2}$ is small when $u$ is sufficiently large. 

\item [(2a)] $r\geq R$ and $v\leq 3u$: In the intermediate region where we still have a lower bound on $r$, the required smallness is given by the \emph{qualitative} information that the Hawking mass approaches $0$. Thus, from some large time onwards, $\int \frac{2 m \lambda \nu}{(1-\mu)r^2}$ becomes sufficiently small and we can integrate the wave equation directly to obtain the desired decay bounds.

\item [(2b)] $r\leq R$ and $v\leq 3u$: In this region, we use the local BV scattering assumption to show that $\int_{\{r\leq R\}} \frac{2 m \lambda \nu}{(1-\mu)r^2}\to 0$ as $u\to\infty$. This smallness allows us to propagate the decay estimates from the curve $r=R$ to the region $r< R$. At this point, we can also recover the control for $\mathcal{B}_1(U)$ and close the estimates in step 1. This allows us to derive all the optimal decay rates for $\phi$ and $\rd_v(r\phi)$

\item [(3)] Bounds for $\rd_u(r\phi)$: To achieve the bounds for $\rd_u(r\phi)$, first note that along the axis we have $\rd_u(r\phi)=-\rd_v(r\phi)$. Thus, by the previous derived control for $\rd_v(r\phi)$, we also have the decay of $\rd_u(r\phi)$ along the axis. We then consider the wave equation as a transport equation in the outgoing direction for $\rd_u(r\phi)$ to obtain the sharp decay for $\rd_u(r\phi)$ in the whole spacetime. 

\end{enumerate}

\subsubsection{Estimates for second derivatives of $r\phi$}

As for the first derivatives, we control the second derivatives by first integrating the equation in the exterior region up to a curve $v=3u$. We then propagate the decay bounds from the exterior region to the interior region using the estimates already derived for the first derivative of $\phi$, as well as the local BV scattering assumption. However, at this level of derivatives, some new difficulties arise as we now describe.
\\
\fparagraph{Renormalization and the null structure}
The assumption of local BV scattering implies that
\bea
\int_{C_u\cap\{r\leq R\}} (|\rd_v\phi|+|\rd_v^2(r\phi)|)\to 0 \label{BV.small.1}
\eea
as $u\to \infty$. When combined with Christodoulou's BV theory, this also implies that as $v\to \infty$, we have
\bea
\int_{\uC_v\cap\{r\leq R\}} (|\rd_u\phi|+|\rd_u^2(r\phi)|) \to 0.  \label{BV.small.2}
\eea
Notice that on $C_u$ (resp. $\uC_v$), we only control the integral of $\rd_v^2(r\phi)$ and $\rd_v\phi$ (resp. $\rd_u^2(r\phi)$ and $\rd_u\phi$).

Suppose when integrating along the incoming direction to control $\rd_v^2(r\phi)$ and $\rd_v\dvr$, we need to estimate terms of the form
$$\int_{\uC_v\cap\{r\leq R\}} |\rd_u\phi \rd_v\phi|.$$
We can apply the BV theory to show that for $v$ sufficiently large,
$$\int_{\uC_v\cap\{r\leq R\}} |\rd_u\phi| \leq \epsilon.$$
On the other hand, one can show that
$$\sup_{\uC_v\cap\{r\leq R\}} |\rd_v\phi|\leq C \sup_{J^-(\uC_v\cap\cpt)} |\rd_v^2(r\phi)|$$
which can be controlled by the quantity that we are estimating.

However, in equation \eqref{eq:eq4dvdvrphi:normal} for $\rd_v^2(r\phi)$ derived by differentiating \eqref{eq:SSESF:dphi}, there are terms of the form
$$\rd_v\phi \rd_v\phi$$
such that neither of the factors can be controlled a priori in $L^1$ by the local BV scattering assumption. In other words, the equation does not obey any null condition.

To deal with this problem, we follow \cite{Christodoulou:1993bt} and introduce the renormalized variables
$ \rd_{v}^{2} (r \phi) - (\rd_{v} \dvr) \phi, $
	$ \rd_{u}^{2} (r \phi) - (\rd_{u} \dur) \phi, $
	$ \rd_{v} \log \dvr - \frac{\dvr}{(1-\mu)} \frac{\mu}{r} + \rd_{v} \phi \bb( \dvr^{-1} \rd_{v} (r \phi) - \dur^{-1} \rd_{u} ( r \phi) \bb), $
	$ \rd_{u} \log (-\dur) - \frac{\dur}{(1-\mu)} \frac{\mu}{r} + \rd_{u} \phi \bb( \dvr^{-1} \rd_{v} (r \phi) - \dur^{-1} \rd_{u} (r \phi) \bb)$
which have the property that the nonlinear terms arising in the equations for these variables in fact have a null structure. In particular, we can apply the above heuristic procedure to obtain decay estimates in the compact region $r\leq R$.
\\
\fparagraph{Non-renormalized variables and decay towards null infinity}

While the renormalization allows us to apply the BV theory in the interior region, it does not give the optimal $r$ decay rates in the exterior region. For example, the renormalized quantity
$$\rd_{v} \log \dvr - \frac{\mu}{(1-\mu)} \frac{\dvr}{r} + \rd_{v} \phi \bb( \dvr^{-1} \rd_{v} (r \phi)  - \dur^{-1} \rd_{u} (r \phi) \bb)$$
decays only as $r^{-2}$ towards null infinity due to the contribution of $\frac{\mu}{(1-\mu)} \frac{\dvr}{r}$, which is weaker than the desired $r^{-3}$ decay for $\rd_v\log \dvr$. Therefore, in order to obtain the optimal estimates everywhere in the spacetime, we need to use the variables $\rd_v^2(r\phi)$, $\rd_u^2(r\phi)$, $\rd_v\dvr$ and $\rd_u\dur$ together with their renormalized versions.
\\
\fparagraph{Coupling of the incoming and outgoing parts}

Finally, an additional challenge is that unlike the estimates for the first derivatives of the scalar field, the bounds for the incoming part of the solution $\rd_v^2(r\phi)$ and $\rd_v\dvr$ are coupled to that for the outgoing part $\rd_u^2(r\phi)$ and $\rd_u\dur$. Likewise, to control $\rd_u^2(r\phi)$, we need estimates for $\rd_v^2(r\phi)$ and $\rd_v\dvr$. For example, in the equation for $\rd_{v} \log \dvr - \frac{\mu}{(1-\mu)} \frac{\dvr}{r} + \rd_{v} \phi \bb( \dvr^{-1} \rd_{v} (r \phi)  - \dur^{-1} \rd_{u} (r \phi) \bb)$, there is a term involving $\rd_u^2(r\phi)$ on the right hand side. In particular, in order to obtain the desired decay for $\rd_v \dvr$, we need to at the same time prove the decay for $\rd_u^2(r\phi)$.
\\
\fparagraph{Strategy for obtaining the decay estimates}
With the above difficulties in mind, we can now give a very rough sketch of the strategy of the proof.

\begin{enumerate}
\item [(1)] Bounds for $\rd_v^2(r\phi)$ and $\rd_v\dvr$ for large $r$: As in the case for the first derivatives, we first prove the optimal $r$ decay for $\rd_v^2(r\phi)$ and $\rd_v\dvr$ in the exterior region. To this end, we integrate the equations satisfied by the \emph{non-renormalized} variables. We note that the error terms can all be bounded using the local BV scattering assumption and the decay estimates already proved for the first derivatives.

\item [(2)] Bounds for all second derivatives: Steps 2 and 3 for the decay bounds for the first derivatives are now coupled. Define \begin{align*}
	\calB_{2}(U) := \sup_{u \in [1, U]} \sup_{C_{u}} \bb( & u^{\om} \abs{\rd_{v}^{2} (r \phi)} +u^{\om} \abs{\rd_{u}^{2} (r \phi)} 
					+ u^{\om} \abs{\rd_{v} \dvr} + u^{\om} \abs{\rd_{u} \dur} \bb).
\end{align*}
We then show that $\calB_{2}(U)$ can control the error terms arising from integrating the \emph{renormalized} equations in the sense that we can obtain an inequality of the form
$$\abs{\mbox{weighted renormalized variables}}\leq C(u_2)+\ep(u_2)\calB_{2}(U),$$
where $\ep(u_2)\to 0$ as $u_2\to \infty$. We then prove that the renormalized variables in fact control all the weighted second derivatives in $\calB_2$. After choosing $u_2$ to be sufficiently large, we show that $\calB_{2}(U)$ is bounded independent of $U$ and thus all the second derivatives have $u^{-\om}$ decay.
\item [(3)] Optimal bounds in terms of $u$ decay: While we have obtained $u^{-\om}$ decay for the second derivatives, the decay rates are not the sharp rates claimed in the main theorem. To finally obtained the desired bounds, we integrate the equations of the \emph{non-renormalized} variables and use the preliminary estimates obtained in (1) and (2) above. Here, we make use of the fact that the estimates obtained in step (2) above are sufficiently strong (both in terms of regularity and decay) to control the error terms in the non-renormalized equations.
\end{enumerate}

\section{Analytic properties of \eqref{eq:SSESF}}\label{sec.anal.prop}
In this section, we discuss the analytic properties of \eqref{eq:SSESF}. These include scaling, monotonicity and the null structure of the system. All these features will play crucial roles in the analysis.
\subsection{Scaling}
For $a>0$, \eqref{eq:SSESF} is invariant under the scaling of the coordinate system
$$ u \mapsto au,\quad v\mapsto av$$
together with the scaling of the functions
$$r \mapsto ar,\quad m\mapsto am,\quad \Omega\mapsto \Omega,\quad\phi\mapsto\phi.$$
This in particular implies that the BV norms
\begin{equation*}
\int_u^{\infty} |\rd_v^2(r\phi)(u,v')| \ud v'
\hbox{ and }
\int_u^{\infty} |\rd_v \lambda(u,v')| \ud v'
\end{equation*}
are scale invariant. Thus the a priori assumptions \eqref{eq:locBVScat} are taken with respect to localized versions of scale invariant norms.

\subsection{Monotonicity properties} \label{subsec:monotonicity}
We first begin with basic monotonicity properties of $r$.
\begin{lemma}[Monotonicity of $r$] \label{lem:mntn4r}
Let $(\phi, r, m)$ be a BV solution to \eqref{eq:SSESF}. Then we have 
\begin{equation*}
	\dur < 0 \hbox{ in } \PD,
\end{equation*}
and
\begin{equation*}
	\left\{
	\begin{aligned}
	\dvr > 0 & \hbox{ when } 1-\mu > 0, \\
	\dvr = 0 & \hbox{ when } 1-\mu = 0, \\
	\dvr < 0 & \hbox{ when } 1-\mu < 0.
	\end{aligned}	
	\right.
\end{equation*}
\end{lemma}
\begin{proof} 
This was proved in \cite[Propositions 1.1 and 1.2]{Christodoulou:1993bt}; we reproduce the proof for the reader's convenience. Note the equation
\begin{equation*}
	\rd_{u} \rd_{v} (r^{2}) = - \frac{1}{2} \Omg^{2}.
 \end{equation*}
which easily follows from \eqref{eq:SSESF}. As $\rd_{u} r^{2} = 2 r \rd_{u} r= 0$ on $\Gmm$ and $r > 0$ on $\PD$, we easily see that $\dur < 0$. Then from the definition of $1-\mu$, the second conclusion also follows. \qedhere
\end{proof}

According to the sign of $\dvr$, a general Penrose diagram $\PD$ is divided into three subregions as follows:
\begin{align*}
	\trpR := \set{(u,v) \in \PD : \dvr < 0}, \quad \appH := \set{(u,v) \in \PD : \dvr = 0}, \quad \regR := \set{(u,v) \in \PD : \dvr > 0}.
\end{align*}

These are called the \emph{trapped region}, \emph{apparent horizon}, and \emph{regular region}, respectively. The next lemma, which we borrow from \cite{Christodoulou:1993bt}, shows that the solutions to \eqref{eq:SSESF} considered in this paper consist only of the regular region $\regR$. Therefore, extensive discussion of $\trpR$ and $\appH$ will be suppressed.

\begin{lemma}[{\cite[Proposition 1.4]{Christodoulou:1993bt}}] \label{lem:regR}
Let $(\phi, r, m)$ be a BV solution to \eqref{eq:SSESF}. Then the causal past of $\Gmm$ in $\PD$ is contained in $\regR$.
In particular, $\PD = \regR$ if $(\phi, r, m)$ satisfies the condition $(1)$ in Definition \ref{def:locBVScat} (future completeness of radial null geodesics).
\end{lemma}

Next, we turn to monotonicity properties of the Hawking mass $m$, which will play an important role in our paper. The following lemma is an obvious consequence of \eqref{eq:SSESF:dm}.

\begin{lemma}[Monotonicity of $m$] \label{lem:mntn4m}
For a BV solution $(\phi, r, m)$ to \eqref{eq:SSESF}, we have
\begin{equation*}
	\rd_{v} m \geq 0, \quad  \rd_{u} m \leq 0 \hbox{ in } \regR.
\end{equation*}
\end{lemma}

By the monotonicity $\rd_{v} m \geq 0$, the limit $M(u):=\lim_{v \to \infty} m(u,v)$ exists (possibly $+\infty$ at this point) for each $u$. This is called the \emph{Bondi mass} at retarded time $u$. The following statement is an easy corollary of the preceding lemma.

\begin{corollary} [Monotonicity of the Bondi mass] \label{cor:mntn4Bondi}
Let $(\phi, r, m)$ be a BV solution to \eqref{eq:SSESF}, and suppose that $C_{u} \subset \regR$ for $u \in [u_{1}, u_{2}]$. Then the Bondi mass $M(u)$ is a non-increasing function on $[u_{1}, u_{2}]$.
\end{corollary}

The following lemma shows that $M_{i} < \infty$ for initial data sets considered in this paper. 
\begin{lemma} \label{lem:bnd4Mi}
Suppose that $\rd_{v}(r \phi)(1, \cdot)$ is asymptotically flat or order $\omg' > 1$ in the sense of Definition \ref{def:AF}. Then we have
\begin{equation} \label{eq:bnd4Mi}
	M_{i} := \lim_{v \to \infty} m(1, v) \leq C \calI_{1}^{2}.
\end{equation}
\end{lemma}

This is an easy consequence of \eqref{eq:SSESF:dm} and Lemma \ref{lem:mntn4r}; we omit its proof.
By the preceding corollary, it follows that $M(u) < \infty$ for each $u$.

We conclude this subsection with additional monotonicity properties of solutions to \eqref{eq:SSESF}, useful for controlling the geometry of locally BV scattering solutions to \eqref{eq:SSESF}.

\begin{lemma} \label{lem:mntn4kpp}
	Let $(\phi, r, m)$ be a BV solution to \eqref{eq:SSESF}. For $(u,v) \in \regR$, we have
	\begin{equation*}
		\frac{\dvr}{1-\mu}(u,v) \leq \frac{\dvr}{1-\mu}(1, v),
	\end{equation*}
	
	$$\rd_u\dvr =\rd_v\dur \leq 0.$$
	\end{lemma}

\begin{proof} 
The lemma follows from the formula
\begin{equation*}
	\rd_{u} \log \abs{\frac{\dvr}{1-\mu}} = - (- \dur)^{-1} r (\rd_{u} \phi)^{2}
\end{equation*} 
and \eqref{eq:SSESF:dr}. \qedhere
\end{proof}

\subsection{Null structure of the evolution equations} \label{subsec:nullStr}
In this subsection, we follow \cite{Christodoulou:1993bt} and demonstrate that the evolution equations verify a form of null structure. In particular, the null structure occurs in the equations for the second derivatives of the scalar field and the metric. However, it is not apparent if we simply take the derivatives of the equations \eqref{eq:SSESF:dr} and \eqref{eq:SSESF:dphi}. Instead, we rewrite the equations in renormalized variables for which the null structure is manifest. We will perform this renormalization separately for the wave equations for $\phi$ and for the equations for $\lambda$ and $\nu$.

\vspace{.1in}
{\it - The wave equation for $\phi$.}
Taking $\rd_{v}$ of the equation \eqref{eq:SSESF:dphi}, we obtain
\begin{equation*}
	\rd_{u} (\rd_{v}^{2} (r \phi)) = \rd_{v} (\rd_{u} \dvr \, \phi) = \rd_{u} \dvr \, \rd_{v} \phi + (\rd_{v} \rd_{u} \dvr) \phi,
\end{equation*}
or equivalently, after substituting in the first equation in \eqref{eq:SSESF:dr},
\begin{equation} \label{eq:eq4dvdvrphi:normal}
\rd_{u} (\rd_{v}^{2} (r \phi)) = 
\frac{2m \dvr \dur}{(1-\mu) r^{2}} \, \rd_{v} \phi + \frac{ \dur}{(1-\mu) }  (\rd_{v} \phi)^{2} \phi 
 + \frac{2m \dur}{(1-\mu) r^{2}} (\rd_{v} \dvr) \phi - \frac{4m}{(1-\mu) r^{3}} \dvr^{2} \dur \phi.
\end{equation}

Some terms on the right hand side, such as $(1-\mu)^{-1} \dur (\rd_{v} \phi)^{2} \phi$, do not exhibit null structure and are dangerous near $\Gmm$. To tackle this, we rewrite
\begin{equation*}
	(\rd_{v} \rd_{u} \dvr) \phi = \rd_{u} [(\rd_{v} \dvr) \phi ] - \rd_{v} \dvr \, \rd_{u} \phi.
\end{equation*}

Thus, from the first equation, we derive
\begin{equation} \label{eq:eq4dvdvrphi}
	\rd_{u} [\rd_{v}^{2} (r \phi) - (\rd_{v} \dvr) \phi] = \rd_{u} \dvr \, \rd_{v} \phi - \rd_{v} \dvr \, \rd_{u} \phi.
\end{equation}
By switching $u$ and $v$, we obtain the following analogous equations in the conjugate direction.
\begin{equation} \label{eq:eq4dudurphi:normal}
\rd_{v} (\rd_{u}^{2} (r \phi)) = 
\frac{2m \dvr \dur}{(1-\mu) r^{2}} \, \rd_{u} \phi + \frac{\dvr }{(1-\mu) }  (\rd_{u} \phi)^{2} \phi 
 + \frac{2m \dvr}{(1-\mu) r^{2}} (\rd_{u} \dur) \phi - \frac{4m}{(1-\mu) r^{3}} \dvr \dur^{2} \phi.
\end{equation}

\begin{equation} \label{eq:eq4dudurphi}
	\rd_{v} [\rd_{u}^{2} (r \phi) - (\rd_{u} \dur) \phi] = \rd_{v} \dur \, \rd_{u} \phi - \rd_{u} \dur \, \rd_{v} \phi.
\end{equation}

\vspace{.1in}
{\it - The equations for $\dvr$ and $\dur$.}
From \eqref{eq:SSESF:dr}, we have
\begin{equation*}
	\rd_{u} \log \dvr = \frac{\mu}{(1-\mu) r} \dur, \quad \rd_{v} \log (-\dur) = \frac{\mu}{(1-\mu) r} \dvr.
\end{equation*}

Take $\rd_{v}$, $\rd_{u}$ of the first and second equations respectively. Using \eqref{eq:SSESF:dr}, it is not difficult to verify that
\begin{align}
\rd_{u} \rd_{v} \log \dvr
=& \frac{1}{(1-\mu) }  \dvr^{-1} \dur (\rd_{v} \phi)^{2} - \frac{4m}{(1-\mu) r^{3}} \dvr \dur, \label{eq:eq4dvdvr:normal} \\
\rd_{v} \rd_{u} \log (-\dur)
=& \frac{1}{(1-\mu) }  \dur^{-1} \dvr (\rd_{u} \phi)^{2} - \frac{4m}{(1-\mu) r^{3}} \dvr \dur. \label{eq:eq4dudur:normal}
\end{align}

To reveal the null structure, we must carry out the renormalization as we have done for \eqref{eq:eq4dvdvrphi}, \eqref{eq:eq4dudurphi}. Following Christodoulou \cite{Christodoulou:1993bt}, it is easy to check that the above two equations are equivalent to
\begin{equation} \label{eq:eq4dvdvr}
\begin{aligned}
& \rd_{u} \bb[ \rd_{v} \log \dvr - \frac{\mu}{(1-\mu)} \frac{\dvr}{r} + \rd_{v} \phi \bb( \dvr^{-1} \rd_{v} (r \phi)  - \dur^{-1} \rd_{u} (r \phi) \bb) \bb] \\
& \qquad = \rd_{u} \phi \, \rd_{v}\bb( \dur^{-1} \rd_{u} (r \phi) \bb)- \rd_{v} \phi \, \rd_{u} \bb( \dur^{-1} \rd_{u} (r \phi) \bb),
\end{aligned}
\end{equation}
and the conjugate equation
\begin{equation} \label{eq:eq4dudur}
\begin{aligned}
& \rd_{v} \bb[ \rd_{u} \log (-\dur) - \frac{\mu}{(1-\mu)} \frac{\dur}{r} + \rd_{u} \phi \bb( \dvr^{-1} \rd_{v} (r \phi) - \dur^{-1} \rd_{u} (r \phi) \bb) \bb]  \\
&\qquad = - \rd_{u} \phi \, \rd_{v}\bb( \dvr^{-1} \rd_{v} (r \phi) \bb) + \rd_{v} \phi \, \rd_{u} \bb( \dvr^{-1} \rd_{v} (r \phi) \bb).
\end{aligned}
\end{equation}

\section{Basic estimates for locally BV scattering solutions} \label{sec.geom}
In this section, we gather some basic estimates concerning locally BV scattering solutions. These estimates will apply, in particular, to solutions satisfying the hypotheses of Theorem \ref{main.thm.1}.

\subsection{Integration lemmas for $\phi$} \label{subsec:est4phi}
We first derive some basic inequalities for $\phi$, $\dvr^{-1} \rd_v(r\phi)$ and $\rd_{v} \phi$. We remark that these are functional inequalities which hold under very general assumptions, and in particular does not rely on the locally BV scattering assumption.

\begin{lemma} \label{lem:est4phi}
Let $\phi(u, \cdot)$ and $r(u, \cdot)$ be Lipschitz functions on $[u,v]$ with $\dvr>0$ and $r(u, u) = 0$. 
Then the following inequality holds.
\begin{equation} \label{eq:intEst4phi:1}
\abs{\phi(u,v)} \leq  \sup_{v' \in [u, v]} \bb\vert \frac{\rd_{v}(r \phi)}{\dvr}(u, v') \bb\vert.
\end{equation}

More generally, for $u \leq v_{1} \leq v_{2}$, we have
\begin{equation} \label{eq:intEst4phi:2}
\abs{r \phi(u,v_{1}) - r \phi(u, v_{2})} \leq  \bb( r(u, v_{2}) - r(u, v_{1}) \bb) \sup_{v' \in [v_{1}, v_{2}]} \bb\vert \frac{\rd_{v}(r \phi)}{\dvr}(u, v') \bb\vert.
\end{equation}

\end{lemma}

\begin{proof} 
We shall prove \eqref{eq:intEst4phi:2}, since \eqref{eq:intEst4phi:1} then follows as a special case. Integrating $\rd_{v} (r \phi)(u, v')$ over $v' \in [v_{1}, v_{2}]$, we get
\begin{align*}
	\abs{r \phi(u,v_{1}) - r \phi (u, v_{2})} 
	\leq & \int_{v_{1}}^{v_{2}} \abs{\rd_{v} (r \phi)(u, v')} \, \ud v' \\
	\leq & \sup_{v' \in [v_{1}, v_{2}]} \bb\vert \frac{\rd_{v}(r \phi)}{\dvr} (u, v') \bb\vert \, \times \int_{v_{1}}^{v_{2}} \dvr(u, v') \, \ud v' \\
	=& \bb( r(u, v_{2}) - r(u, v_{1}) \bb) \sup_{v' \in [v_{1}, v_{2}]} \bb\vert \frac{\rd_{v}(r \phi)}{\dvr} (u, v') \bb\vert . \qedhere
\end{align*}
\end{proof}

\begin{lemma} \label{lem:est4dvphi}
Let $\phi(u, \cdot)$ and $r(u, \cdot)$ be functions on $[u, v]$ such that $\rd_v\phi$ is integrable, $r$ is Lipschitz with $\lambda>0$ and $r(u, u) = 0$. 
Suppose furthermore that $\dvr^{-1} \rd_{v} (r \phi)(u, \cdot)$ is BV on $[u, v]$. Then the following statements hold.
\begin{enumerate}
\item We have
\begin{equation} \label{eq:est4dvphi:2}
\int_{u}^{v} \abs{\rd_{v} \phi(u,v')} \, \ud v' \leq \int_{u}^{v}\abs{\rd_{v}(\dvr^{-1} \rd_{v} ( r \phi))(u,v')} \, \ud v'.
\end{equation}

\item Suppose, in addition, that $\dvr^{-1} \rd_{v}(r \phi)(u, \cdot)$ is Lipschitz on $[u,v]$. Then we have
\begin{equation} \label{eq:est4dvphi:1}
\abs{\rd_{v} \phi(u,v)} \leq \frac{1}{2} \frac{\sup_{v' \in [u,v]} \dvr(u, v')}{\inf_{v' \in [u,v]} \dvr(u, v')}  \sup_{v' \in [u, v]} \abs{\rd_{v} (\dvr^{-1} \rd_{v} ( r \phi))(u, v')}.
\end{equation}
\end{enumerate}
\end{lemma}

\begin{proof} 
We proceed formally to compute
\begin{align*}
	\rd_{v} \phi(u,v) 
	=& \frac{\dvr}{r} \bb( \dvr^{-1} \rd_{v} ( r \phi) - \phi\bb) (u,v) \\
	=& \frac{\dvr}{r^{2}} (u, v) \int_{u}^{v} \bb( \int_{v'}^{v} \rd_{v} (\dvr^{-1} \rd_{v} ( r \phi)) (u, v'') \, \ud v'' \bb) \dvr(u, v')\, \ud v' \\
	=& \frac{\dvr}{r^{2}} (u,v) \int_{u}^{v} r(u, v'') \rd_{v} (\dvr^{-1} \rd_{v} ( r \phi)) (u, v'') \, \ud v''.
\end{align*}

The above computation is justified thanks to the hypotheses, where we interpret 
\begin{equation*}
	\rd_{v} (\dvr^{-1} \rd_{v} ( r \phi)) (u, v'') \, \ud v''
\end{equation*}
to be the the weak derivative of $\dvr^{-1} \rd_{v} ( r \phi)$, which is a finite signed measure. For a fixed $(u, v)$, observe that
\begin{equation*}
	\sup_{v'' \in [u, v]} r(u, v'') \int_{v''}^{v} \frac{\dvr(u,v')}{r^{2}(u,v')} \, \ud v' \leq 1.
\end{equation*}

This proves \eqref{eq:est4dvphi:2}. For \eqref{eq:est4dvphi:1}, note that the function $\dvr^{-1} \rd_{v} (r \phi)$ is absolutely continuous on $[u,v]$, so $\rd_{v}(\dvr^{-1} \rd_{v} ( r \phi)(u, \cdot))$ exists almost everywhere on $[u, v]$; moreover, it belongs to $L^{\infty}$ by the Lipschitz assumption. Noting that
\begin{equation*}
	\sup_{v' \in [u, v]} \frac{\dvr(u,v')}{r^{2}(u,v')} \int_{u}^{v'} r(u, v'') \, \ud v'' \leq \frac{1}{2} \frac{\sup_{v' \in [u,v]} \dvr(u, v')}{\inf_{v' \in [u,v]} \dvr(u, v')}
\end{equation*}
we obtain \eqref{eq:est4dvphi:1}.
\end{proof}

\subsection{Geometry of locally BV scattering solutions}
The goal of this subsection is to prove the following proposition.
\begin{proposition} \label{prop:geomLocBVScat}
Let $(\phi, r, m)$ be a locally BV scattering solution to \eqref{eq:SSESF} as in Definition \ref{def:locBVScat}. Assume furthermore that on the initial slice $C_{1}$, we have $\dvr(1, \cdot) = \frac{1}{2}$ and
\begin{equation*}
	\sup_{C_{1}} \abs{\rd_{v}(r \phi)} + M_{i} < \infty.
\end{equation*}
	
Then there exist finite constants $K, \Lmb > 0$ such that the following bounds hold for all $(u, v) \in \PD$:
\begin{gather}
	\Lmb^{-1} \leq \dvr(u,v) \leq \frac{1}{2} \label{eq:bnd4dvr} \\
	\Lmb^{-1} \leq - \dur(u,v) \leq K \label{eq:bnd4dur} \\
	1 \leq (1-\mu(u,v))^{-1} \leq K \Lmb. \label{eq:bnd4mu}\\
	0 < \frac{- \dur}{1-\mu(u,v)} \leq K. \label{eq:bnd4conjKpp}
\end{gather}

Moreover, there exists a finite constant $\Psi > 0$ such that for all $(u,v) \in \PD$, we have
\begin{gather}
	\abs{\rd_{v}(r \phi)(u,v)} \leq \Psi, \label{eq:bnd4dvrphi} \\
	\abs{\phi(u,v)} \leq \Lmb \Psi. \label{eq:bnd4phi}
\end{gather}
\end{proposition}

{\bf Once we have this proposition, we will denote by $\Lambda$, $K$ and $\Psi$ the best constants such that \eqref{eq:bnd4dvr}-\eqref{eq:bnd4phi} hold.}

By Lemma \ref{lem:regR}, we already know that $\dvr > 0$, $- \dur > 0$ and $(1-\mu)^{-1} < \infty$. The first three bounds, namely \eqref{eq:bnd4dvr}--\eqref{eq:bnd4mu}, ensure that these bounds concerning the geometry of the spacetime does not degenerate anywhere, in particular along the axis $\Gmm$. They will be very useful in the analysis in the later section of the paper. 

The proof of Proposition \ref{prop:geomLocBVScat} will consist of several steps. We begin with elementary bounds for $\dvr$ and $\dur$.

\begin{lemma} \label{lem:basicEst4dr}
Let $(\phi, r, m)$ be a BV solution to \eqref{eq:SSESF} with $\PD = \regR$. Then for every $(u,v) \in \PD$, we have
\begin{align}
	\dvr(u,v) \leq& \, \dvr(1, v), \label{eq:basicEst4dr:1} \\
	\dvr^{-1}(u,v) \leq& \, \lim_{u' \to v-} \dvr^{-1}(u',v), \label{eq:basicEst4dr:2} \\
	\dur(u,v) \leq& - \lim_{v' \to u+} \dvr(u, v'). \label{eq:basicEst4dr:3}
\end{align}
\end{lemma}

\begin{proof}
 By \eqref{eq:SSESF:dr}, we have
\begin{equation*}
\begin{aligned}
	\dvr(u,v) =& \dvr(1, v) \exp \bb( \int_{1}^{u} \bb( \frac{2m}{(1-\mu)r^{2}} \dur \bb) (u', v) \, \ud u' \bb), \\
	\dvr^{-1}(u,v) =& \lim_{u' \to v-} \dvr(u', v)^{-1} \exp \bb( \int_{u}^{v} \bb( \frac{2m}{(1-\mu)r^{2}} \dur \bb) (u', v) \, \ud u' \bb), \\
	\dur(u,v) =& \lim_{v' \to u+} \dur(u, v') \exp \bb( \int_{u}^{v} \bb( \frac{2m}{(1-\mu)r^{2}} \dvr \bb) (u, v') \, \ud v' \bb).
\end{aligned}
\end{equation*}

Since $-\dur, (1-\mu) > 0$ everywhere, \eqref{eq:basicEst4dr:1} and \eqref{eq:basicEst4dr:2} follow. Moreover, since 
\begin{equation*}
\lim_{v' \to u+} \dur(u,v') = - \lim_{v' \to u+}\dvr(u,v'),
\end{equation*}
and $\dvr > 0$ on $\PD$, \eqref{eq:basicEst4dr:3} follows as well. \qedhere
\end{proof}

By Lemma \ref{lem:regR}, $\PD = \regR$ holds for a solution \eqref{eq:SSESF} satisfying the hypotheses of Proposition \ref{prop:geomLocBVScat}. As an immediate corollary, we have the following easy upper bound for $\dvr$.
\begin{corollary}  \label{cor:est4dr}
Let $(\phi, r, m)$ be a solution to \eqref{eq:SSESF} satisfying the hypotheses of Proposition \ref{prop:geomLocBVScat}. Then by the coordinate condition $\dvr(1, v) = \frac{1}{2}$ and \eqref{eq:basicEst4dr:1}, we have
\begin{equation*}
	\sup_{\PD} \dvr \leq \frac{1}{2} \, .
\end{equation*}
\end{corollary}

Next, we proceed to prove the lower bounds of \eqref{eq:bnd4dvr} and \eqref{eq:bnd4dur}. We begin with a technical lemma concerning a large-$r$ region, which will also be useful in our proof of \eqref{eq:bnd4dvrphi} and \eqref{eq:bnd4phi}.

\begin{lemma} \label{lem:babySmllPtnl:1}
Let $(\phi, r, m)$ be a solution to \eqref{eq:SSESF} satisfying the hypotheses of Proposition \ref{prop:geomLocBVScat}. Then for arbitrarily small $\eps > 0$, there exists $r_{0} > 1$ such that
\begin{align} 
	\sup_{(u,v) \in \set{r \geq r_{0}}} \int_{1}^{u} \abs{\frac{\mu}{(1-\mu)} \frac{\dur}{r} (u', v)} \, \ud u' <&  \eps \, . \label{eq:babySmllPtnl:1}
\end{align}
\end{lemma}
\begin{proof} 
For $(u,v) \in \set{r \geq r_{0}}$, we begin by simply estimating as follows:
\begin{equation*}
	\abs{\frac{\mu}{(1-\mu)} \frac{\dur}{r}} \leq \frac{2 M_{i}}{(1-\frac{2M_{i}}{r_{0}})} \frac{(- \dur)}{r^{2}}
\end{equation*}

The above inequality holds as long as\footnote{Indeed, it suffices to choose $r_0>2M_i$ here. The condition $r_0> R$ will be used in the proof of Lemma \ref{lem:babySmllPtnl:2}.} we choose $r_{0} > \max \set{2 M_{i}, R}$. Note that if $(u, v) \in \set{r \geq r_{0}}$, then the null curve $\set{(u', v) : u' \in [1, u]}$ from the initial slice $C_{1}$ to $(u,v)$ lies entirely in $\set{r \geq r_{0}}$. Integrating along this curve, we obtain for $(u, v) \in \set{r \geq r_{0}}$
\begin{equation*}
	\int_{1}^{u} \abs{\frac{\mu}{(1-\mu)} \frac{\dur}{r}(u',v)} \, \ud u' < \frac{2 M_{i}}{(1-\frac{2M_{i}}{r_{0}})} \frac{1}{r_{0}}
\end{equation*}

Taking $r_{0}$ sufficiently large, \eqref{eq:babySmllPtnl:1} follows. \qedhere

\end{proof}

Next, we prove an analogous result in a large $u$ region. Key to its proof will be the identity \eqref{eq:babySmllPtnl:pf:0} below, which will also be used to relate \eqref{eq:babySmllPtnl:1} and \eqref{eq:babySmllPtnl:2} to the desired lower bounds of $\dvr$ and $-\dur$.

\begin{lemma} \label{lem:babySmllPtnl:2}
Let $(\phi, r, m)$ be a solution to \eqref{eq:SSESF} satisfying the hypotheses of Proposition \ref{prop:geomLocBVScat}. Then for arbitrarily small $\eps > 0$, there exists $U > 1$ such that
\begin{align} 
	\sup_{v \geq U} \int_{U}^{v} \abs{\frac{\mu}{1-\mu} \frac{\dur}{r} (u', v)} \, \ud u' <&  \eps \, . \label{eq:babySmllPtnl:2}
\end{align}
\end{lemma}

\begin{proof} 
Let $\eps > 0$ be an arbitrary positive number. Using \eqref{eq:SSESF:dr} and the fact that $1-\mu > 0, -\nu > 0$ on $\PD$, we have for any $1 \leq u_{1} \leq u_{2} < v$,
\begin{equation} \label{eq:babySmllPtnl:pf:0}
	\int_{u_{1}}^{u_{2}} \abs{\frac{\mu}{1-\mu} \frac{\dur}{r} (u', v)} \, \ud u' = \log \dvr(u_{1}, v) - \log \dvr(u_{2}, v).
\end{equation}

In order to prove \eqref{eq:babySmllPtnl:2}, it therefore suffices to exhibit $U > 1$ such that 
\begin{equation} \label{eq:babySmllPtnl:pf:1}
	\sup_{(u, v), (u', v') \in \set{u \geq U}} \abs{\log \dvr(u, v) - \log \dvr(u', v')} < \eps.
\end{equation} 

In order to proceed, we divide $\PD$ into three regions: $\cpt := \set{r \leq R}$, $\PD_{[R, r_{0}]} := \set{R \leq r \leq r_{0}}$ and $\PD_{[r_{0}, \infty)} := \set{r \geq r_{0}}$, where $r_{0} > \max\{2M_i,R\}$ is chosen via Lemma \ref{lem:babySmllPtnl:1} so that
\begin{equation*} 
	\sup_{(u,v) \in \PD_{[r_{0}, \infty)}}\int_{1}^{u} \abs{\frac{\mu}{1-\mu} \frac{\dur}{r}(u',v)} \, \ud u' < \frac{\eps}{8}.
\end{equation*}

Using \eqref{eq:babySmllPtnl:pf:0} and the fact that $\log \dvr(1, v) = \frac{1}{2}$, the preceding inequality is equivalent to
\begin{equation} \label{eq:babySmllPtnl:pf:2}
	\sup_{(u,v) \in \PD_{[r_{0}, \infty)} } \abs{\log \dvr(u,v) - \frac{1}{2}} < \frac{\eps}{8}.
\end{equation}

Next, we turn to the region $\PD_{[R, r_{0}]}$; here we exploit the vanishing of the final Bondi mass. Indeed, taking $U_{1}$ large enough so that $2 M(U_{1}) < R$, we may estimate
\begin{equation*}
	\abs{\frac{\mu}{1-\mu} \frac{\dur}{r}} \leq \frac{2 M(U_{1})}{(1-\frac{2 M(U_{1})}{R}) R^{2}} (-\dur)\quad\mbox{for }u\geq U_1.
\end{equation*}

Consider now the time-like curve given by $\gmm_{0} := \set{(u',v') : r(u', v') = r_{0}}$. On $\gmm_{0} \cap \set{(u,v) : u \geq U_{1}}$, note that \eqref{eq:babySmllPtnl:pf:2} holds. Integrating the preceding inequality along incoming null curves emanating from $\gmm_{0} \cap \set{(u,v) : u \geq U_{1}}$, we obtain for $(u, v) \in \PD_{[R, r_{0}]} \cap \set{(u,v) : u \geq U_{2}}$
\begin{equation*} 
	\abs{\log \dvr(u,v) - \frac{1}{2}} < \frac{\eps}{8} + \frac{2 M(U_{1}) (r_{0} - R)}{(1- \frac{2 M(U_{1})}{R}) R^{2}} .
\end{equation*}
where $U_{2} = U_{2}(U_{1}, r_{0})$ is the future endpoint of the incoming null curve in $\PD_{[R, r_{0}]}$ from the past endpoint of $\gmm_{0} \cap \set{(u,v) : u \geq U_{1}}$; more precisely, $U_{2} = \sup \set{u : r(u, V_{1}) \geq R}$, where $V_{1}$ is the defined by $r(U_{1}, V_{1}) = r_{0}$. Choosing $U_{1}$ sufficiently large, we then obtain
\begin{equation} \label{eq:babySmllPtnl:pf:3}
	\sup_{(u, v) \in \PD_{[R, r_{0}]} \cap \set{u \geq U_{2}}} \abs{\log \dvr(u,v) - \frac{1}{2}} < \frac{\eps}{4}.
\end{equation}

Finally, in $\cpt$, we use the local BV scattering condition \eqref{eq:locBVScat} to choose $U \geq U_{2}$ large enough so that we have
\begin{equation} \label{eq:babySmllPtnl:pf:4}
	\sup_{(u, v), (u, v') \in \cpt \cap \set{u \geq U}}\abs{\log \dvr (u, v) - \log \dvr(u, v')} < \frac{\eps}{4}.
\end{equation}

To compare $\log \dvr(u, v)$ and $\log \dvr(u', v')$ with $u \neq u'$, we use \eqref{eq:babySmllPtnl:pf:3}, \eqref{eq:babySmllPtnl:pf:4} and the triangle inequality. Thus, the desired conclusion \eqref{eq:babySmllPtnl:pf:1} follows. \qedhere
\end{proof}

As a corollary of the preceding lemmas and \eqref{eq:babySmllPtnl:pf:0} (or, more directly, \eqref{eq:babySmllPtnl:pf:1} and \eqref{eq:babySmllPtnl:pf:2}), we immediately see that $\dvr$ and $-\dur$ is uniformly bounded away from zero.
\begin{corollary} \label{cor:lowerBnd4dvr}
Let $(\phi, r, m)$ be a solution to \eqref{eq:SSESF} satisfying the hypotheses of Proposition \ref{prop:geomLocBVScat}. Then there exists $0 < \Lmb < \infty$ such that for all $(u, v) \in \PD$, we have
\begin{equation*}
	\Lmb^{-1} \leq \dvr(u, v), \quad
	\Lmb^{-1} \leq - \dur(u,v).
\end{equation*}
\end{corollary}

Together with Corollary \ref{cor:est4dr}, this concludes the proof of \eqref{eq:bnd4dvr}. Next, using Lemmas \ref{lem:est4phi}, \ref{lem:babySmllPtnl:1}, \ref{lem:babySmllPtnl:2} and the wave equation \eqref{eq:SSESF:dphi} for $\phi$, we prove \eqref{eq:bnd4dvrphi}, \eqref{eq:bnd4phi} in the following lemma.

\begin{lemma}  \label{lem:bnd4dvrphiphi}
Let $(\phi, r, m)$ be a solution to \eqref{eq:SSESF} satisfying the hypotheses of Proposition \ref{prop:geomLocBVScat}. Then there exists a constant $0 < \Psi < \infty$ such that
\begin{equation} \label{eq:bnd4dvrphiphi}
	\sup_{\PD} \abs{\rd_{v}(r \phi)} \leq \Psi, \quad 
	\sup_{\PD} \abs{\phi} \leq \Lmb \Psi,
\end{equation}
where $\Lmb$ is the best constant such that Corollary \ref{cor:lowerBnd4dvr} holds.
\end{lemma}
\begin{proof} 
Note that the second inequality of \eqref{eq:bnd4dvrphiphi} is an immediate consequence of the first inequality, Lemma \ref{lem:est4phi} and Corollary \ref{cor:lowerBnd4dvr}. The proof of the first inequality will proceed in two steps: First, we shall show that $\rd_{v}(r \phi)$ is uniformly bounded on the large $r$ region, essentially via Lemma \ref{lem:babySmllPtnl:1}. By compactness, it immediately follows that $\rd_{v}(r \phi)$ is uniformly bounded on the finite $u$ region. Then in the second step, we shall show that $\rd_{v}(r \phi)$ is uniformly bounded on a large $u$ region as well using Lemma \ref{lem:babySmllPtnl:2}.

By Lemma \ref{lem:babySmllPtnl:1}, choose $r_{0} > 0$ so that 
\begin{equation} \label{eq:bdd4dvrphiphi:pf:1}
	\sup_{(u,v) \in \set{r \geq r_{0}}} \int_{1}^{u} \abs{\frac{\mu}{1-\mu} \frac{\dur}{r} (u', v)} \, \ud u' < \frac{1}{10 \Lmb}.
\end{equation}

We also borrow the notation $\PD_{[r_{0}, \infty)} := \set{(u,v) : r(u,v) \geq r_{0}}$ from the proof of Lemma \ref{lem:babySmllPtnl:2}. Given $U \geq 1$, define $\Psi_{[r_{0}, \infty)}(U)$ to be
\begin{equation*}
	\Psi_{[r_{0}, \infty)}(U) := \sup_{(u, v) \in \PD_{[r_{0}, \infty)} \cap \set{1 \leq u \leq U}} \abs{\rd_{v} (r \phi)(u, v)}.  
\end{equation*}

Let $(u,v) \in \PD_{[r_{0}, \infty)}$. Using \eqref{eq:SSESF:dphi}, we then write
\begin{align*}
	\rd_{u} \rd_{v} ( r \phi)
	= & \frac{\mu}{1-\mu} \frac{\dur}{r} \bb( \frac{\dvr}{r} (r \phi - r_{0} \phi_{r_{0}}) + \frac{\dvr}{r} r_{0} \phi_{r_{0}} \bb).
\end{align*} 

Here, $\phi_{r_{0}}(u,v) := \phi(u, v^{\star}_{0}(u))$, where $v^{\star}_{0}(u)$ is the unique $v$-value for which $r(u, v^{\star}_{0}(u)) = r_{0}$. Note that the outgoing null curve from $(u, v^{\star}_{0}(u))$ to $(u,v) \in \PD_{[r_{0}, \infty)}$ lies entirely in $\PD_{[r_{0}, \infty)}$. Thus, by Lemma \ref{lem:est4phi} and \eqref{eq:bnd4dvr}, we see that for $(u, v) \in \PD_{[r_{0}, \infty)}$ with $1 \leq u \leq U$, 
\begin{align*}
	\abs{\rd_{u} \rd_{v} ( r \phi)} 
	\leq & \abs{\frac{\mu}{1-\mu} \frac{\dur}{r}} \bb( \frac{(r - r_{0})}{2 r} \Lmb \Psi_{[r_{0}, \infty)}(U) + \frac{r_{0}}{2 r} \abs{\phi_{r_{0}}} \bb) \\
	\leq & \abs{\frac{\mu}{1-\mu} \frac{\dur}{r}} \bb( \Lmb \Psi_{[r_{0}, \infty)}(U) + \abs{\phi_{r_{0}}} \bb).
\end{align*}

Integrating this equation over the incoming null curve from $(1, v)$ to $(u, v)$ (which lies in $\PD_{[r_{0}, \infty)} \cap \set{1 \leq u \leq U}$) and using Lemma \ref{lem:babySmllPtnl:1}, we then obtain
\begin{align*}
	\Psi_{[r_{0}, \infty)}(U)	\leq \sup_{C_{1} \cap \PD_{[r_{0}, \infty)}} \abs{\rd_{v} (r \phi)} + \frac{1}{10} \Psi_{[r_{0}, \infty)}(U) + \frac{1}{10 \Lmb} \sup_{\gmm_{0} \cap \set{1 \leq u \leq U}} \abs{\phi}
\end{align*}
where $\gmm_{0}$ is the time-like curve $\set{(u,v) : r(u,v) = r_{0}}$. Note that the first term on the right-hand side is finite by the assumptions on the initial data, whereas the last term is finite for every $1 \leq U < \infty$ by compactness of $\gmm_{0} \cap \set{(u,v) : 1 \leq u \leq U}$ and continuity of $\phi$. Then, by a simple continuity argument, it follows that $\Psi_{[r_{0}, \infty)}(U) < \infty$ for every $1 \leq U < \infty$. Moreover, by compactness of $\set{(u, v) : r(u,v) \leq r_{0}, \, 1 \leq u \leq U}$, as well as the uniform BV assumption on $\rd_{v}(r \phi)$, we also have
\begin{equation*}
	\Psi_{[0, \infty)}(U) := \sup_{(u,v) \in \set{1 \leq u \leq U}} \abs{\rd_{v}(r \phi)(u,v)} < \infty.
\end{equation*}

We now proceed to deal with the large-$u$ region, namely $\set{(u,v) : u \geq U}$. Using Lemma \ref{lem:babySmllPtnl:2}, we choose $U_{0} \geq 1$ sufficiently large so that
\begin{equation}
	\sup_{v \geq U_{0}} \int_{U_{0}}^{v} \abs{\frac{\mu}{1-\mu} \frac{\dur}{r} (u', v)} \, \ud u' < \frac{1}{10 \Lmb}.
\end{equation}

Proceeding as before via Lemma \ref{lem:est4phi}, we estimate for $(u,v) \in \set{(u,v) : u \geq U_{0}}$ 
\begin{align*}
	\abs{\rd_{u} \rd_{v} (r \phi)(u,v)} \leq \abs{\frac{\mu}{1-\mu} \frac{\dur}{r}} \, \Lmb \sup_{v' \in [u, v]} \abs{\rd_{v}(r \phi)(u, v')} .
\end{align*}

Integrating along incoming null curves from $C_{U_{0}}$, we see that
\begin{equation*}
	\Psi_{[0, \infty)}(U) \leq \Psi_{[0, \infty)}(U_{0}) + \frac{1}{10} \Psi_{[0, \infty)}(U)
\end{equation*}
for any $U \geq U_{0}$. Absorbing the second term on the right-hand side into the left-hand side and taking $U \to \infty$, we obtain \eqref{eq:bnd4dvrphiphi} with $\Psi \leq \frac{10}{9} \Psi_{[0, \infty)}(U_{0}) < \infty$. \qedhere
\end{proof}

We are finally ready to conclude the proof of Proposition \ref{prop:geomLocBVScat}, by proving \eqref{eq:bnd4conjKpp}. Indeed, the upper bounds in \eqref{eq:bnd4dur} and \eqref{eq:bnd4mu} would then follow immediately. Moreover, the lower bound in \eqref{eq:bnd4mu} is trivial, as $\mu = \frac{2m}{r} \geq 0$.

\begin{lemma} \label{lem:est4dur}
Let $(\phi, r, m)$ be a solution to \eqref{eq:SSESF} satisfying the hypotheses of Proposition \ref{prop:geomLocBVScat}. Then there exists a finite constant $K > 0$ such that for all $(u,v) \in \PD$,
\begin{equation} \label{eq:est4dur:key}	
 \frac{- \dur}{1-\mu} (u,v) \leq K.
\end{equation}
\end{lemma}

\begin{proof} 
To prove \eqref{eq:est4dur:key}, we shall rely on the equation
\begin{equation} \label{eq:mntn4durOver1-mu}
	\rd_{v} \log \bb( \frac{-\dur}{1-\mu} \bb) = \dvr^{-1} r (\rd_{v} \phi)^{2},
\end{equation}
which may be easily derived from \eqref{eq:SSESF:dr} and \eqref{eq:SSESF:dm}.

For $(u, v) \in \PD$, we begin by integrating \eqref{eq:mntn4durOver1-mu} on the outgoing null curve from $(u,u) \in \Gmm$ to $(u,v)$, which gives
\begin{equation*}
	\bb( \frac{-\dur}{1-\mu} \bb)(u,v) \leq \bb( \lim_{v' \to u+} \bb( \frac{-\dur}{1-\mu} \bb) (u, v') \bb)  \exp \bb( \int_{u}^{v} \dvr^{-1} r (\rd_{v} \phi)^{2} (u, v') \, \ud v' \bb).
\end{equation*}

We claim that $\lim_{v' \to u+} (- \dur) (u, v') = \lim_{v' \to u+} \dvr (u, v')\leq \frac 12$ and $\lim_{v' \to u+} \mu(u,v') = 0$. The first assertion is obvious. To prove the second one, we first use \eqref{eq:SSESF:dm} to write
\begin{equation*}
m(u,v) \leq \tfrac{1}{2} \bb( \sup_{v' \in [u, v]} \abs{r^{2} \rd_{v} \phi}(u, v') \bb) \int_{u}^{v} \abs{\rd_{v} \phi(u, v')} \, \ud v'.
\end{equation*}
Now observe that $\sup_{v' \in [u, v]} \abs{r^{2} \rd_{v} \phi}(u, v') \leq C r(u,v) \sup_{v' \in [u,v]} \abs{\rd_{v}(r \phi)}$, and the remaining integral goes to $0$ as $v \to u+$, since $\phi$ is assumed to be absolutely continuous on $C_{u}$ near the axis by Definition \ref{def:BVsolution}. 

By the above claim, we have
\begin{equation*}
	\bb( \frac{-\dur}{1-\mu} \bb)(u,v) \leq \frac{1}{2} \exp \bb( \int_{u}^{v} \dvr^{-1} r (\rd_{v} \phi)^{2} (u, v') \, \ud v' \bb).
\end{equation*}

%
The lemma would therefore follow if we could prove
\begin{equation*}
\sup_{(u,v) \in \PD} \int_{u}^{v} \dvr^{-1} r (\rd_{v} \phi)^{2} (u, v') \, \ud v' < \infty.
\end{equation*}

To achieve this, we shall divide the integral into two parts, one in $\cpt$ and the other in its complement $\cpt^c$. Indeed, defining $v^{\star}(u)$ to be the unique $v$ value such that $r(u, v^{\star}(u)) = R$, we divide the integral into $\int_{u}^{v^{\star}(u)}$ and $\int_{v^{\star}(u)}^{v}$. If $v < v^{\star}(u)$, the latter integral will be taken to be zero.

For the first integral, let us begin by pulling out $\dvr^{-1} r \rd_{v} \phi$ from the integral. Using the identity $\dvr^{-1} r \rd_{v} \phi = \dvr^{-1} \rd_{v} (r \phi) - \phi$ we have
\begin{align*}
&\int_{u}^{v^{\star}(u)} \dvr^{-1} r (\rd_{v} \phi)^{2} (u, v') \, \ud v'\\
& \quad \leq \sup_{v' \in [u, v^{\star}(u)]} \bb( \dvr^{-1} \abs{\rd_{v}(r\phi)}(u, v') + \abs{\phi}(u, v') \bb) \int_{u}^{v^{\star}(u)} \abs{\rd_{v} \phi(u,v')} \, \ud v'.
\end{align*}

Then by Lemmas \ref{lem:est4dvphi}, \ref{lem:bnd4dvrphiphi} and the local BV scattering assumption, the right-hand side is uniformly bounded in $u$ from above, as desired. For the second integral, note that, by Lemma \ref{lem:mntn4kpp} and Corollary \ref{cor:lowerBnd4dvr}, we have
\begin{equation*}
(1-\mu)^{-1}(u,v) \leq \Lmb \frac{\dvr}{1-\mu}(u,v) \leq \frac{\Lmb}{2} \sup_{C_{1}} (1-\mu)^{-1}. 
\end{equation*}

Notice that the quantity $\sup_{C_{1}}(1-\mu)^{-1}$ for the initial data is finite, since $1-\mu > 0$ everywhere and $1-\mu(1, v) \to 1$ as $v \to \infty$.
Moreover, for $v \geq v^{\star}(u)$, we have $r(u, v) \geq R$. Therefore, in view of \eqref{eq:SSESF:dm}, we may estimate
\begin{align*}
	\int_{v^{\star}(u)}^{v} \dvr^{-1} r (\rd_{v} \phi)^{2} \, \ud v'
	\leq & \frac{\Lmb}{R} \sup_{C_{1}} (1-\mu)^{-1} \int_{v^{\star}(u)}^{v} \frac{1}{2} \dvr^{-1} (1-\mu) r^{2} (\rd_{v} \phi)^{2} (u, v') \, \ud v' \\
	\leq & \frac{\Lmb}{R} \sup_{C_{1}} (1-\mu)^{-1} (m(u, v) - m(u, v^{\star}(u))) \\
	\leq & C_{\Lmb, R, M_{i}, \sup_{C_{1}} (1-\mu)^{-1}} < \infty,
\end{align*}
from which the lemma follows. \qedhere
\end{proof}

We conclude this subsection with a pair of identities which are useful for estimating $\int\abs{\rd_{u} \dvr} \, \ud u$ and $\int \abs{\rd_{v} \dur} \, \ud v$ in terms of information on $\phi$.
\begin{lemma} \label{lem:auxEqs}
From \eqref{eq:SSESF}, the following identities hold:
\begin{align}
	\int_{u}^{v} \frac{\mu}{1-\mu} \frac{\dvr}{r} (u,v') \, \ud v' = & \log(1-\mu)(u,v) + \int_{u}^{v} \dvr^{-1} r (\rd_{v} \phi)^{2} (u, v') \, \ud v', \label{eq:auxEqs:1} \\
	\int_{u}^{v} \frac{\mu}{1-\mu} \frac{(-\dur)}{r} (u',v) \, \ud u' = & \log(1-\mu)(u,v) + \int_{u}^{v} (-\dur)^{-1} r (\rd_{u} \phi)^{2} (u', v) \, \ud u'. \label{eq:auxEqs:2}
\end{align}
\end{lemma}

\begin{proof} 
	We shall prove \eqref{eq:auxEqs:1}, leaving the similar proof of \eqref{eq:auxEqs:2} to the reader. 
	From the proof of Lemma \ref{lem:basicEst4dr}, we have
	\begin{equation*}
		\int_{u}^{v} \frac{\mu}{1-\mu} \frac{\dvr}{r}  (u,v') = \log \frac{\dur(u, v)}{\lim_{v' \to u+} \dur(u,v')}.
	\end{equation*}
	
	Comparing with the integral of \eqref{eq:mntn4durOver1-mu}, along with the fact that $\lim_{v' \to u+} (1-\mu)(u, v') = 1$, we arrive at \eqref{eq:auxEqs:1}. \qedhere
\end{proof}

\subsection{Normalization of the coordinate system}\label{sec.coord}

In \S \ref{subsec:coordSys}, the coordinates are normalized such that $\dvr$ is constant on the initial hypersurface $\{u=1\}$. Alternatively, one can introduce a new coordinate system $(u_{\infty},v_{\infty})$ which is normalized at future null infinity by requiring that $\nu_{\infty}\to-\frac 12$ along each outgoing null curve towards null infinity and require, as before, that $\Gamma=\{(u,v):u=v\}$. We will show that the coordinate functions $u$ and $u_{\infty}$ are comparable and thus the main theorem on the decay rates can also be stated in this alternatively normalized coordinate system.

We can compute explicitly the coordinate change, which is given by
\begin{equation*}
\frac{du_{\infty}}{du}(u)=-2\lim_{v\to\infty}\nu(u,v),\quad u_\infty(1)=1
\hbox{ and }
v_{\infty}(v)=u_{\infty}(v).
\end{equation*}

Notice that the limit $\displaystyle\lim_{v\to\infty}\nu(u,v)$ is well-defined due to the monotonicity of $\nu$.

\begin{equation*}
	u_{\infty}(u) = - 2\int_{1}^{u} \bb(\lim_{v\to\infty}\nu(u',v)\bb) \, \ud u' + 1,
\end{equation*}

By Proposition \ref{prop:geomLocBVScat}, the following estimate holds:
\begin{equation*}
	2(\Lmb)^{-1} (u-1) \leq u_{\infty}-1 \leq 2 K (u-1).
\end{equation*}

\subsection{Consequence of local BV scattering}
In this subsection, we give some estimates for $\rd_u^2(r\phi)$, $\rd_u\phi$ and $\rd_u \nu$ that follow from from the local BV scattering assumption. To this end, we will need the analysis for solutions to \eqref{eq:SSESF} with small bounded variation norm by Christodoulou in \cite{Christodoulou:1993bt}. In particular, Christodoulou proved
\begin{theorem}[{Christodoulou \cite[Theorem 6.2]{Christodoulou:1993bt}}]\label{Chr.BV.Thm}
There exists universal constants $\ep_0$ and $C_0$ such that for $\ep<\ep_0$, if $\dvr(1, \cdot) = \frac{1}{2}$ and $\rd_{v} (r \phi)(1, \cdot)$ is of bounded variation with
\begin{equation} \label{Chr.BV.Thm.hyp}
\int_{C_1} |\rd_v^2(r\phi)| <\ep,
\end{equation}
then its maximal development $(\phi, r, m)$ satisfies condition $(1)$ in Definition \ref{def:locBVScat} (future completeness of radial null geodesics) and obeys
\begin{gather}
	\frac{1}{3} \leq \dvr \leq \frac{1}{2}, \quad 
	\frac{1}{3} \leq - \dur \leq \frac{2}{3}, \quad
	\frac{2}{3} \leq (1-\mu) \leq 1, \label{Chr.BV.Thm.geom} \\
	\sup_{u \geq 1} \int_{C_{u}} \bb( \abs{\rd_{v} (\dvr^{-1} \rd_{v} (r \phi))} + \abs{\rd_{v} \phi} + \abs{\rd_{v} \log \dvr} \bb) < C_{0} \eps, \label{Chr.BV.Thm.dv} \\
	\sup_{v \geq 1} \int_{\uC_{v}} \bb( \abs{\rd_{u} (\dur^{-1} \rd_{u} (r \phi))} + \abs{\rd_{u} \phi} + \abs{\rd_{u} \log \dur} \bb) < C_{0} \eps. \label{Chr.BV.Thm.du}
\end{gather}
\end{theorem}

\begin{remark} 
We remark that in \cite[Theorem 6.2]{Christodoulou:1993bt}, it is implicitly assumed that $\phi(1, 1) = 0$; see \cite[Section 4]{Christodoulou:1993bt}. 
Note, however, that the bounds in the above theorem are stated in such a way that they are invariant under the transform $(\phi, r, m) \mapsto (\phi + c, r, m)$, under which \eqref{eq:SSESF} is also invariant. Any solution may be then transformed to satisfy $\phi(1, 1) = 0$. As a consequence, we do not need to check $\phi(1,1) = 0$ in order to apply the theorem.
\end{remark}
Using Theorem \ref{Chr.BV.Thm}, we prove the following bound for locally BV scattering solution to \eqref{eq:SSESF}.

\begin{theorem} \label{thm:decayInCpt}
Let $(\phi, r, m)$ be a locally BV scattering solution to \eqref{eq:SSESF}. For every $\eps > 0$, there exists $u_{0} > 1$ such that the following estimate holds.
\begin{align*}
	\sup_{v \in [u_{0}, \infty)} \bb( \int_{\uC_{v} \cap \set{u \geq u_{0}}\cap \cpt} \abs{\rd_{u}^{2} (r \phi)} 
	+ \int_{\uC_{v} \cap \set{u \geq u_{0}} \cap \cpt} \abs{\rd_{u} \phi} 
	+ \int_{\uC_{v} \cap \set{u \geq u_{0}} \cap \cpt} \abs{\rd_{u} \log \dur} \bb) < \eps.
\end{align*}

Moreover, we also have
\begin{equation} \label{eq:bnd4durphi}
	\sup_{\PD} \abs{\rd_{u} (r \phi)} \leq C_{K, \Lmb} \Psi.
\end{equation}
\end{theorem}

\begin{proof}
We first show that for a locally BV scattering solution to \eqref{eq:SSESF},
\begin{equation*}
	\int_{C_{u} \cap \cpt} \abs{\rd_{v} (\dvr^{-1} \rd_{v} (r \phi))} \to 0 \hbox{ as } u \to \infty,
\end{equation*}
Expanding this expression, we have
\beaa
\int_{C_{u} \cap \cpt} \abs{\rd_{v} (\dvr^{-1} \rd_{v} (r \phi))}
\leq \int_{C_{u} \cap \cpt} \dvr^{-1} (\abs{ \rd_{v}^2 (r \phi)}+\abs{ (\rd_{v} \log \lambda) \rd_v(r \phi)})
\eeaa
By \eqref{eq:bnd4dvr} and \eqref{eq:bnd4dvrphi}, we have
$$\int_{C_{u} \cap \cpt} \abs{\rd_{v} (\dvr^{-1} \rd_{v} (r \phi))}\leq 
C_{\Lambda, \Psi}\int_{C_{u} \cap \cpt} \abs{ \rd_{v}^2 (r \phi)}+\abs{ \rd_{v} \log \lambda },$$
which by \eqref{eq:locBVScat} in Definition \ref{def:locBVScat} (Scattering in BV in a compact $r$-region) tends to 0 as $u\to \infty$.
Notice that the quantity $\int_{C_{u} \cap \cpt} \abs{\rd_{v} (\dvr^{-1} \rd_{v} (r \phi))}$ which we have controlled is invariant under any rescaling of the coordinate $v$, and also under the transform $(\phi, r, m) \to (\phi + c, r, m)$.

We now proceed to the proof of the theorem. Let $v_0$ be sufficiently large and $u^{\star}(v_0)$ be the unique $r(u^{\star}(v_0),v_0)=R$. By the finite speed of propagation of the equations, the solution on $\uC_{v_0}\cap\cpt$ depends only on the data on $C_{u^{\star}(v_0)} \cap\cpt$.

In order to apply Theorem \ref{Chr.BV.Thm}, we change coordinates $(u,v)\mapsto (U(u),V(v))$ in the region bounded by $C_{u^{\star}(v_0)}$ and $\uC_{v_0}$ to a new double null coordinate $(U,V)$ such that for $U^{\star}=U(u^{\star}(v_0))$, we have $\lambda(U^{\star},V)=\frac 12$. To this end, define $V(v)$ by
$$\frac{dV}{dv}=2\dvr(u^{\star}(v_0),v),\quad V(v_0)=v_0.$$
Notice that this is acceptable since $\dvr>0$. In order for the condition $U = V$ to hold on $\Gmm$, we require
$U(u)=V(u).$
Then with respect to the coordinate $V$
$$\rd_V r(U^{\star},V)=\frac 12.$$
By \eqref{eq:bnd4dvr}, we have
$$\Lambda^{-1}\leq \frac{dV}{dv}, \frac{dU}{du}\leq \frac 12.$$
Moreover,
$$\int_{u^{\star}(v_0)}^{v_0} |\frac{d^2V}{dv^2}(v')|dv'\leq 2\int_{u^{\star}(v_0)}^{v_0} |\rd_v\lambda(u^{\star}(v_0),v')| \ud v',$$
which tends to $0$ as $v_0\to \infty$ by the assumption of local BV scattering. For $v_0$ sufficiently large, in the $(U,V)$ coordinate system, $\int_{C_{u^{\star}(v_0)} \cap \cpt} \abs{\rd_{V} ((\rd_{V} r)^{-1} \rd_{V} (r \phi))} dV$ is small and $\rd_V r=\frac 12$. The data satisfy the assumptions of Theorem \ref{Chr.BV.Thm} and therefore\footnote{More precisely, we apply Theorem \ref{Chr.BV.Thm} to the truncated initial data $$ \rd_{V} (r \widetilde{\phi})(U^{\star}, V) = \left\{ \begin{array}{cc} \rd_{V}(r \phi)(U^{\star}, V) & \hbox{ for } V < v_{0} \\ \rd_{V}(r \phi)(U^{\star}, v_{0}) & \hbox{ for } V \geq v_{0} \end{array} \right.$$}
$$\int_{\uC_{v_0}\cap\cpt}(\abs{\rd_{U} ((\rd_{U} r)^{-1} \rd_{U} (r \phi))} + \abs{\rd_{U} \phi} + \abs{\rd_{U} \log \rd_{U} r} ) dU \to 0$$
as $v_{0} \to \infty$. Returning to the original coordinate system $(u,v)$, the first statement easily follows.

Finally, for the $L^\infty$ estimate for $\rd_{u}(r\phi)$, notice that $\abs{\rd_{u}(r\phi)} \leq \Psi$ at the axis by \eqref{eq:bnd4dvrphi} and $(7)$ of Definition \ref{def:BVsolution} (BV solutions to \eqref{eq:SSESF}). Using \eqref{eq:SSESF:dphi''}, \eqref{eq:SSESF:dr} (in particular, the fact that $\rd_{v} \dur \leq 0$), \eqref{eq:bnd4phi} and \eqref{eq:bnd4dur}, we have
\begin{equation*}
\abs{\rd_{u}(r\phi)(u,v)} \leq \Psi+ \Lmb \Psi \int_{u}^{v} (-\rd_{v} \dur) \, \ud v' \leq C_{K, \Lmb} \Psi. \qedhere
\end{equation*}
\end{proof}

\section{Decay of $\phi$ and its first derivatives}\label{sec.decay1}
In this section, we prove the first main theorem (Theorem \ref{main.thm.1}). Throughout this section, we assume that $(\phi, r, m)$ is a locally BV scattering solution to \eqref{eq:SSESF} with asymptotically flat initial data of order $\omg'$ in BV, as in Definitions \ref{def:locBVScat} and \ref{def:AF}, respectively. Let $\omg = \min \set{\omg', 3}$.

\subsection{Preparatory lemmas}
The following lemma will play a key role in the proof of both Theorems \ref{main.thm.1} and \ref{main.thm.2}. It is a consequence of the scattering assumption \eqref{eq:locBVScat} and vanishing of the final Bondi mass.

\begin{lemma} \label{lem:smallPtnl}
Let $\eps > 0$ be an arbitrary positive number. For $u_{1} > 1$ sufficiently large, we have
\begin{align} 
	\sup_{v \in [u_{1}, \infty)} \int_{\uC_{v} \cap \set{u \geq u_{1}}} \abs{\frac{2m \dur}{(1-\mu) r^{2}}}   < \eps, \label{eq:smallPtnl:u} \\ 
	\sup_{u \in [u_{1}, \infty)} \int_{C_{u}} \abs{\frac{2m \dvr}{(1-\mu) r^{2}}} < \eps. \label{eq:smallPtnl:v}
\end{align}
\end{lemma}

\begin{proof} 
The first statement \eqref{eq:smallPtnl:u} was proved in Lemma \ref{lem:babySmllPtnl:2}; thus it only remains to prove \eqref{eq:smallPtnl:v}.

Divide $\PD$ into $\cpt =\PD\cap\set{r \leq R}$ and $\cpt^{c} := \PD \setminus \cpt$. First, note that by \eqref{eq:locBVScat} we have
\begin{align*}
	\sup_{u \in [u_{1}, \infty)} \int_{C_{u} \cap \cpt} \abs{\frac{2m \dvr}{(1-\mu) r^{2}}} < \eps/2
\end{align*}
for $u_{1}$ sufficiently large. 
Next, we consider $\cpt^{c}$. Define $v^{\star}(u) := \sup \set{v \in [u, \infty) : r(u,v) \geq R}$; note that $r(u, v^{\star}(u)) = R$ by continuity. We now compute
\begin{align*}
	\int_{C_{u} \cap \cpt^{c}} \abs{\frac{2 m \dvr}{(1-\mu) r^{2}}}  
	= & \int_{v^{\star}(u)}^{\infty} \abs{\frac{2m \dvr}{(1-\mu) r^{2}}(u,v')} \, \ud v' \\
	\leq & 2 K \Lmb M(u_{1}) \int_{v^{\star}(u)}^{\infty} \frac{\dvr}{r^{2}}(u, v') \, \ud v' \\
	\leq & 2 R^{-1} K \Lmb M(u_{1}).
\end{align*}
uniformly in $u \geq u_{1}$. As $\lim_{u_{1} \to \infty} M(u_{1}) = 0$ by \eqref{eq:zeroMf} (vanishing final Bondi mass), the last line can be made arbitrarily small by taking $u_{1}$ sufficiently large. This proves \eqref{eq:smallPtnl:v}. \qedhere
\end{proof}

The following lemma allows us to estimate $\phi$ in terms of $\abs{\rd_{v} (r \phi)}$.
\begin{lemma} \label{lem:intEst4phi}
The following estimates hold.
\begin{align*} 
	\abs{\phi}(u,v) \leq & \Lmb \sup_{C_{u} } \abs{\rd_{v} (r \phi)} , \\
	r u^{\om-1} \abs{\phi}(u,v) \leq & C \Lmb \bb( \sup_{C_{u}} u^{\om} \abs{\rd_{v} (r \phi)} + \sup_{C_{u}} r^{\om} \abs{\rd_{v} (r \phi)} \bb). 
\end{align*}
\end{lemma}

\begin{proof} 
The first estimate follows from Lemma \ref{lem:est4phi} and Proposition \ref{prop:geomLocBVScat}. The second estimate is a consequence of the first when $r(u,v) \leq u$, so it suffices to assume $r(u,v) \geq u$. Introducing a parameter $v_{1} \in [u, v]$, we estimate
\begin{align*}
	r u^{\om-1} \abs{\phi}(u,v) 
	\leq & u^{\om-1} \int_{u}^{v} \abs{\rd_{v} (r \phi)(u, v')} \, \ud v' \\
	\leq & \Lambda u^{\om-1}(\sup_{C_u}|\rd_v(r\phi)|)\int_u^{v_1}\dvr(u,v')\ud v'+\Lambda u^{\om-1}(\sup_{C_u}r^{\om}|\rd_v(r\phi)|)\int_{v_1}^{v}\frac{\dvr}{r^{\om}}(u,v')\ud v'\\
	\leq & \Lmb (r(u, v_{1})/u) \sup_{C_{u}} u^{\om} \abs{\rd_{v} (r \phi)} + \frac{\Lmb}{\omg-1} (u^{\om-1}/r(u, v_{1})^{\om-1}) \sup_{C_{u}} r^{\om} \abs{\rd_{v} (r \phi)}.
\end{align*}

Choosing $v_{1}$ so that $r(u, v_{1}) = u$ (which is possible since $r(u,v) \geq u$), the desired estimate follows.
\end{proof}

\subsection{Preliminary $r$-decay for $\phi$} \label{subsec:decay1:rDecay}
In this subsection, we derive bounds for $\phi$ which are sharp in terms of $r$-weights. As a consequence, they give sharp decay rates towards null infinity.

\begin{lemma} \label{lem:decay1:cptu:0}
	There exists a constant $0 < H_{1} < \infty$ such that the following estimate holds.
	\begin{equation} \label{eq:decay1:cptu:0}
		\sup_{\PD} (1+r) \abs{\phi}   \leq H_{1}.
	\end{equation}
\end{lemma}

\begin{proof} 
Let $r_{1} > 0$ be a large number to be chosen below. Different arguments will be used in $\set{r \geq r_{1}}$ and $\set{r \leq r_{1}}$. For each $u \geq 1$ let $v^{\star}_{1}(u)$ be the unique $v$-value for which $r(u, v_{1}^{\star}(u)) = r_{1}$. By the fundamental theorem of calculus, we have
\begin{equation} \label{eq:decay1:cptu:0:pf:1}
	r \phi = r_{1} \phi(u, v^{\star}_{1} (u)) + \int_{v^{\star}_{1}(u)}^{v} \rd_{v} (r \phi) (u, v') \, \ud v'.
\end{equation} 

Integrate \eqref{eq:SSESF:dphi} along the incoming direction from $(1,v)$ to $(u,v)$. By Corollary \ref{cor:mntn4Bondi} and Proposition \ref{prop:geomLocBVScat}, we have
\begin{align*}
	\abs{\rd_{v} (r \phi)(u,v)}
	\leq& \abs{\rd_{v} (r \phi)(1, v)} + \abs{\int_{1}^{u} \frac{2m \dvr \dur}{(1-\mu) r^{3}} (r\phi) (u', v) \, \ud u'} \\
	\leq &\abs{\rd_{v} (r \phi)(1, v)} + \frac{K \Lmb M_{i}}{2} \frac{1}{r^{2}(u,v)}   \sup_{u' \in [1, u]} \abs{r \phi(u', v)}.
\end{align*}

Substituting the preceding bound into \eqref{eq:decay1:cptu:0:pf:1}, we obtain
\begin{equation} \label{eq:decay1:cptu:0:pf:2}
\begin{aligned}
	\sup_{C_{u} \cap \set{r \geq r_{1}}} \abs{r \phi} 
	\leq & \abs{r_{1} \phi(u, v^{\star}_{1} (u))} + \int_{v_{1}^{\star}(u)}^{v} \abs{\rd_{v} (r \phi)(1, v')} \, \ud v' \\
	& + \frac{K \Lmb^{2} M_{i}}{2 r_{1}}  \sup_{u' \in [1, u]} \sup_{C_{u'} \cap \set{r \geq r_{1}}} \abs{r \phi}.
\end{aligned}
\end{equation}

The first term on the right-hand side is bounded by $r_{1} \Lmb \Psi$ by \eqref{eq:bnd4phi}, whereas the second term depends only on the initial data and can be estimated in terms of $\calI_{1}$ as follows:
\begin{equation*}
	\int_{v_{1}^{\star}(u)}^{v} \abs{\rd_{v} (r \phi)(1, v')} \, \ud v' \leq \Lmb \calI_{1} \int_{1}^{\infty}  (1+r(1, v'))^{-\om'} \dvr(1,v') \, \ud v' \leq \frac{\Lmb}{\om'-1} \calI_{1}.
\end{equation*}

Moreover, choosing $r_{1}$ to be large enough so that 
\begin{equation*}
\frac{K \Lmb^{2}  M_{i}}{2 r_{1}} \leq \frac{1}{2},
\end{equation*}
the last term of \eqref{eq:decay1:cptu:0:pf:2} can be absorbed in to the left-hand side and we conclude
\begin{equation*}
	\sup_{\set{r \geq r_{1}}} \abs{r \phi} \leq 2 r_{1} \Lmb \Psi + \frac{2}{\om'-1}\Lmb \calI_{1}.
\end{equation*}

On the other hand, in $\set{r \leq r_{1}}$ we have
\begin{equation*}
	\sup_{\set{r \leq r_{1}}} \abs{r \phi} \leq r_{1} \Lmb \Psi
\end{equation*}
by \eqref{eq:bnd4phi}. Combining the bounds in $\{r\geq r_1\}$ and $\{r\leq r_1\}$, the lemma follows. \qedhere
\end{proof}

\begin{remark} 
The preceding argument shows that Lemma \ref{lem:decay1:cptu:0} holds with\footnote{Notice that while the constant $C_{\calI_{1}, K, \Lmb}$ depends on $\calI_1$, the preceding argument moreover allows us to choose $C_{\calI_{1}, K, \Lmb}$ to be non-decreasing in $\calI_1$. In particular, \emph{for $\calI_1$ sufficiently small}, we have $H_{1} \leq C_{K, \Lmb} \, (\calI_{1} + \Psi)$. It is for this reason that we prefer to write the expression $C_{\calI_{1}, K, \Lmb} \, (\calI_{1} + \Psi)$ instead of the more general $C_{\calI_{1}, K, \Lmb, \Psi}$.}
\begin{equation} \label{eq:decay1:H1}
	H_{1} \leq C_{\calI_{1}, K, \Lmb} \, (\calI_{1} + \Psi).
\end{equation}
\end{remark}

\subsection{Propagation of $u$-decay for $\rd_{u} (r \phi)$}
Here, we show that $u$-decay estimates proved for $\rd_{v} (r \phi)$ and $\phi$ may be `transferred' to $\rd_{u} (r \phi)$; this reduces the proof of Theorem \ref{main.thm.1} to showing only \eqref{eq:decay1:1} and \eqref{eq:decay1:2}. To this end, we integrate $\rd_{v} \rd_{u} (r \phi)$ from the axis $\Gmm$, along which $\rd_{u} (r \phi) = - \rd_{v} (r \phi)$. 

\begin{lemma} \label{lem:decay1:uDecay4durphi}
Suppose that there exists a finite positive constant $A$ such that 
\begin{equation*}
	\sup_{\PD} \abs{\phi} \leq A u^{-\om}, \qquad 
	\sup_{\PD} \abs{\rd_{v} (r \phi)} \leq A u^{-\om}.
\end{equation*}

Then the following estimate holds.
\begin{equation*}
	\sup_{\PD} \abs{\rd_{u} (r \phi)} \leq (1+K)A u^{-\om}.
\end{equation*}
\end{lemma}

\begin{proof}
Fix $u \geq 1$ and $v \geq u$. Integrate \eqref{eq:SSESF:dphi''} along the outgoing direction from $(u, u)$ to $(u, v)$ and take the absolute value. Using $(7)$ of Definition \ref{def:BVsolution} (BV solutions to \eqref{eq:SSESF}), \eqref{eq:SSESF:dr} (in particular, $\rd_{v} \dur\leq 0$), \eqref{eq:bnd4dur} and the hypotheses, we have
\begin{align*}
	\abs{\rd_{u}(r\phi)(u,v)} 
	\leq & \lim_{v' \to u+} \abs{\rd_{v} (r \phi)(u, v')} + \sup_{u \leq v' \leq v} \abs{\phi(u,v')} \int_{u}^{v} (-\rd_{v} \dur) \, \ud v' \\
	\leq & A u^{-\omg} + K A u^{-\omg}. \qedhere
\end{align*}
\end{proof}

\subsection{Full decay for $\phi$ and $\rd_{v} (r \phi)$}\label{sec.full.decay.1}
In this subsection, we finish the proof of Theorem \ref{main.thm.1}. By Lemma \ref{lem:decay1:uDecay4durphi}, it suffices to establish the full decay of $\phi$ and $\rd_{v} ( r \phi)$, i.e., \eqref{eq:decay1:1} and \eqref{eq:decay1:2}. For the convenience of the reader, we recall these estimates below:
\begin{align*}
		\abs{\phi} \leq & A \min \set{u^{-\om}, r^{-1} u^{-(\om-1)}}, \tag{\ref{eq:decay1:1}} \\
		\abs{\rd_{v} (r \phi)} \leq & A \min \set{u^{-\om}, r^{-\om}}. \tag{\ref{eq:decay1:2}}
\end{align*}

For $U > 1$, let
\begin{equation*}
\calB_{1}(U) := \sup_{u \in [1, U]} \sup_{C_{u}} \bb( u^{\om} \abs{\phi} + r u^{\om-1} \abs{\phi}  \bb).
\end{equation*}

Notice that this is finite for every fixed $U$ by Lemma \ref{lem:decay1:cptu:0}. To establish the decay estimate \eqref{eq:decay1:1}, it suffices to prove that $\calB_{1}(U)$ is bounded by a finite constant which is \emph{independent of $U$}. We will show that this implies also \eqref{eq:decay1:2}. Divide $\PD$ into $\extr \cup \intr$, defined by
\begin{equation*}
	\extr := \set{(u,v) \in \PD : v \geq 3u}, \quad \intr := \set{(u,v) \in \PD : v \leq 3u}.
\end{equation*}

We first establish a bound for $\rd_{v} (r \phi)$ with the sharp $r$-weight, which thus gives the sharp decay rate in $\extr$. 
\begin{lemma} \label{lem:decay1:extr}
Let $u_{1} > 1$. Then for $u_{1}\leq u\leq U$, the following estimate holds.
\begin{equation} \label{eq:decay1:extr}
	\sup_{C_{u}} r^{\om} \abs{\rd_{v} (r \phi)} \leq \calI_{1} +   C_{K, M_{i}} \, u_{1} H_{1} + C u_{1}^{-1} K M_{i} \, \calB_{1}(U).
\end{equation}
\end{lemma}

\begin{proof} 
We separate the proof into cases $\om \geq 2 $ and $1<\om\leq 2$.

\noindent{\bf Case 1: $\om\geq 2$}

First, notice that
$$|\phi|\leq \calB_1(U) (r^{-1}u^{-(\om-1)})^{\om-2} (u^{-\om})^{1-(\om-2)}\leq \calB_1(U) r^{-(\om-2)}u^{-2}.$$
Applying Lemma \ref{lem:decay1:cptu:0}, we also have
$$|\phi|\leq (1+r)^{-1}H_1.$$
By Corollary \ref{cor:mntn4Bondi} and Proposition \ref{prop:geomLocBVScat}, we have the following pointwise bounds:
\begin{equation*}
\sup_{u'\in [1,u_1]} |\frac{m\dvr\dur}{1-\mu}|\leq \frac{K M_i}{2}\, , \quad
\sup_{u'\in [u_{1},\infty)} |\frac{m\dvr\dur}{1-\mu}|\leq \frac{K M(u_{1})}{2}\, .
\end{equation*}
Therefore, integrating \eqref{eq:SSESF:dphi} along the incoming direction from $(1, v)$ to $(u, v)$, we have
\begin{align*}
	&\abs{\rd_{v} (r \phi)(u,v)}\\
	& \quad \leq \abs{\rd_{v} (r \phi)(1, v)} + \abs{\int_{1}^{u} \frac{2m \dvr \dur\phi}{(1-\mu) r^{2}}  (u', v) \, \ud u'} \\
	& \quad \leq \abs{\rd_{v} (r \phi)(1, v)} + \frac{K M_{i}}{r^2(u,v)(1+r(u,v))} H_{1} \int_{1}^{u_{1}}  \, \ud u' + \frac{K M(u_1)}{r^{\om}(u,v)} \calB_{1}(U) \int_{u_{1}}^{u} (u')^{-2} \, \ud u' \\
	& \quad \leq \abs{\rd_{v} (r \phi)(1, v)} + \frac{u_{1} K M_{i}}{r^2(u,v)(1+r(u,v))} H_{1} + \frac{K M(u_{1})}{u_{1} r^{\om}(u,v)} \calB_{1}(U).
\end{align*}

Multiplying both sides by $r^{\om}(u,v)$ and using the fact that $r(u,v) \leq r(1, v)$, we conclude
\begin{align*}
	r^{\om} \abs{\rd_{v} (r \phi)}(u,v) 
	\leq & r^{\om}\abs{\rd_{v} (r \phi)}(1,v) + u_{1} \frac{r^{\om-2}}{(1+r)} K M_{i} \, H_{1}+ u_{1}^{-1} K M(u_{1}) \, \calB_{1}(U) \\
	\leq & \calI_{1}+ C_{u_{1}, K, M_{i}}  H_{1} + u_{1}^{-1} K M_{i} \, \calB_{1}(U).
\end{align*}

\noindent{\bf Case 2: $1<\om\leq 2$}

We will use the following bounds for $\phi$. First, 
$$|\phi|\leq \calB_1(U) (r^{-1}u^{-(\om-1)})^{\om-1} (u^{-\om})^{(2-\om)}\leq \calB_1(U) r^{-(\om-1)}u^{-1}.$$

Also, Lemma \ref{lem:decay1:cptu:0} implies
$$|\phi|\leq (1+r)^{-1}H_1.$$

As in Case 1 we integrate \eqref{eq:SSESF:dphi} along the incoming direction from $(1, v)$ to $(u, v)$:
\begin{align*}
	&\abs{\rd_{v} (r \phi)(u,v)}\\
	& \quad \leq \abs{\rd_{v} (r \phi)(1, v)} + \abs{\int_{1}^{u} \frac{2m \dvr \dur\phi}{(1-\mu) r^{2}}  (u', v) \, \ud u'} \\
	& \quad \leq \abs{\rd_{v} (r \phi)(1, v)} + \frac{K \Lmb M_{i} H_{1}}{(1+r)} \int_{1}^{u_{1}} \frac{-\dur}{r^2}  \, \ud u' 
								+ \frac{K \Lmb M(u_1)}{u_1} \calB_{1}(U) \int_{u_{1}}^{u}  \frac{-\dur}{r^{\om+1}} \, \ud u' \\
	& \quad \leq \abs{\rd_{v} (r \phi)(1, v)} + \frac{\om K \Lmb M_{i}}{r(u,v)(1+r(u,v))} H_{1} + \frac{\om K \Lmb M(u_{1})}{u_{1} r^{\om}(u,v)} \calB_{1}(U).
\end{align*}

Multiply both sides by $r^{\om}$ to arrive at the conclusion as in Case 1. In this case, note that the second term is a bit better than what is claimed, as there is no dependence on $u_{1} \geq 1$. \qedhere
\end{proof}

\begin{remark} 
Note that the proof of this lemma limits $\omg$ to be $\leq 3$. More precisely, this limitation comes from the contribution of the right-hand side of \eqref{eq:SSESF:dphi}
\end{remark}

We are now ready to prove bounds \eqref{eq:decay1:1} and \eqref{eq:decay1:2}. The idea is  to `propagate' the exterior decay estimate \eqref{eq:decay1:extr} into $\intr$ to obtain decay in $u$, using the smallness coming from Lemma \ref{lem:smallPtnl} in the region where $u$ is sufficiently large. On the other hand, the preliminary $r$-decay estimates proved in \S \ref{subsec:decay1:rDecay} will give the desired $r$-decay rates in rest of the space-time.

\begin{proof}[Proof of \eqref{eq:decay1:1} and \eqref{eq:decay1:2}] 
Let $1 \leq u_{1} \leq U$. For $(u, v) \in \PD$ with $u \in [3 u_{1}, U]$, integrate \eqref{eq:SSESF:dphi} along the incoming direction from $(u/3, v)$ to $(u, v)$. Then
\begin{equation} \label{eq:decay1:intr:pf:1}
\begin{aligned}
	\abs{\rd_{v} (r \phi) (u,v)} \leq 
	& \abs{\rd_{v} (r \phi) (u/3,v)} \\
	& + \frac{1}{2} (\sup_{u' \in [u/3, u]} \sup_{C_{u'}} \abs{\phi}) \int_{u/3}^{u} \abs{\frac{2m \dur}{(1-\mu) r^{2}} (u', v)} \, \ud u'.
\end{aligned}
\end{equation} 

Multiply both sides by $u^{\om}$ and estimate each term on the right-hand side.  For the first term, the key observation is the following: For $v \geq u$, the point $(u/3, v)$ lies in $\extr$, where \eqref{eq:decay1:extr} is effective. Indeed, note that
\begin{equation*}
	(2/3\Lmb)  u \leq \Lmb^{-1} ( v- (u/3) ) \leq r(u/3, v).
\end{equation*}

Thus, by \eqref{eq:decay1:extr},
\begin{align*}
	u^{\om} \abs{\rd_{v} (r \phi) (u/3,v)} 
	\leq & (3 \Lmb /2)^{\om} \bb( r^{\om}(u/3, v) \abs{\rd_{v} (r \phi) (u/3,v)} \bb) \\
	\leq & (3 \Lmb/2)^{\om} \bb( \calI_{1} +   C_{u_{1}, K, M_{i}}  H_{1} + C u_{1}^{-1} K M_{i} \, \calB_{1}(U) \bb) \\
	\leq & C_{u_{1}, K, \Lmb, M_{i}} (\calI_{1} + H_{1}) + C_{K, \Lmb} M_{i} u_{1}^{-1} \, \calB_{1}(U).
\end{align*}

For the second term on the right-hand side of \eqref{eq:decay1:intr:pf:1}, we have
\begin{align*}
\frac{u^{\om}}{2} (\sup_{u' \in [u/3, u]} \sup_{C_{u'}} \abs{\phi}) \int_{u/3}^{u} \abs{\frac{2m \dur}{(1-\mu) r^{2}} (u', v)} \, \ud u'
\leq & \frac{3^{\om}}{2} \bb( \int_{u/3}^{u} \abs{\frac{2m \dur}{(1-\mu) r^{2}} (u', v)} \, \ud u' \bb) \calB_{1}(U).
\end{align*}

Combining these estimates, we deduce
\begin{equation} \label{eq:decay1:intr:pf:2}
\begin{aligned}
	\sup_{C_{u}} u^{\om} \abs{\rd_{v}(r \phi)(u, v)} \leq & C_{u_{1}, K, \Lmb, M_{i}} (\calI_{1} + H_{1}) \\ 
	&+ \bb( C_{K, \Lmb} M_{i} u_{1}^{-1} + C \int_{u/3}^{u}  \abs{\frac{2m \dur}{(1-\mu) r^{2}} (u', v)} \, \ud u' \bb) \, \calB_{1}(U).
\end{aligned}
\end{equation}

Recalling the bounds of $\phi$ in terms of $\rd_v(r\phi)$ in Lemmas \ref{lem:intEst4phi}, we have
\begin{align*}
	\calB_{1}(U) 
	\leq & (1+2\Lmb) \sup_{u \in [1, U]} \sup_{C_{u}} \bb( u^{\om} \abs{\rd_{v}(r \phi)} + r^{\om} \abs{\rd_{v} (r \phi)} \bb).
\end{align*}
The right-hand side can be controlled by \eqref{eq:decay1:intr:pf:2} and \eqref{eq:decay1:extr}, from which we conclude
\begin{equation} \label{eq:decay1:intr:pf:key}
	\calB_{1}(U) \leq C_{u_{1}, K, \Lmb, M_{i}} (\calI_{1} + H_{1}) + \bb( C_{K, \Lmb} M_{i} u_{1}^{-1} + C \int_{u/3}^{u}  \abs{\frac{2m \dur}{(1-\mu) r^{2}} (u', v)} \, \ud u'\bb) \calB_{1}(U).
\end{equation}

As a consequence of Lemma \ref{lem:smallPtnl}, the entire coefficient in front of $\calB_{1}(U)$ can made to be smaller than (say) $1/2$ by taking $u_{1}$ sufficiently large. Since $\calB_{1}(U) < \infty$, we can then absorb this term into the left-hand side. Observing that this bound is independent of $U > 1$, we have thus obtained \eqref{eq:decay1:1}.

To prove \eqref{eq:decay1:2}, simply apply \eqref{eq:decay1:intr:pf:2} and \eqref{eq:decay1:extr}, which shows that
\begin{align*}
&\sup_{u \in [1, U]} \sup_{C_{u}} \bb( u^{\om} \abs{\rd_{v}(r \phi)} + r^{\om} \abs{\rd_{v} (r \phi)} \bb)\\
& \quad \leq  C_{u_{1}, K, \Lmb, M_{i}} (\calI_{1} + H_{1}) + \bb( C_{K, \Lmb} M_{i} u_{1}^{-1} + C \int_{u/3}^{u}  \abs{\frac{2m \dur}{(1-\mu) r^{2}} (u', v)} \, \ud u'\bb) \calB_{1}(U).
\end{align*}
This boundedness of $\calB_1(U)$ that we just proved thus implies \eqref{eq:decay1:2}.
\qedhere
\end{proof}

\begin{remark} 
According to the proof that we have just given, the constant $A_{1} > 0$ depends on our choice of $u_{1} > 1$, which in turn depends on how fast the coefficient in front of $\calB_{1}(U)$ in \eqref{eq:decay1:intr:pf:key} vanishes as $u_{1} \to \infty$. This explains why $A_{1} > 0$ does not depend only on the size of the initial data, as remarked in Section \ref{sec.main.thm}. Controlling the size of $u_{1} > 1$ under an additional small data assumption will be key to proving Statement (1) of Theorem \ref{thm:smallData} in Section \ref{sec:smallData}.
\end{remark}

\subsection{Additional decay estimates}
We end this section with the following decay estimates for $\rd_{v} \phi$, $\rd_{u} \phi$ and $m$.
\begin{corollary} \label{cor:decay1}
Let $(\phi, r, m)$ be a locally BV scattering solution to \eqref{eq:SSESF} with asymptotically flat initial data of order $\omg'$ in BV, and define $\omg = \min \set{\omg', 3}$.
Let $A_{1}$ be the constant in Theorem \ref{main.thm.1}. Then the following decay estimates hold.
\begin{align} 
	\abs{\rd_{v} \phi} \leq & C A_{1} \min \set{r^{-1} u^{-\om}, r^{-2} u^{-(\om-1)}}, \label{eq:decay1:4} \\
	\abs{\rd_{u} \phi} \leq & C_{K} A_{1} \, r^{-1} u^{-\om}, \label{eq:decay1:5} \\
	m \leq & C_{\Lmb} A_{1}^{2} \min \set{r u^{-2\om}, u^{-(2\om-1)}}.  \label{eq:decay1:6}
\end{align}
\end{corollary}
\begin{proof} 
	Let $u \geq 1$ and $v \geq u$. Since
	\begin{equation*}
		r \rd_{v} \phi = \rd_{v} ( r \phi) - \dvr \phi, \qquad
		r \rd_{u} \phi = \rd_{u} ( r \phi) - \dur \phi,
	\end{equation*}
	the estimates \eqref{eq:decay1:4}, \eqref{eq:decay1:5} follow from \eqref{eq:decay1:1}--\eqref{eq:decay1:3} and the fact that $\sup_{\PD} \abs{\dvr} \leq 1/2$, $\sup_{\PD} \abs{\dur} \leq K$. 
	
	On the other hand, by \eqref{eq:SSESF:dm}, we have 
	\begin{equation} \label{eq:decay1:6:pf:1}
		m(u,v) = \frac{1}{2} \int_{u}^{v} \dvr^{-1} (1-\mu) r^{2} (\rd_{v} \phi)^{2} (u, v')\, \ud v'.
	\end{equation}
	
	Using $\abs{\rd_{v} \phi (u,v) } \leq C A_{1} r^{-1} u^{-\om}$ (which has just been established), we obtain
	\begin{equation*}
		m(u,v) \leq C_{\Lmb} A_{1}^{2} \, r u^{-2\om},
	\end{equation*}
	which proves a `half' of \eqref{eq:decay1:6}. 
	To prove the other `half', let us introduce a parameter $r_{1} > 0$ (to be determined later) and define $v_{1}^{\star}(u)$ to be the unique $v$-value such that $r(u, v^{\star}_{1}(u)) = r_{1}$. For $v \geq v^{\star}_{1}(u)$, divide the $v'$-integral in \eqref{eq:decay1:6:pf:1} into $\int_{u}^{v^{\star}_{1}(u)} + \int_{v^{\star}_{1}(u)}^{v}$ and use $\abs{\rd_{v} \phi (u,v) } \leq C A_{1} \, r^{-1} u^{-\om}$ for the former and $\abs{\rd_{v} \phi (u,v) } \leq C A_{1} \, r^{-2} u^{-(\om-1)}$ for the latter. As $m(u,v)$ is non-decreasing in $v$, we then arrive at the estimate
	\begin{equation*}
		\sup_{C_{u}} m \leq C_{\Lmb} A_{1}^{2} \, r_{1} u^{-2\om} + C_{\Lmb} A_{1}^{2} \, r_{1}^{-1} u^{-2(\om-1)}.
	\end{equation*}
	
Choosing $r_{1} = u$, we obtain \eqref{eq:decay1:6}. \qedhere
\end{proof}

\section{Decay of second derivatives}\label{sec.decay2}
In this section, we establish our second main theorem (Theorem \ref{main.thm.2}). Throughout the section, we assume that $(\phi, r, m)$ is a locally BV scattering solution to \eqref{eq:SSESF} with asymptotically flat initial data of order $\omg'$ in $C^{1}$, as in Definitions \ref{def:locBVScat} and \ref{def:AF}. As discussed in Remark \ref{rem:wp}, $(\phi, r, m)$ is then a $C^{1}$ solution to \eqref{eq:SSESF}. As before, let $\omg = \min\set{\omg', 3}$. 

\subsection{Preparatory lemmas}
The following lemma, along with Lemma \ref{lem:smallPtnl}, provides the crucial smallness for our proof of Theorem \ref{main.thm.2}.
\begin{lemma} \label{lem:smallDphi}
For every $\eps > 0$, there exists $u_{2} > 1$ such that
\begin{align}
	\sup_{v \in [u_{2}, \infty)} \int_{\uC_{v} \cap \set{u \geq u_{2}}} \abs{\rd_{u} \phi}  < \eps, \label{eq:smallDphi:u} \\ 
	\sup_{u \in [u_{2}, \infty)} \int_{C_{u}} \abs{\rd_{v} \phi} < \eps. \label{eq:smallDphi:v}
\end{align}
\end{lemma}
\begin{proof} 
We will only prove \eqref{eq:smallDphi:u}, leaving the similar proof of \eqref{eq:smallDphi:v} to the reader. As in the proof of Lemma \ref{lem:smallPtnl}, we divide $\PD$ into $\cpt:=\PD\cap\{r\leq R\}$ and $\cpt^{c} := \PD \setminus \cpt$, and argue separately. First, by Theorem \ref{thm:decayInCpt}, we have
\begin{equation*}
	\sup_{v \in [u_{2}, \infty)} \int_{\uC_{v} \cap \set{u \geq u_{2}} \cap \cpt} \abs{\rd_{u} \phi}  < \eps/2, 
\end{equation*}
for $u_{2}$ sufficiently large. Next, to derive \eqref{eq:smallDphi:u} in $\cpt^{c}$, we define $u^{\star}(v) := \sup \set{u \in [u_{2}, v] : r(u,v) \geq R}$, where we use the convention $u^{\star}(v) = u_{2}$ when the set is empty. Then using Proposition \ref{prop:geomLocBVScat} and Schwarz, we compute
\begin{align*}
	\int_{\uC_{v} \cap \set{u \geq u_{2}} \cap \cpt^{c}} \abs{\rd_{u} \phi} 
	= & \int_{u_{2}}^{u^{\star}(v)} \abs{\rd_{u} \phi(u',v)} \, \ud u' \\
	\leq & \sqrt{\frac{2 K \Lmb}{R}} \sqrt{\int_{u_{2}}^{u^{\star}(v)} \frac{1}{2}(-\dur)^{-1} (1-\mu) r^{2} (\rd_{u} \phi)^{2} (u', v) \, \ud u'} \\
	\leq & \sqrt{\frac{2 K \Lmb}{R} m(u_{2}, v)} \leq \sqrt{\frac{2 K \Lmb}{R} M(u_{2})}.
\end{align*}

By \eqref{eq:zeroMf} (Vanishing final Bondi mass), $\lim_{u_{2} \to \infty} M(u_{2}) = 0$. \eqref{eq:smallDphi:u} thus follows. \qedhere
\end{proof}

The next lemma allows us to estimate the first derivative of $\phi$ at $(u,v)$ in terms of information on $C_{u} \cap \set{(u, v'): u \leq v' \leq v}$.
\begin{lemma} \label{lem:dphi}
For every $(u,v) \in \PD$, the following inequalities hold.
\begin{align*}
	& \abs{\rd_{v} \phi(u,v)} \leq 			\frac{\Lmb^{2}}{4} \sup_{u \leq v' \leq v} \abs{\rd_{v}^{2}(r\phi)(u, v')} \\
	& \phantom{\abs{\rd_{v} \phi(u,v)} \leq} 	+ \frac{\Lmb^{3}}{4} \sup_{u \leq v' \leq v} \abs{\rd_{v} (r \phi)(u, v')} \sup_{u \leq v' \leq v} \abs{\rd_{v} \dvr (u, v')}, \\
	& \abs{\rd_{u} \phi(u,v)} \leq \Lmb  \sup_{u \leq v' \leq v} (- \dur)(u, v') \abs{\rd_{v} \phi(u,v')}.
\end{align*}
\end{lemma}

\begin{proof} 
The first is an easy consequence of \eqref{eq:est4dvphi:1} in \S \ref{subsec:est4phi}. To prove the second inequality, we start from the equation
\begin{equation*}
	\rd_{v} (r \rd_{u} \phi) = - \dur \rd_{v} \phi,
\end{equation*}
which follows from \eqref{eq:SSESF:dr} and \eqref{eq:SSESF:dphi}. Therefore, we have
\begin{align*}
	\abs{\rd_{u} \phi (u,v)} 
	\leq & \frac{1}{r(u,v)} \int_{u}^{v} (-\dur) \abs{\rd_{v} \phi} (u, v') \, \ud v',
\end{align*}
from which the second inequality easily follows. \qedhere
\end{proof}

In the next lemma, we show that improved estimates for $m$ near $\Gmm$ hold if we assume an $L^{\infty}$ control of ${\rd_{v} \phi}$.
\begin{lemma} \label{lem:muOverR}
For every $(u,v) \in \PD$, the following inequalities hold:
	\begin{align} 
	\frac{\mu}{r}(u,v) \leq & \Lmb^{2} \sup_{u \leq v' \leq v} \abs{\rd_{v} (r \phi)(u, v')} \sup_{u \leq v' \leq v} \abs{\rd_{v} \phi (u, v')}, \label{eq:muOverR:1} \\
	\frac{\mu}{r^{2}}(u,v) \leq & \frac{\Lmb^{2}}{3} \sup_{u \leq v' \leq v} \abs{\rd_{v} \phi(u, v')}^{2}. \label{eq:muOverR:2}
	\end{align}
\end{lemma}
\begin{proof} 
	Recall $\mu = 2m/r$. By \eqref{eq:SSESF:dm}, we have
	\begin{equation*}
		2 m(u,v) = \int_{u}^{v} (1-\mu) \dvr^{-1} r^{2} (\rd_{v} \phi)^{2} (u, v') \, \ud v'.
	\end{equation*}
	
	Pulling everything except $r^{2} \dvr$ outside the integral and using $\int_{u}^{v} r^{2} \dvr(u, v') \, \ud v' = (1/3) r^{3}(u,v)$, we obtain \eqref{eq:muOverR:2}. 
	On the other hand, using $\dvr^{-1} r \rd_{v} \phi = \dvr^{-1} \rd_{v} (r \phi) - \phi$ and $\int_{u}^{v} r \dvr(u, v') \, \ud v' = (1/2) r^{2}(u,v)$, we easily deduce
\begin{equation*}
	\frac{\mu}{r}(u,v) \leq \frac{1}{2} \sup_{u \leq v' \leq v} \bb( \Lmb^{2} \abs{\rd_{v} (r \phi)(u, v')} + \Lmb \abs{\phi(u, v')} \bb) \abs{\rd_{v} \phi(u, v')}.
\end{equation*}
	
	From the fact that $\abs{\phi(u,v)} \leq \Lmb \sup_{u \leq v' \leq v} \abs{\rd_{v} (r \phi)(u, v')}$, \eqref{eq:muOverR:1} easily follows.\qedhere
\end{proof}

\subsection{Preliminary $r$-decay for $\rd_{v}^{2} (r \phi)$ and $\rd_{v} \dvr$}
In this subsection, we establish decay estimates for $\rd_{v}^{2} (r \phi)$ and $\rd_{v} \dvr$ which are sharp in terms of $r$-weights in the region $\extr$. We remind the reader the decomposition $\PD = \extr \cup \intr$, where
\begin{equation*}
	\extr = \set{(u,v) \in \PD : v \geq 3u}, \quad \intr = \set{(u,v) \in \PD : v \leq 3u}.
\end{equation*}

In particular, note that $r \geq 2 \Lmb^{-1} u > 0$ in $\extr$.

\begin{lemma} \label{lem:decay2:rDecay}
The following estimates hold.
\begin{align}
	 \sup_{\extr} r^{3} \abs{\rd_{v} \dvr} \leq & C_{K, \Lmb} A_{1}^{2}, 	 \label{eq:decay2:rDecay:1} \\
	\sup_{\extr} r^{\om+1} \abs{\rd_{v}^{2} (r \phi)} \leq & C \calI_{2} + C_{K, \Lmb, M_{i}} A_{1}^{3}. \label{eq:decay2:rDecay:2}
\end{align}
\end{lemma}

\begin{proof} 
We begin by proving \eqref{eq:decay2:rDecay:1}. Recall \eqref{eq:eq4dvdvr:normal}:
\begin{equation*} \tag{\ref{eq:eq4dvdvr:normal}}
\rd_{u} \rd_{v} \log \dvr
= \frac{1}{(1-\mu)}  \dvr^{-1} \dur (\rd_{v} \phi)^{2} - \frac{4 m}{(1-\mu) r^{3}} \dvr \dur.
\end{equation*}

Note that $\rd_{v} \log \dvr= 0$ on $C_{1}$ by our choice of coordinates. Therefore, integrating the preceding equation along the incoming direction from $(1,v)$ to $(u,v)$, we have
\begin{equation*}
	\abs{\rd_{v} \log \dvr(u, v)} \leq \int_{1}^{u} \abs{\frac{1}{(1-\mu)}  \dvr^{-1} \dur (\rd_{v} \phi)^{2} (u', v)} \, \ud u' + \int_{1}^{u} \abs{\frac{4 m}{(1-\mu) r^{3}} \dvr \dur (u', v)} \, \ud u'.
\end{equation*}

Then \eqref{eq:decay2:rDecay:1} follows using Proposition \ref{prop:geomLocBVScat}, \eqref{eq:decay1:4} and \eqref{eq:decay1:6}. We remark that the power of $r$ is dictated by the second integral.

The proof of \eqref{eq:decay2:rDecay:2} is very similar. We start by recalling \eqref{eq:eq4dvdvrphi:normal}:
\begin{equation*} \tag{\ref{eq:eq4dvdvrphi:normal}}
\rd_{u} (\rd_{v}^{2} (r \phi)) = 
\frac{2m \dvr \dur}{(1-\mu) r^{2}} \, \rd_{v} \phi + \frac{ \dur}{(1-\mu) }  (\rd_{v} \phi)^{2} \phi 
 + \frac{2m \dur}{(1-\mu) r^{2}} (\rd_{v} \dvr) \phi - \frac{4m}{(1-\mu) r^{3}} \dvr^{2} \dur \phi.
\end{equation*}

For $u \geq 1$, we have $r(u,v) \leq r(1,v)$; moreover, by hypothesis, we have the estimate for the initial data term 
$$(1+r(1,v))^{\omg'+1} \abs{\rd_{v}^{2}(r \phi)(1,v)} \leq \calI_{2} \, .$$
Therefore, by the fundamental theorem of calculus, it suffices to bound
\begin{align*}
& \int_{1}^{u} \abs{\frac{2m \dvr \dur}{(1-\mu) r^{2}} \, \rd_{v} \phi (u', v)} \, \ud u' + \int_{1}^{u} \abs{\frac{ \dur}{(1-\mu) }  (\rd_{v} \phi)^{2} \phi (u', v)} \, \ud u' \\
& \quad + \int_{1}^{u} \abs{\frac{2m \dur}{(1-\mu) r^{2}} (\rd_{v} \dvr) \phi (u', v) } \, \ud u'  + \int_{1}^{u} \abs{\frac{4m}{(1-\mu) r^{3}} \dvr^{2} \dur \phi (u', v)} \, \ud u'
\end{align*}
by $C_{K, \Lmb, M_{i}} A_{1}^{3} r^{-(\om+1)}$. This is an easy consequence of Proposition \ref{prop:geomLocBVScat}, \eqref{eq:decay1:1}, \eqref{eq:decay1:4}, \eqref{eq:decay1:6} and also \eqref{eq:decay2:rDecay:1} that has just been established. Note that the last term is what limits $\omg \leq 3$. \qedhere
\end{proof}

\subsection{Propagation of $u$-decay for $\rd_{u}^{2} (r \phi)$ and $\rd_{u} \dur$}
Here, we show that certain $u$-decay for $\rd_{u}^{2} (r \phi)$ and $\rd_{u} \dur$ proved in $\intr$ can be propagated to $\PD$. The technique employed is very similar to that in the previous subsection.
\begin{lemma} \label{lem:decay2:uDecayInExtr}
For $U \geq 1$, suppose that there exists a finite positive constant $A, k_{1}, k_{2}$ such that
\begin{equation*}
	0 \leq k_{1} \leq 2\om+1, \quad
	0 \leq k_{2} \leq 3\om+1, 
\end{equation*}
and for $u \in [1, U]$, we have
\begin{equation*}
	\sup_{C_{u} \cap \intr} u^{k_{1}} \abs{\rd_{u} \dur} \leq A, \quad
	\sup_{C_{u} \cap \intr} u^{k_{2}} \abs{\rd_{u}^{2}(r\phi)} \leq A.
\end{equation*}

Then for $u \in [1, U]$, the following estimates hold.
\begin{align}
	\sup_{C_{u}} u^{k_{1}} \abs{\rd_{u} \dur} \leq & C_{K, \Lmb} A + C_{K, \Lmb} A_{1}^{2}, \label{eq:decay2:uDecayInExtr:1} \\
	\sup_{C_{u}} u^{k_{2}} \abs{\rd_{u}^{2}(r\phi)} \leq & A + C_{K, \Lmb} A_{1}^{3} + C_{K, \Lmb} A_{1}^{3} \, \sup_{C_{u}} u \abs{\rd_{u} \dur} . \label{eq:decay2:uDecayInExtr:2}
\end{align}

Furthermore, the following alternative to \eqref{eq:decay2:uDecayInExtr:2} also holds.
\begin{equation} \label{eq:decay2:uDecayInExtr:3}
\sup_{C_{u}} u^{k_{2}} \abs{\rd_{u}^{2}(r\phi)} 
	\leq A + C_{K, \Lmb} A_{1}^{3} + 
	C_{K, \Lmb} \Psi \int_{3u}^{\infty} \abs{\frac{2 m \dvr}{(1-\mu) r^{2}}}(u, v') \, \ud v' \cdot \sup_{C_{u}} u^{k_{2}} \abs{\rd_{u} \dur}. 
\end{equation}
\end{lemma}
\begin{proof}
Let us begin with \eqref{eq:decay2:uDecayInExtr:1}. Recall \eqref{eq:eq4dudur:normal}
\begin{equation*} \tag{\ref{eq:eq4dudur:normal}}
\rd_{v} \rd_{u} \log \dur
= \frac{1}{(1-\mu)}  \dvr \dur^{-1} (\rd_{u} \phi)^{2} - \frac{4m}{(1-\mu) r^{3}} \dvr \dur.
\end{equation*}

Given $(u,v) \in \extr$ (with $u \in [1,U]$), let us integrate this equation along the outgoing direction from $(u, 3u)$ to $(u,v)$, take the absolute value and multiply by $u^{k_{1}}$. Using the hypothesis
\begin{equation*}
	\sup_{\intr \cap \set{(u,v) \in \PD : u \in [1,U]}} u^{k_{1}} \abs{\rd_{u} \dur} \leq A,
\end{equation*}
\eqref{eq:decay2:uDecayInExtr:1} is reduced to showing
\begin{align} 
	u^{k_{1}} \int_{3u}^{\infty} \abs{\frac{1}{(1-\mu)} \dvr \dur^{-1} (\rd_{u} \phi)^{2} (u, v)} \, \ud v \leq & C_{K, \Lmb} A_{1}^{2}, \label{eq:decay2:uDecayInExtr:pf:1} \\ 
	u^{k_{1}} \int_{3u}^{\infty} \abs{\frac{4m}{(1-\mu) r^{3}} \dvr \dur (u, v)} \, \ud v \leq & C_{K, \Lmb} A_{1}^{2}, \label{eq:decay2:uDecayInExtr:pf:2}
\end{align}
for $u \in [1, U]$.

Using Proposition \ref{prop:geomLocBVScat} and \eqref{eq:decay1:5}, the left-hand side of \eqref{eq:decay2:uDecayInExtr:pf:1} is bounded by
\begin{equation*}
	C_{K, \Lmb} A_{1}^{2} \, u^{k_{1} - 2\om} \int_{3u}^{\infty} \frac{1}{r^{2}} \dvr \, \ud v 
	\leq C_{K, \Lmb} A_{1}^{2} \, u^{k_{1}- 2\om} r^{-1}(u, 3u).
\end{equation*}

As $u \geq 1$ and $r(u, 3u) \geq 2 \Lmb^{-1} u$, \eqref{eq:decay2:uDecayInExtr:pf:1} follows. Similarly, by \eqref{eq:bnd4mu} and \eqref{eq:decay1:6}, the left-hand side of \eqref{eq:decay2:uDecayInExtr:pf:2} is also bounded by  $C_{K, \Lmb} A_{1}^{2} \, u^{k_{1}- 2\om} r^{-1}(u, 3u)$, from which \eqref{eq:decay2:uDecayInExtr:pf:2} immediately follows.

Next, we turn to \eqref{eq:decay2:uDecayInExtr:2} and \eqref{eq:decay2:uDecayInExtr:3}; as they are proved similarly as before, we will only outline the main points. Recall \eqref{eq:eq4dudurphi:normal}:
\begin{equation*} \tag{\ref{eq:eq4dudurphi:normal}}
\rd_{v} (\rd_{u}^{2} (r \phi)) = 
\frac{2m \dvr \dur}{(1-\mu) r^{2}} \, \rd_{u} \phi + \frac{\dvr }{(1-\mu) }  (\rd_{u} \phi)^{2} \phi 
 + \frac{2m \dvr}{(1-\mu) r^{2}} (\rd_{u} \dur) \phi - \frac{4m}{(1-\mu) r^{3}} \dvr \dur^{2} \phi.
\end{equation*}

Fix $(u,v) \in \extr$ with $u \in [1,U]$. We then integrate the preceding equation along the outgoing direction from $(u, 3u)$ to $(u, v)$, take the absolute value and multiply by $u^{k_{2}}$. In order to prove \eqref{eq:decay2:uDecayInExtr:2}, in view of the hypothesis
\begin{equation*}
	\sup_{\intr \cap \set{(u, v) \in \PD : u \in [1, U]}} u^{k_{2}} \abs{\rd_{u}^{2} (r \phi)} \leq A,
\end{equation*}
it suffices to establish the following estimates for $u \in [1, U]$:
\begin{align*}
	u^{k_{2}} \int_{3u}^{\infty} \abs{\frac{2m \dvr \dur}{(1-\mu) r^{2}} \, \rd_{u} \phi (u,v)}\, \ud v 
	\leq & C_{K, \Lmb} A_{1}^{3}, \\
	u^{k_{2}} \int_{3u}^{\infty} \abs{\frac{\dvr }{(1-\mu) }  (\rd_{u} \phi)^{2} \phi  (u,v)}\, \ud v 
	\leq & C_{K, \Lmb} A_{1}^{3}, \\
	u^{k_{2}} \int_{3u}^{\infty} \abs{\frac{2m \dvr}{(1-\mu) r^{2}} (\rd_{u} \dur) \phi  (u,v)}\, \ud v 
	\leq & C_{K, \Lmb} A_{1}^{3} \, \sup_{\PD} u \abs{\rd_{u} \dur}, \\
	u^{k_{2}} \int_{3u}^{\infty} \abs{\frac{4m}{(1-\mu) r^{3}} \dvr \dur^{2} \phi (u,v)} \, \ud v
	\leq & C_{K, \Lmb} A_{1}^{3}.
\end{align*}

The proof of these estimates are similar to that of \eqref{eq:decay2:uDecayInExtr:pf:1}, \eqref{eq:decay2:uDecayInExtr:pf:2}; we omit the details. To prove \eqref{eq:decay2:uDecayInExtr:3}, we replace the third estimate by
\begin{equation*}
	u^{k_{2}} \int_{3u}^{\infty} \abs{\frac{2m \dvr}{(1-\mu) r^{2}} (\rd_{u} \dur) \phi  (u,v)}\, \ud v 
	\leq C_{K, \Lmb} \Psi \int_{3u}^{\infty} \abs{\frac{2 m \dvr}{(1-\mu) r^{2}} }(u, v') \, \ud v' \cdot \sup_{C_{u}} u^{k_{2}} \abs{\rd_{u} \dur},
\end{equation*}
which is an easy consequence of Proposition \ref{prop:geomLocBVScat}.
\qedhere
\end{proof}

\subsection{Full decay for $\rd_{v}^{2} (r \phi)$, $\rd_{u}^{2} (r \phi)$, $\rd_{v} \dvr$ and $\rd_{u} \dur$} \label{sec.full.decay.2}
With all the preparations so far, we are finally ready to prove Theorem \ref{main.thm.2}. Our proof consists of two steps.
The first step is use the local BV scattering assumption to prove a preliminary decay rate of $u^{-\omg}$ for $\rd_{v}^{2} (r \phi)$, $\rd_{u}^{2} (r \phi)$, $\rd_{v} \dvr$ and $\rd_{u} \dur$. In this step, it is crucial to pass to the \emph{renormalized variables} and exploit the null structure of \eqref{eq:SSESF}, in order to utilize the a priori bounds in the local BV scattering assumption. The second step to upgrade these decay rates to those that are claimed in Theorem \ref{main.thm.2}. Thanks to the preliminary $u^{-\omg}$ decay from the first step, the null structure is not necessary at this point.

We now begin with the first step. The null structure of \eqref{eq:SSESF} as demonstrated in \S \ref{subsec:nullStr} is used in an essential way.

\begin{proposition} \label{prop:decay2:nullStr}
There exists a finite constant $A_{2}' > 0$ such that the following estimates hold.
\begin{align*}
	& \abs{\rd_{v}^{2} (r \phi)} \leq A_{2}' u^{-\om}, \quad
	\abs{\rd_{u}^{2} (r \phi)} \leq A_{2}' u^{-\om},  \\
	& \abs{\rd_{v} \dvr} \leq A_{2}' u^{-\om},  \qquad
	\abs{\rd_{u} \dur} \leq A_{2}' u^{-\om}.
\end{align*}
\end{proposition}

\begin{proof} 
For $U > 1$, we define 
\begin{equation} \label{eq:decay2:def4B2}
	\calB_{2}(U) := \sup_{u \in [1, U]} \sup_{C_{u}} \bb(  u^{\om} \abs{\rd_{v}^{2} (r \phi)} +u^{\om} \abs{\rd_{u}^{2} (r \phi)} 
					+ u^{\om} \abs{\rd_{v} \dvr} + u^{\om} \abs{\rd_{u} \dur} \bb).
\end{equation}

Notice that the above is finite for every fixed $U$ due to Lemmas \ref{lem:decay2:rDecay} and \ref{lem:decay2:uDecayInExtr}. As indicated earlier, we need to use the null structure of the system \eqref{eq:SSESF} as in \S \ref{subsec:nullStr}. For convenience, we define the shorthands
\begin{align*}
	F_{1} := & \rd_{v}^{2} (r \phi) - (\rd_{v} \dvr) \phi, \\
	G_{1} := & \rd_{u}^{2} (r \phi) - (\rd_{u} \dur) \phi, 
\end{align*}
and
\begin{align*}
	F_{2} := & \rd_{v} \log \dvr - \frac{\dvr}{(1-\mu)} \frac{\mu}{r} + \rd_{v} \phi \bb( \dvr^{-1} \rd_{v} (r \phi) - \dur^{-1} \rd_{u} ( r \phi) \bb), \\
	G_{2} := & \rd_{u} \log (-\dur) - \frac{\dur}{(1-\mu)} \frac{\mu}{r} + \rd_{u} \phi \bb( \dvr^{-1} \rd_{v} (r \phi) - \dur^{-1} \rd_{u} (r \phi) \bb).
\end{align*}

Then \eqref{eq:eq4dvdvrphi}, \eqref{eq:eq4dudurphi}, \eqref{eq:eq4dvdvr} and \eqref{eq:eq4dudur} may be rewritten in the following fashion.
\begin{align} 
& \rd_{u} F_{1} = \rd_{u} \dvr \, \rd_{v} \phi - \rd_{v} \dvr \, \rd_{u} \phi, \label{eq:decay2:nullStr:pf:1} \\
& \rd_{u} F_{2} = \rd_{u} \phi \, \rd_{v}\bb( \dur^{-1} \rd_{u} (r \phi) \bb)- \rd_{v} \phi \, \rd_{u} \bb( \dur^{-1} \rd_{u} (r \phi) \bb),  \label{eq:decay2:nullStr:pf:2}  \\
& \rd_{v} G_{1} = \rd_{v} \dur \, \rd_{u} \phi - \rd_{u} \dur \, \rd_{v} \phi,\label{eq:decay2:nullStr:pf:3}  \\
& \rd_{v} G_{2} = - \rd_{u} \phi \, \rd_{v} \bb( \dvr^{-1} \rd_{v} (r \phi) \bb) + \rd_{v} \phi \, \rd_{u} \bb( \dvr^{-1} \rd_{v} (r \phi) \bb). \label{eq:decay2:nullStr:pf:4}
\end{align}

The following lemma is the key technical component of the proof.

\begin{lemma} \label{lem:decay2:key4nullStr}
There exists a finite positive constant $C = C_{A_{1}, \calI_{2}, K, \Lmb}$ and positive function $\eps(u)$ satisfying
\begin{equation*}
	\eps(u) \to 0 \hbox{ as } u \to \infty
\end{equation*}
such that the following inequalities holds for $1 \leq u_{2} \leq U$:
\begin{align}
	\sup_{\intr \cap \set{(u,v) \in \PD : u \in [3 u_{2}, U]}} \bb( u^{\om} \abs{F_{1}} + u^{\om} \abs{G_{1}} \bb)
	 \leq & C_{\Lmb} \calI_{2} + C_{K, \Lmb, M_{i}} A_{1}^{3} + \eps(u_{2}) \calB_{2}(U), \label{eq:decay2:key4nullStr:1} \\
	 \sup_{\intr \cap \set{(u,v) \in \PD : u \in [3 u_{2}, U]}} \bb( u^{\om} \abs{F_{2}} + u^{\om} \abs{G_{2}} \bb)
	  \leq & C_{K, \Lmb} A_{1}^{2} + \eps(u_{2}) \calB_{2}(U). \label{eq:decay2:key4nullStr:2} 
\end{align}
\end{lemma}

We defer the proof of this lemma until later. Instead, we first finish the proof of Proposition \ref{prop:decay2:nullStr}, assuming Lemma \ref{lem:decay2:key4nullStr}.

\noindent\emph{Proof of Proposition \ref{prop:decay2:nullStr}.}
First, we claim that \eqref{eq:decay2:key4nullStr:1} and \eqref{eq:decay2:key4nullStr:2} imply
\begin{equation} \label{eq:decay2:nullStr:pf:5}
	\sup_{\intr \cap \set{(u,v) \in \PD : u \in [3 u_{2}, U]}} u^{\om} \bb( \abs{\rd_{v}^{2} (r \phi)} + \abs{\rd_{u}^{2} (r \phi)} + \abs{\rd_{v} \dvr} + \abs{\rd_{u} \dur} \bb)
	\leq H_{2} + (\eps + \eps')(u_{2}) \calB_{2}(U).
\end{equation}
for some constant $0 < H_{2} < \infty$ and some positive function $\eps'(u_{2})$ which tends to zero as $u_{2} \to \infty$.

The point is that $F_{1}, F_{2}, G_{1}, G_{2}$ controls $\rd_{v}^{2} (r \phi)$, $\rd_{u}^{2} (r \phi)$, $\rd_{v} \dvr$, $\rd_{u} \dur$, respectively, up to higher order terms, which may be absorbed into the second term on the right-hand side. Indeed, consider $u \in [3 u_{2}, U]$. For $F_{1}$ and $G_{1}$, we estimate
\begin{align*}
	& u^{\om} \abs{\rd_{v}^{2} (r \phi)(u,v)} = u^{\om} \abs{F_{1} +  (\rd_{v} \dvr) \phi}(u,v) \leq u^{\om} \abs{F_{1}(u,v)} + \sup_{C_{u}} \abs{\phi} \cdot \calB_{2}(U), \\
	& u^{\om} \abs{\rd_{u}^{2} (r \phi)(u,v)} = u^{\om} \abs{G_{1} + (\rd_{u} \dur) \phi}(u,v) \leq u^{\om} \abs{G_{1}(u,v)} + \sup_{C_{u}} \abs{\phi} \cdot \calB_{2}(U),
\end{align*}
which are acceptable, as $\sup_{C_{u}} \abs{\phi} \to 0$ as $u \geq 3 u_{2} \to \infty$ by Theorem \ref{main.thm.1}. 
For $F_{2}$, we use Proposition \ref{prop:geomLocBVScat} to estimate
\begin{align*}
	u^{\om} \abs{\rd_{v} \dvr} = & u^{\om} \dvr \bb\vert F_{2} + \frac{\dvr}{1-\mu} \frac{\mu}{r} + \rd_{v} \phi ( \dvr^{-1} \rd_{v} (r \phi) - \dur^{-1} \rd_{u} (r \phi) ) \bb\vert \\
	\leq &\frac{1}{2} u^{\om} \abs{F_{2}} + \frac{K \Lmb}{4} u^{\om} \abs{\frac{\mu}{r}} + \frac{\Lmb}{2} u^{\om} \abs{\rd_{v} \phi} \bb( \abs{\rd_{v} (r \phi)} + \abs{\rd_{u} (r \phi)} \bb).
\end{align*}

Applying \eqref{eq:muOverR:1} (from Lemma \ref{lem:muOverR}) to the second term on the last line, and then using Lemma \ref{lem:dphi} to control $u^{\om} \abs{\rd_{v} \phi}$, we arrive at 
\begin{equation*}
u^{\om} \abs{\rd_{v} \dvr(u,v)} \leq \frac{1}{2} u^{\om} \abs{F_{2}(u,v)} + C_{K, \Lmb} \, \Psi \sup_{C_{u}} \bb( \abs{\rd_{v} (r \phi)} + \abs{\rd_{u} (r \phi)} \bb) \cdot \calB_{2}(U),
\end{equation*}
which is acceptable in view of Theorem \ref{main.thm.1}. Proceeding similarly, but also using the second inequality of Lemma \ref{lem:dphi} to control $\abs{\rd_{u} \phi}$, we obtain
\begin{equation*}
u^{\om} \abs{\rd_{u} \dur(u,v)} \leq K u^{\om} \abs{G_{2}(u,v)} + C_{K, \Lmb} \, \Psi \sup_{C_{u}} \bb( \abs{\rd_{v} (r \phi)} + \abs{\rd_{u} (r \phi)} \bb) \cdot \calB_{2}(U).
\end{equation*}

Combining these estimates with \eqref{eq:decay2:key4nullStr:1} and \eqref{eq:decay2:key4nullStr:2}, we conclude \eqref{eq:decay2:nullStr:pf:5} with
\begin{align} 
	H_{2} =& C_{\Lmb} \calI_{2} + C_{K, \Lmb, M_{i}} A_{1}^{3} + C_{K, \Lmb} A_{1}^{2}, 
	\label{eq:decay2:H2} \\
	\eps'(u_{2}) =& C \sup_{u \geq 3u_{2}} \abs{\phi} + C_{K, \Lmb} \Psi \sup_{u \geq 3u_{2}} \bb( \abs{\rd_{v}(r \phi)} + \abs{\rd_{u}(r \phi)} \bb).
	\label{eq:decay2:eps'}
\end{align}

Next, note that the (non-decreasing) function
\begin{equation} \label{eq:decay2:def4H'2}
	H'_{2}(u_{2}) := \sup_{\intr \cap \set{(u,v) \in \PD : u \in [1, 3u_{2}]}} u^{\om} \bb( \abs{\rd_{v}^{2} (r \phi)} + \abs{\rd_{u}^{2} (r \phi)} + \abs{\rd_{v} \dvr} + \abs{\rd_{u} \dur} \bb) \geq 0
\end{equation}
is always \emph{finite} for any fixed $u_{2} \geq 1$, as the region $\intr \cap \set{(u,v) \in \PD : u \in [1, 3u_{2}]}$ is compact and each of these terms is a continuous function, since $(\phi, r, m)$ is a $C^{1}$ solution (see Definition \ref{def:C1solution}). Combining with \eqref{eq:decay2:nullStr:pf:5}, we obtain
\begin{equation*} 
	\sup_{\intr \cap \set{(u,v) \in \PD : u \in [1, U]}} u^{\om} \bb( \abs{\rd_{v}^{2} (r \phi)} + \abs{\rd_{u}^{2} (r \phi)} + \abs{\rd_{v} \dvr} + \abs{\rd_{u} \dur} \bb)
	\leq H_{2} + H'_{2}(u_{2}) + (\eps + \eps')(u_{2}) \calB_{2}(U),
\end{equation*}
for every $u_2\in [1,U]$.

Now apply \eqref{eq:decay2:uDecayInExtr:1}, \eqref{eq:decay2:uDecayInExtr:3} in Lemma \ref{lem:decay2:uDecayInExtr} to $\rd_{u}^{2}(r \phi)$, $\rd_{u} \dur$. Apply also Lemma \ref{lem:decay2:rDecay} (along with the fact that $r(u,v) \geq 2 \Lmb^{-1} u$ in $\extr$ and $\omg \leq 3$) to $\rd_{v}^{2} (r \phi)$, $\rd_{v} \dvr$ in $\extr$. Then we see that there exist a non-negative and non-decreasing function $H''_{2}(u_2)$ and a positive function $\eps''(u_{2})$ such that
\begin{equation*} 
	\calB_{2}(U) \leq H_{2}''(u_2) +  \eps''(u_{2}) \calB_{2}(U),
\end{equation*}
and $\eps''(u_{2}) \to 0$ as $u_{2} \to \infty$. Taking $u_{2}$ sufficiently large, the second term on the right-hand side can be absorbed into the left-hand side; then we conclude that $\calB_{2}(U) \leq C_{A_{1}, K, \Lmb} H''_{2}(u_{2})$. As this bound is independent of $U$, Proposition \ref{prop:decay2:nullStr} then follows. \qedhere
\end{proof}

\begin{remark} 
Using \eqref{eq:decay2:uDecayInExtr:1}, \eqref{eq:decay2:uDecayInExtr:3} in Lemma \ref{lem:decay2:uDecayInExtr} and \eqref{eq:decay2:rDecay:1}, \eqref{eq:decay2:rDecay:2} in Lemma \ref{lem:decay2:rDecay}, the functions $H_{2}''(u_{2})$ and $\eps''(u_{2})$ can be explicitly bounded from the above as follows:
\begin{align}
	H_{2}''(u_2) \leq & C_{K, \Lmb} \bb( 1 + \Psi \int_{3}^{\infty} \abs{\frac{2 m \dvr}{(1-\mu) r^{2}}}(u, v') \, \ud v' \bb) \cdot (H_{2} + H'(u_{2}) + A_{1}^{2} + A_{1}^{3}) 	
	\label{eq:decay2:H2''} \\
			& + C \calI_{2} + C_{K, \Lmb} A_{1}^{2} + C_{K, \Lmb, M_{i}} A_{1}^{3} 
	\notag \\
	\eps''(u_{2}) \leq & C_{K, \Lmb} \bb( 1 + \Psi \int_{3}^{\infty} \abs{\frac{2 m \dvr}{(1-\mu) r^{2}}}(u, v') \, \ud v' \bb) \cdot (\eps + \eps')(u_{2}).
	\label{eq:decay2:eps''}
\end{align}

These bounds will be useful in our proof of Theorem \ref{thm:smallData} in Section \ref{sec:smallData}.
\end{remark}

At this point, in order to complete the proof of Proposition \ref{prop:decay2:nullStr}, we are only left to prove Lemma \ref{lem:decay2:key4nullStr}.
\begin{proof}[Proof of Lemma \ref{lem:decay2:key4nullStr}]
Let $(u, v) \in \intr$ (i.e., $v \in [u, 3u]$) with $u \in [3 u_{2}, U]$. In this proof, we will use the notation $\eps(u_{2})$ to refer to a positive quantity which may be made arbitrarily small by choosing $u_{2}$ large enough, which may vary from line to line.

We first estimate $F_1$ and $F_2$. Integrating \eqref{eq:decay2:nullStr:pf:1} and \eqref{eq:decay2:nullStr:pf:2} along the incoming direction from $(u/3, v)$ to $(u,v)$, we obtain
\begin{align*}
	\abs{F_{1}(u,v)} \leq & \abs{F_{1}(u/3, v)} + \int_{u/3}^{u} \abs{\rd_{u} \dvr \, \rd_{v} \phi(u', v)} + \abs{\rd_{v} \dvr \, \rd_{u} \phi (u', v)} \, \ud u',  \\
	\abs{F_{2}(u,v)} \leq & \abs{F_{2}(u/3, v)} + \int_{u/3}^{u} \abs{\rd_{u} \phi \, \rd_{v}( \dur^{-1} \rd_{u} (r \phi) )(u', v)} + \abs{\rd_{v} \phi \, \rd_{u} ( \dur^{-1} \rd_{u} (r \phi) )(u', v)}  \, \ud u'.
\end{align*}

Multiply both sides of these inequalities by $u^{\om}$. For $v \in [u, 3u]$, note that $(u/3, v) \in \extr$ and $u \leq (3 \Lmb/2) r(u/3, v)$. 
Therefore, using Theorem \ref{main.thm.1} for $\phi$, $\rd_v(r\phi)$, Corollary \ref{cor:decay1} for $\rd_{v} \phi$, Lemma \ref{lem:muOverR} for $\mu/r$ and Lemma \ref{lem:decay2:rDecay} for $\rd_{v}^{2}(r\phi)$, $\rd_{v} \dvr$, we have
\begin{align*}
	u^{\om} \abs{F_{1} (u/3, v)} 
		\leq & C_{\Lmb} r^{\om} \bb( \abs{\rd_{v}^{2} (r \phi)} + \abs{(\rd_{v} \dvr) \phi} \bb) (u/3, v)  \\
		\leq & C_{\Lmb} \calI_{2} + C_{K, \Lmb, M_{i}} A_{1}^{3}, \\
	u^{\om} \abs{F_{2} (u/3, v)} 
		\leq & C_{\Lmb} r^{\om} \bb( \abs{\dvr^{-1} \rd_{v} \dvr} + \frac{\mu}{(1-\mu)} \frac{\dvr}{r} + \abs{\rd_{v} \phi ( \dvr^{-1} \rd_{v} (r \phi) - \dur^{-1} \rd_{u} ( r \phi) )} \bb) (u/3, v)\\
		\leq & C_{K, \Lmb} A_{1}^{2}.
\end{align*}

Therefore, we only need to deal with the $u'$-integrals. For $u \in [3 u_{2}, U]$, we claim that
\begin{align}
	u^{\om} \int_{u/3}^{u} \abs{\rd_{u} \dvr(u', v)} \abs{\rd_{v} \phi(u',v)} \, \ud u' \leq & \eps(u_{2}) \calB_{2}(U), \label{eq:decay2:key4nullStr:pf:1} \\
	u^{\om} \int_{u/3}^{u} \abs{\rd_{v} \dvr(u', v)} \abs{\rd_{u} \phi(u',v)}\, \ud u' \leq & \eps(u_{2}) \calB_{2}(U), \label{eq:decay2:key4nullStr:pf:2} \\
	u^{\om} \int_{u/3}^{u} \abs{\rd_{u} \phi(u', v)} \abs{\rd_{v} (\dur^{-1} \rd_{u} (r \phi))(u', v)} \, \ud u' \leq & \eps (u_{2}) \calB_{2}(U), \label{eq:decay2:key4nullStr:pf:3} \\
	u^{\om} \int_{u/3}^{u} \abs{\rd_{v} \phi(u', v)} \abs{\rd_{u} (\dur^{-1} \rd_{u} (r \phi))(u', v)} \, \ud u' \leq &  \eps (u_{2}) \calB_{2}(U). \label{eq:decay2:key4nullStr:pf:4}
\end{align}

\pfstep{Proof of \eqref{eq:decay2:key4nullStr:pf:1}}
We proceed similarly as in the proof of Theorem \ref{main.thm.1}. By \eqref{eq:SSESF:dr}, \eqref{eq:bnd4dvrphi} and Lemma \ref{lem:dphi}, we estimate
\begin{align*}
& u^{\om} \int_{u/3}^{u} \abs{\rd_{u} \dvr(u', v)} \abs{\rd_{v} \phi(u',v)} \, \ud u' \\
& \quad \leq C_{\Lmb} \bb( \int_{u_{2}}^{v} \abs{\frac{2m \dur}{(1-\mu) r^{2}}(u', v)} \, \ud u' \bb) \sup_{u' \in [u/3, u]} \sup_{C_{u'}} (u')^{\omg} (\abs{\rd^{2}_{v}(r \phi)} + \Psi \abs{\rd_{v} \dvr} ) \\
& \quad \leq C_{\Lmb, \Psi} \bb( \int_{u_{2}}^{v} \abs{\frac{2m \dur}{(1-\mu) r^{2}}(u', v)} \, \ud u' \bb) \calB_{2}(U).
\end{align*}

Thus \eqref{eq:decay2:key4nullStr:pf:1} follows by Lemma \ref{lem:smallPtnl}.

\pfstep{Proof of \eqref{eq:decay2:key4nullStr:pf:2}}
We have
\begin{align*}
u^{\om} \int_{u/3}^{u} \abs{\rd_{v} \dvr(u', v)} \abs{\rd_{u} \phi(u',v)}\, \ud u'
\leq & C \bb( \int_{u_{2}}^{v} \abs{\rd_{u} \phi(u', v)} \, \ud u' \bb) \sup_{u' \in [u/3, u]} \sup_{C_{u'}} (u')^{\om} \abs{\rd_{v} \dvr} \\
\leq & C \bb( \int_{u_{2}}^{v} \abs{\rd_{u} \phi(u', v)} \, \ud u' \bb) \calB_{2}(U).
\end{align*}

Thus \eqref{eq:decay2:key4nullStr:pf:2} follows by Lemma \ref{lem:smallDphi}.

\pfstep{Proof of \eqref{eq:decay2:key4nullStr:pf:3}}
We start with the identity
\begin{equation*}
	\rd_{v} (\dur^{-1} \rd_{u} (r\phi)) = - \frac{2m}{(1-\mu)r^{2}} \dvr ( \dur^{-1} \rd_{u}(r \phi) - \phi).
\end{equation*}
which is readily verifiable using \eqref{eq:SSESF:dr} and \eqref{eq:SSESF:dphi}. By \eqref{eq:bnd4phi} and \eqref{eq:bnd4durphi}, we estimate
\begin{align*}
& u^{\om} \int_{u/3}^{u} \abs{\rd_{u} \phi(u', v)} \abs{\rd_{v} (\dur^{-1} \rd_{u} (r \phi))(u', v)} \, \ud u' \\
& \quad \leq C_{K, \Lmb} \Psi \bb( \int_{u_{2}}^{v} \abs{\frac{2m \dur}{(1-\mu) r^{2}} (u', v)} \, \ud u' \bb) \,  \sup_{u' \in [u/3, u]} \sup_{C_{u'}} (u')^{\om}\abs{\rd_{u} \phi}.
\end{align*}

The $u'$-integral vanishes as $u_{2} \to \infty$ by Lemma \ref{lem:smallPtnl}. On the other hand, by Lemma \ref{lem:dphi} and Proposition \ref{prop:geomLocBVScat}, we have
\begin{equation} \label{eq:decay2:key4nullStr:pf:3:1}
\sup_{C_{u'}} (u')^{\om}\abs{\rd_{u} \phi}
\leq C_{K, \Lmb} \sup_{C_{u'}} (u')^{\om}\abs{\rd_{v} \phi} 
\leq C_{K, \Lmb, \Psi} \calB_{2}(U),
\end{equation}
for any $u' \in [1, U]$.
Therefore, \eqref{eq:decay2:key4nullStr:pf:3} follows.

\pfstep{Proof of \eqref{eq:decay2:key4nullStr:pf:4}}
Here we divide the integral into two, one in $\cpt$ and the other outside. Recall the notation $u^{\star}(v) = \sup \set{u \in [1, v] : r (u,v) \geq R}$. Below, we will consider the case $u^{\star}(v) \in [u/3, u]$, i.e., when the line segment $\set{(u', v) \in \PD : u' \in [u/3, u]}$ crosses $\set{r = R}$; the other case is easier, and can be handled with a minor modification.

We first deal with the integral over the portion in $\cpt$. We claim that
\begin{equation*}
	u^{\om} \int_{u^{\star}(v)}^{u} \abs{\rd_{v} \phi(u', v)} \abs{\rd_{u} (\dur^{-1} \rd_{u} (r \phi))(u', v)} \, \ud u' \leq  \eps (u_{2}) \calB_{2}(U). 
\end{equation*}

This is an easy consequence of the bound for $|\rd_v\phi|$ in Lemma \ref{lem:dphi}, the fact that $u$, $u'$ are comparable over the domain of integration and
\begin{equation*}
	\sup_{v \in [u_{2}, \infty)} \int_{\uC_{v} \cap \set{u \geq u_{2}} \cap \cpt} \abs{\rd_{u} (\dur^{-1} \rd_{u} (r \phi))}   \to 0 \hbox{ as } u_{2} \to \infty
\end{equation*}
which follows from \eqref{eq:decay1:3}, \eqref{eq:bnd4dur} and Theorem \ref{thm:decayInCpt}.

We now consider the remaining contribution to the integral. We begin as follows.
\begin{align*}
	& u^{\om} \int_{u/3}^{u^{\star}(v)} \abs{\rd_{v} \phi(u', v)} \abs{\rd_{u} (\dur^{-1} \rd_{u} (r \phi))(u', v)} \, \ud u'  \\
	& \quad \leq C_{K, \Lmb} \bb( \int_{u/3}^{u^{\star}(v)} \abs{\rd_{v} \phi(u', v)} \, \ud u' \bb) \sup_{u' \in [u/3, u^{\star}(v)]} \sup_{C_{u'}} (u')^{\omg} (\abs{\rd_{u}^{2}(r \phi)} + \Psi \abs{\rd_{u} \dur}) \\
	& \quad \leq C_{K, \Lmb, \Psi} \bb( \int_{u/3}^{u^{\star}(v)} \abs{\rd_{v} \phi(u', v)} \, \ud u' \bb) \calB_{2}(U).
\end{align*}

For $u' \in [u/3, u^{\star}(v)]$, we have $r (u', v) \geq R$. Thus, by \eqref{eq:decay1:4}, we have
\begin{equation*}
	\int_{u/3}^{u^{\star}(v)} \abs{\rd_{v} \phi(u', v)} \, \ud u' \leq \frac{C_{K} A_{1}}{R} \int_{u_{2}}^{\infty} (u')^{-\om} \, \ud u'  \leq \frac{C_{K} A_{1}}{R} u_{2}^{-(\om-1)},
\end{equation*}
which vanishes as $u_{2} \to \infty$. Therefore, in the case under consideration, \eqref{eq:decay2:key4nullStr:pf:4} follows.

\vspace{.5em}

We have therefore obtained the desired bounds for $F_1$ and $F_2$. Next, estimate $G_1$ and $G_2$. Let us integrate \eqref{eq:decay2:nullStr:pf:3} and \eqref{eq:decay2:nullStr:pf:4} along the outgoing direction from $(u, u)$ on the axis to $(u,v)$. Then we obtain
\begin{align*}
	\abs{G_{1}(u,v)} \leq & \lim_{v' \to u+} \abs{G_{1}(u, v')} + \int_{u}^{v} \abs{\rd_{v} \dur \, \rd_{u} \phi (u, v')} + \abs{\rd_{u} \dur \, \rd_{v} \phi (u, v')} \, \ud v',  \\
	\abs{G_{2}(u,v)} \leq & \lim_{v' \to u+} \abs{G_{2}(u, v')} + \int_{u}^{v} \abs{\rd_{v} \phi \, \rd_{u} ( \dvr^{-1} \rd_{v} (r \phi) ) (u, v')} + \abs{\rd_{u} \phi \, \rd_{v} ( \dvr^{-1} \rd_{v} (r \phi) ) (u, v')} \, \ud v'.
\end{align*}

Note that
\begin{equation*}
	\lim_{v \to u+} \frac{\mu}{r} (u, v) = 0, \quad
	\lim_{v \to u+} \bb( \dvr^{-1} \rd_{v} (r \phi)(u,v) - \dur^{-1} \rd_{u} ( r \phi) (u,v) \bb) = 0,
\end{equation*}
since $(\phi, r, m)$ is a $C^{1}$ solution. It follows that $\lim_{v \to u+} \rd_{u} \rd_{v} (r \phi)(u, v) = 0$ and $\lim_{v \to u+} \rd_{u} \rd_{v} r (u,v) = 0$. Moreover, we also have
\begin{align*}
	& \lim_{v \to u+} \rd_{v}^{2} (r \phi) (u,v) = - \lim_{v \to u+} \rd_{u}^{2} (r \phi) (u,v), \quad
	\lim_{v \to u+} \rd_{v} \dvr (u, v)= - \lim_{v \to u+} \rd_{u} \dur (u,v).
\end{align*}

As a consequence, 
\begin{equation*}
	\lim_{v' \to u+} G_{1}(u, v') = - \lim_{v' \to u+} F_{1}(u, v'), \quad
	\lim_{v' \to u+} G_{2}(u, v') = \lim_{v' \to u+} F_{2}(u, v').
\end{equation*}

Therefore, by the previous estimates for $F_{1}, F_{2}$, we have
\begin{align*}
	& u^{\om} \lim_{v' \to u+} \abs{G_{1}(u, v')} \leq C_{\Lmb} \calI_{2} + C_{K, \Lmb, M_{i}} A_{1}^{3} + \eps(u_{2}) \calB_{2}(U), \\
	& u^{\om} \lim_{v' \to u+} \abs{G_{2}(u, v')} \leq C_{K, \Lmb} A_{1}^{2} + \eps(u_{2}) \calB_{2}(U),
\end{align*}
which are acceptable. Recalling that we are considering $(u,v) \in \intr$, hence $v \in [u, 3u]$, we are now left to establish the following estimates: 
\begin{align}
	u^{\om} \int_{u}^{3u} \abs{\rd_{v} \dur(u, v')} \abs{\rd_{u} \phi(u,v')} \, \ud v' \leq & \eps(u_{2}) \calB_{2}(U), \label{eq:decay2:key4nullStr:pf:5} \\
	u^{\om} \int_{u}^{3u} \abs{\rd_{u} \dur(u, v')} \abs{\rd_{v} \phi(u,v')}\, \ud v' \leq & \eps(u_{2}) \calB_{2}(U), \label{eq:decay2:key4nullStr:pf:6} \\
	u^{\om} \int_{u}^{3u} \abs{\rd_{v} \phi(u, v')} \abs{\rd_{u} (\dvr^{-1} \rd_{v} (r \phi))(u, v')} \, \ud v' \leq & \eps (u_{2}) \calB_{2}(U), \label{eq:decay2:key4nullStr:pf:7} \\
	u^{\om} \int_{u}^{3u} \abs{\rd_{u} \phi(u, v')} \abs{\rd_{v} (\dvr^{-1} \rd_{v} (r \phi))(u, v')} \, \ud v' \leq &  \eps (u_{2}) \calB_{2}(U). \label{eq:decay2:key4nullStr:pf:8}
\end{align}

\pfstep{Proof of \eqref{eq:decay2:key4nullStr:pf:5}}
Substituting $\rd_{v} \dur$ by \eqref{eq:SSESF:dr} and using \eqref{eq:decay2:key4nullStr:pf:3:1}, we have
\begin{align*}
u^{\om} \int_{u}^{3u} \abs{\rd_{v} \dur(u, v')} \abs{\rd_{u} \phi(u,v')} \, \ud v'
\leq & K \bb( \int_{u}^{\infty} \abs{\frac{2m \dvr}{(1-\mu) r^{2}}(u, v')} \, \ud v' \bb) \sup_{v' \in [u, 3u]} u^{\om} \abs{\rd_{u} \phi(u,v')}  \\
\leq & C_{K, \Lmb, \Psi} \bb( \sup_{u \geq 3u_{2}} \int_{u}^{\infty} \abs{\frac{2m \dvr}{(1-\mu) r^{2}}(u, v')} \, \ud v' \bb) \calB_{2}(U).
\end{align*}

Thus \eqref{eq:decay2:key4nullStr:pf:5} follows by Lemma \ref{lem:smallPtnl}.

\pfstep{Proof of \eqref{eq:decay2:key4nullStr:pf:6}}
We have
\begin{align*}
u^{\om} \int_{u}^{3u} \abs{\rd_{u} \dur(u, v')} \abs{\rd_{v} \phi(u,v')}\, \ud v'
\leq & \int_{u}^{\infty} \abs{\rd_{v} \phi(u, v')} \, \ud v'  \sup_{v' \in [u, 3u]} u^{\om} \abs{\rd_{u} \dur(u, v')}  \\
\leq & \bb( \sup_{u \geq 3u_{2}} \int_{u}^{\infty} \abs{\rd_{v} \phi(u, v')} \, \ud v'  \bb) \calB_{2}(U).
\end{align*}

Thus \eqref{eq:decay2:key4nullStr:pf:6} follows by Lemma \ref{lem:smallDphi}.

\pfstep{Proof of \eqref{eq:decay2:key4nullStr:pf:7}}
By \eqref{eq:SSESF:dr} and \eqref{eq:SSESF:dphi}, we have the identity
\begin{equation*}
	\rd_{u} (\dvr^{-1} \rd_{v} (r\phi)) = - \frac{2m}{(1-\mu)r^{2}} \dur ( \dvr^{-1} \rd_{v}(r \phi) - \phi).
\end{equation*}

Then by Proposition \ref{prop:geomLocBVScat}, we have
\begin{align*}
&u^{\om} \int_{u}^{3u} \abs{\rd_{v} \phi(u, v')} \abs{\rd_{u} (\dvr^{-1} \rd_{v} (r \phi))(u, v')} \, \ud v' \\
& \quad \leq C_{K, \Lmb} \Psi \bb( \int_{u}^{\infty} \abs{\frac{2m \dvr}{(1-\mu) r^{2}} (u, v')} \, \ud v' \bb)  \sup_{v' \in [u, 3u]} u^{\om} \abs{\rd_{v} \phi(u,v')}.
\end{align*}

Now \eqref{eq:decay2:key4nullStr:pf:7} follows by Lemmas \ref{lem:smallPtnl} and \ref{lem:dphi} and \eqref{eq:bnd4dvrphi}.

\pfstep{Proof of \eqref{eq:decay2:key4nullStr:pf:8}}
As in the proof of \eqref{eq:decay2:key4nullStr:pf:4}, we will divide the integral into two pieces. More precisely, let us define $v^{\star}(u)$ to be the unique $v$-value such that $r(u, v^{\star}(u)) = R$. Assuming $v^{\star}(u) \in [u, 3u]$, the integral $\int_{u}^{3u}$ will be divided into $\int_{u}^{v^{\star}(u)}$ and $\int_{v^{\star}(u)}^{3u}$. The remaining case $v^{\star}(u) > 3u$ can be dealt with by adapting the argument for the first integral.

For the first integral, we claim that
\begin{equation*}
	u^{\om} \int_{u}^{v^{\star}(u)} \abs{\rd_{u} \phi(u, v')} \abs{\rd_{v} (\dvr^{-1} \rd_{v} (r \phi))(u, v')} \, \ud v' \leq \eps(u_{2}) \calB_{2}(U).
\end{equation*}

From the locally BV scattering assumption \eqref{eq:locBVScat}, we have
\begin{equation*}
	\sup_{u \in [3 u_{2}, \infty)} \int_{C_{u} \cap \cpt} \abs{\rd_{v} (\dvr^{-1} \rd_{v} (r \phi))}  \to 0 \hbox{ as } u_{2} \to \infty.
\end{equation*}

Combined with \eqref{eq:decay2:key4nullStr:pf:3:1}, the claim follows.

Next, we turn to the second integral. By \eqref{eq:bnd4dvr} and \eqref{eq:bnd4dvrphi}, we estimate
\begin{align*}
& u^{\om} \int_{v^{\star}(u)}^{3u} \abs{\rd_{u} \phi(u, v')} \abs{\rd_{v} (\dvr^{-1} \rd_{v} (r \phi))(u, v')} \, \ud v'  \\
& \quad \leq \sup_{v' \in [v^{\star}(u), 3u]} u^{\om} \abs{\rd_{v} (\dvr^{-1} \rd_{v} (r \phi))(u, v')} \int_{v^{\star}(u)}^{3u} \abs{\rd_{u} \phi(u, v')} \, \ud v' \\
& \quad \leq C_{\Lmb, \Psi} \calB_{2}(U) \int_{v^{\star}(u)}^{3u} \abs{\rd_{u} \phi(u, v')} \, \ud v'.
\end{align*}

For $v' \in [v^{\star}(u), 3u]$, we have $r (u, v) \geq R$. Thus, by \eqref{eq:decay1:5}, we have
\begin{equation*}
	\int_{v^{\star}(u)}^{3u} \abs{\rd_{u} \phi(u, v')} \, \ud v' \leq \frac{C_{K} A_{1}}{R} \int_{u}^{3u} u^{-\om} \, \ud v'  \leq \frac{C_{K} A_{1}}{R} u_{2}^{-(\om-1)},
\end{equation*}
which vanishes as $u_{2} \to \infty$, and therefore finishes the proof of \eqref{eq:decay2:key4nullStr:pf:8}. We remark that the fact that we are in $\intr$ is used crucially here, as otherwise the integral would not be convergent. 
\end{proof}

\begin{remark} 
In the case where we have global BV scattering (i.e., conditions $(2)$ and $(3)$ of Definition \ref{def:locBVScat} are satisfied with $R=\infty$), we can take $R = \infty$ in the preceding argument to obtain the following explicit upper bound on $\eps(u_{2})$:
\begin{equation} \label{eq:decay2:eps}
\begin{aligned}
	\eps(u_{2}) 
	\leq & C_{K, \Lmb, \Psi} \sup_{v \in [u_{2}, \infty)} \int_{\uC_{v} \cap \set{u \geq u_{2}}} \abs{\frac{2m \dur}{(1-\mu) r^{2}}}
		+C \sup_{v \in [u_{2}, \infty)} \int_{\uC_{v} \cap \set{u \geq u_{2}}} \abs{\rd_{u} \phi} \\
		& + C_{K, \Lmb, \Psi} \sup_{v \in [u_{2}, \infty)} \int_{\uC_{v} \cap \set{u \geq u_{2}}} \abs{\rd_{u} (\dur^{-1} \rd_{u} (r \phi))} \\
		& +C_{K, \Lmb, \Psi}  \sup_{u \geq 3u_{2}} \int_{C_{u}} \abs{\frac{2m \dvr}{(1-\mu) r^{2}}}
		+C  \sup_{u \geq 3u_{2}} \int_{C_{u}} \abs{\rd_{v} \phi}  \\
		& +C_{K, \Lmb, \Psi} \sup_{u \in [3 u_{2}, \infty)} \int_{C_{u}} \abs{\rd_{v} (\dvr^{-1} \rd_{v} (r \phi))}.
\end{aligned}
\end{equation}

This will be useful in our proof of the sharp decay rate in the case of small BV norm (Theorem \ref{thm:smallData}) in Section \ref{sec:smallData}.
\end{remark}

In the second step of our proof of Theorem \ref{main.thm.2}, we use the preliminary $u^{-\omg}$ decay proved in Proposition \ref{prop:decay2:nullStr} to obtain the optimal the $u$-decay. Key to this step is the following proposition, which claims optimal $u$-decay in $\intr$.

\begin{proposition} \label{prop:decay2:final}
There exists a constant $0 < A_{2}'' < \infty$ such that the following estimates hold.
\begin{align} 
\sup_{\intr} u^{\om+1} \abs{\rd_{v}^{2} (r \phi)} \leq & A_{2}'',
	\label{eq:decay2:final:1} \\
\sup_{\intr} u^{\om+1} \abs{\rd_{u}^{2} (r \phi)} \leq & A_{2}'', 
	\label{eq:decay2:final:2} \\
\sup_{\intr} u^{3} \abs{\rd_{v} \dvr} \leq & A_{2}'', 
	\label{eq:decay2:final:3} \\
\sup_{\intr} u^{3} \abs{\rd_{u} \dur} \leq & A_{2}''. 
	\label{eq:decay2:final:4}
\end{align}
\end{proposition}

Once we establish Proposition \ref{prop:decay2:final}, the desired decay for $\rd_{v}^{2}(r \phi)$ and $\rd_{v} \dvr$ follow from Lemma \ref{lem:decay2:rDecay} and the fact that $r \geq 2 \Lmb^{-1} u$ in $\extr$. Furthermore, the desired decay for $\rd_{u}^{2}(r \phi)$ and $\rd_{u} \dur$ follow from Lemma \ref{lem:decay2:uDecayInExtr}.

\begin{proof} 
Thanks to the fact that we have pointwise bounds for sufficient number of derivatives (albeit with sub-optimal decay) near $\Gmm$ at this point, it suffices to work with the `non-renormalized' equations \eqref{eq:eq4dvdvrphi:normal}, \eqref{eq:eq4dudurphi:normal}, \eqref{eq:eq4dvdvr:normal} and \eqref{eq:eq4dudur:normal}. In particular, we need not utilize the null structure of \eqref{eq:SSESF}. 

Let $(u,v) \in \intr$ (i.e., $v \in [u, 3u]$) with $u \geq 3$. We begin with \eqref{eq:decay2:final:1}. Integrating $\rd_{u} \rd_{v}^{2}(r\phi)$ in the $u$-direction from $u/3$ to $u$, multiplying by $u^{\omg+1}$ and using $r(u/3, v) \geq (2/3) \Lmb^{-1} u$, we obtain
\begin{equation} \label{eq:decay2:final:1:pf}
	u^{\omg+1}\abs{\rd_{v}^{2}(r \phi)}(u,v)
	\leq C_{\Lmb} r^{\omg+1} \abs{\rd_{v}^{2}(r \phi)}(u/3, v) + u^{\omg+1} \int_{u/3}^{u} \abs{\rd_{u} \rd_{v}^{2}(r \phi)}(u', v) \, \ud u'.
\end{equation}

Since $(u/3, v) \in \extr$, the first term on the right-hand side is bounded by $\leq C_{\Lmb} \calI_{2} + C_{K, \Lmb, M_{i}} A_{1}^{3}$, thanks to Lemma \ref{lem:decay2:rDecay}. To estimate the $u'$-integral, we substitute $\rd_{u} \rd_{v}^{2} (r \phi)$ by \eqref{eq:eq4dvdvrphi:normal}. Then applying Proposition \ref{prop:geomLocBVScat}, Lemma \ref{lem:dphi}, Lemma \ref{lem:muOverR}, Theorem \ref{main.thm.1} and Proposition \ref{prop:decay2:nullStr}, we obtain
\begin{equation*}
	\abs{\rd_{u} \rd_{v}^{2}(r \phi)}(u', v)
	\leq C_{A_{1}, K, \Lmb} (u')^{-3\omg} A_{1} (A_{2}')^{2}.
\end{equation*}

Thus we have
\begin{equation} \label{eq:decay2:final:1:explicit}
	u^{\omg+1}\abs{\rd_{v}^{2}(r \phi)}(u,v)
	\leq C_{\Lmb} \calI_{2} + C_{K, \Lmb, M_{i}} A_{1}^{3} + C_{A_{1}, K, \Lmb} A_{1} (A_{2}')^{2},
\end{equation}
where we have used the fact that $\omg > 1$, and thus $3 \omg - 1 > \omg + 1$ to throw away the $u$-weight in the last term. This proves \eqref{eq:decay2:final:1}.

Next, we prove \eqref{eq:decay2:final:2}. Integrating $\rd_{v} \rd_{u}^{2} (r \phi)$ in the $v$-direction from $u+$ to $v$ and multiplying by $u^{\omg+1}$, we have
\begin{equation} \label{eq:decay2:final:2:pf}
	u^{\omg+1} \abs{\rd_{u}^{2}(r \phi)}(u,v)
	\leq \lim_{v' \to u+} u^{\omg+1} \abs{\rd_{u}^{2} (r \phi)}(u, v') + u^{\omg+1} \int_{u}^{3u} \abs{\rd_{v} \rd_{u}^{2} (r \phi)}(u, v') \, \ud v'.
\end{equation}

Recall that $\lim_{v' \to u+} \rd_{u}^{2} (r \phi)(u, v') = \lim_{v' \to u+} \rd_{v}^{2} (r \phi)(u, v')$, as $(\phi, r, m)$ is a $C^{1}$ solution. Thus the first term on the right-hand side can be estimated via \eqref{eq:decay2:final:1:explicit}. Substitute $\rd_{v} \rd_{u}^{2} (r \phi)$ by \eqref{eq:eq4dudurphi:normal} and apply, as before, Proposition \ref{prop:geomLocBVScat}, Lemma \ref{lem:dphi}, Lemma \ref{lem:muOverR}, Theorem \ref{main.thm.1} and Proposition \ref{prop:decay2:nullStr}. Then we have
\begin{equation*}
	\abs{\rd_{v} \rd_{u}^{2} (r \phi)}(u, v')
	\leq C_{A_{1}, K, \Lmb} u^{-3\omg} A_{1} (A_{2}')^{2}.
\end{equation*}

It now follows that 
\begin{equation} \label{eq:decay2:final:2:explicit}
	u^{\omg+1}\abs{\rd_{u}^{2}(r \phi)}(u,v)
	\leq C_{\Lmb} \calI_{2} + C_{K, \Lmb, M_{i}} A_{1}^{3} + C_{A_{1}, K, \Lmb} A_{1} (A_{2}')^{2},
\end{equation}
which proves \eqref{eq:decay2:final:2}.

At this point, combining Lemma \ref{lem:dphi}, Theorem \ref{main.thm.1}, Lemma \ref{lem:decay2:rDecay} and \eqref{eq:decay2:final:1:explicit}, note that we have the following improved $u$-decay for $\rd_{v} \phi$:
\begin{equation} \label{eq:decay2:final:impDvphi} 
\begin{aligned}
	\sup_{\PD} u^{\omg+1} \abs{\rd_{v}\phi} 
	\leq C_{\Lmb} \sup_{\PD} (u^{\omg+1} \abs{\rd_{v}^{2}(r \phi)} + u A_{1} \abs{\rd_{v} \dvr})
	\leq B
\end{aligned}
\end{equation}
where 
\begin{equation} \label{eq:decay2:final:aux}
B := C_{\Lmb} \calI_{2} + C_{K, \Lmb, M_{i}} A_{1}^{3} + C_{A_{1}, K, \Lmb} A_{1} (A_{2}')^{2} + C_{\Lmb} A_{1} A_{2}'.
\end{equation}

We now turn to \eqref{eq:decay2:final:3}. Integrating $\rd_{u} \rd_{v} \log \dvr$ in the $u$-direction from $u/3$ to $u$, multiplying by $u^{3}$ and using $r(u/3, v) \geq (2/3) \Lmb^{-1} u$, we obtain
\begin{equation} \label{eq:decay2:final:3:pf}
	u^{3} \abs{\rd_{v} \log \dvr }(u, v) \leq 
	C r^{3} \abs{\rd_{v} \log \dvr}(u/3, v) + u^{3} \int_{u/3}^{u} \abs{\rd_{u} \rd_{v} \log \dvr}(u', v) \, \ud u'
\end{equation}

Since $(u/3, v) \in \extr$, the first term on the right-hand side is estimated $\leq C_{K, \Lmb} A_{1}^{2}$ by Lemma \ref{lem:decay2:rDecay} and the fact that $\dvr^{-1} \leq \Lmb$. Next, substituting $\rd_{u} \log \dvr$ by \eqref{eq:eq4dvdvr:normal}, applying Proposition \ref{prop:geomLocBVScat}, Lemma \ref{lem:muOverR}, Lemma \ref{lem:dphi} and using the improved bound \eqref{eq:decay2:final:impDvphi}, we have 
\begin{equation*}
	\abs{\rd_{u} \rd_{v} \log \dvr}(u',v) \leq C_{K, \Lmb} B^{2} (u')^{-2(\omg+1)}.
\end{equation*}

Therefore
\begin{equation} \label{eq:decay2:final:3:explicit}
	u^{3} \abs{\rd_{v} \dvr}(u,v) \leq C_{K, \Lmb} A_{1}^{2} + C_{K, \Lmb} B^{2},
\end{equation}
where we used $2(\omg + 1) - 1 > 3$ to throw away the $u$-weight in the last term. This proves \eqref{eq:decay2:final:3}.

Finally, we prove \eqref{eq:decay2:final:4}. Integrating $\rd_{v} \rd_{u} \log \dur$ in the $v$-direction from $u+$ to $v$ and multiplying by $u^{3}$, we have
\begin{equation} \label{eq:decay2:final:4:pf}
	u^{3} \abs{\rd_{u} \log \dur}(u,v) \leq \lim_{v' \to u+} u^{3} \abs{\rd_{u} \log \dur}(u, v') + u^{3} \int_{u}^{3u} \abs{\rd_{v} \rd_{u} \log \dur}(u, v') \, \ud v'.
\end{equation}

Since $\lim_{v' \to u+} \rd_{u} \dur(u, v') = - \lim_{v' \to u+} \rd_{v} \dvr(u, v')$, the first term is bounded by \eqref{eq:decay2:final:3:explicit}. Furthermore, substituting $\rd_{v} \rd_{u} \log \dur$ by \eqref{eq:eq4dudur:normal} and applying Proposition \ref{prop:geomLocBVScat}, Lemma \ref{lem:muOverR}, Lemma \ref{lem:dphi} and using the improved bound \eqref{eq:decay2:final:impDvphi}, we have 
\begin{equation*}
	\abs{\rd_{v} \rd_{u} \log \dur}(u, v') \leq C_{K, \Lmb} B^{2} u^{-2(\omg+1)}.
\end{equation*}

As before, it follows that
\begin{equation} \label{eq:decay2:final:4:explicit}
	u^{3} \abs{\rd_{u} \dur}(u,v) \leq C_{K, \Lmb} A_{1}^{2} + C_{K, \Lmb} B^{2},
\end{equation}
which proves \eqref{eq:decay2:final:4}.
\end{proof}

\begin{remark} 
Combining \eqref{eq:decay2:final:1:explicit}, \eqref{eq:decay2:final:2:explicit}, \eqref{eq:decay2:final:3:explicit} and \eqref{eq:decay2:final:4:explicit}, we see that Proposition \ref{prop:decay2:final} holds with
\begin{equation} \label{eq:decay2:A2''}
	A_{2}'' \leq C_{\Lmb} \calI_{2} + C_{K, \Lmb, M_{i}} A_{1}^{3} + C_{A_{1}, K, \Lmb} A_{1} (A_{2}')^{2}
		+ C_{K, \Lmb} A_{1}^{2} + C_{K, \Lmb} B^{2}
\end{equation}
where $B$ is as in \eqref{eq:decay2:final:aux}.
\end{remark}

\begin{remark} 
According to the argument of this subsection, note that the size of $A_{2}'$ in Proposition \ref{prop:decay2:nullStr} depends on the choice of $u_{2}$ through the term $H_{2}'(u_{2})$, where the size of $u_{2}$ depends on the rate of convergence of $\eps''(u_{2}) \to 0$ as $u_{2} \to \infty$. This explains why $A_{2}$ does not depend only on the size of the initial data, as remarked in Section \ref{sec.main.thm}. On the other hand, as stated in Statement (2) of Theorem \ref{thm:smallData}, we shall show that in the case of small BV initial data, $A_2$ depends only on the size of the initial data. To achieve this, we show in Section \ref{sec:smallData} that we may take $u_{2} = 1$ under this small data assumption. 
\end{remark}

\subsection{Additional decay estimates}
As in the previous section, we conclude this section by providing additional decay rates concerning second derivatives of $\phi$, $r$ and improved decay for $m$ near $\Gmm$.

\begin{corollary} \label{cor:decay2}
Let $(\phi, r, m)$ be a locally BV scattering solution to \eqref{eq:SSESF} with asymptotically flat $C^{1}$ initial data of order $\omg'$.
Let $A_{1}$ and $A_{2}$ be the constants in Theorems \ref{main.thm.1} and \ref{main.thm.2}, respectively. Then 
the following bounds hold.
\begin{align} 
	\abs{\rd_{v} \phi} \leq & C_{\Lmb} (A_{1} + A_{2} + A_{1} A_{2}) \min \set{u^{-(\om+1)}, r^{-2} u^{-(\om-1)}} \label{eq:decay2:5} \\
	\abs{\rd_{u} \phi} \leq & C_{K, \Lmb} (A_{1} + A_{2} + A_{1} A_{2}) \min \set{u^{-(\om+1)}, r^{-1} u^{-\om}} \label{eq:decay2:6} \\
	\abs{\rd_{v}^{2} \phi} \leq & C_{\Lmb} (A_{1} + A_{2} + A_{1} A_{2}) \min \set{r^{-1} u^{-(\om+1)}, r^{-3} u^{-(\om-1)}}, \label{eq:decay2:7}\\
	\abs{\rd_{u} \rd_{v} \phi} \leq & C_{K, \Lmb} (A_{1} + A_{2} + A_{1} A_{2}) \min \set{r^{-1} u^{-(\om+1)}, r^{-2} u^{-\om}}, \label{eq:decay2:8}\\
	\abs{\rd_{u}^{2} \phi} \leq & C_{K, \Lmb} (A_{1} + A_{2} + A_{1} A_{2}) \, r^{-1} u^{-(\om+1)}, \label{eq:decay2:9}\\
	\abs{\rd_{u} \rd_{v} r} \leq & C_{K, \Lmb} (A_{1} + A_{2} + A_{1} A_{2})^{2} \min \set{r u^{-(2\om+2)}, r^{-2} u^{-(2\om-1)}}, \label{eq:decay2:10} \\
	m \leq & C_{K, \Lmb} (A_{1} + A_{2} + A_{1} A_{2})^{2} \min \set{r^{3} u^{-(2\om+2)}, u^{-(2\om-1)}}. \label{eq:decay2:11}
\end{align}
\end{corollary}

This corollary follows immediately from the estimates derived in Theorem \ref{main.thm.2}. We sketch the proof below.

\begin{proof} 
First, note that \eqref{eq:decay2:5} and \eqref{eq:decay2:6} follows from Corollary \ref{cor:decay1}, Theorem \ref{main.thm.2} and Lemma \ref{lem:dphi}. 
Next, \eqref{eq:decay2:7} and \eqref{eq:decay2:9} are easy consequences of the preceding estimates, Theorems \ref{main.thm.1}, \ref{main.thm.2} and the identities
\begin{equation*}
	r \rd_{v}^{2} \phi = \rd_{v}^{2} ( r\phi) - (\rd_{v} \dvr) \phi - 2 \dvr \rd_{v} \phi, \quad
	r \rd_{u}^{2} \phi = \rd_{u}^{2} ( r\phi) - (\rd_{u} \dur) \phi - 2 \dur \rd_{u} \phi.
\end{equation*}

On the other hand, for \eqref{eq:decay2:8}, we use the identity
\begin{equation*}
	r \rd_{u} \rd_{v} \phi = - \dvr \rd_{u} \phi - \dur \rd_{v} \phi,
\end{equation*}
which may be verified from \eqref{eq:SSESF:dr} and \eqref{eq:SSESF:dphi}.

Next, \eqref{eq:decay2:11} follows from Corollary \ref{cor:decay1}, Lemma \ref{lem:muOverR} and \eqref{eq:decay2:5}. Finally, using Corollary \ref{cor:est4dr}, Lemma \ref{lem:est4dur}, \eqref{eq:decay2:11} and the equation \eqref{eq:SSESF:dr}, we conclude \eqref{eq:decay2:10}. \qedhere
\end{proof}

\section{Decay and blow up at infinity}\label{sec.dichotomy}

In this section, we prove Theorem \ref{thm.dichotomy}, i.e., unless the solution blows up at infinity, a `future causally geodesically complete' solution scatters in BV.

Take a BV solution to \eqref{eq:SSESF} satisfying the hypotheses of Theorem \ref{thm.dichotomy}, which does not blow up at infinity. Note, in particular, that $\PD = \regR$ by $(1)$ of Definition \ref{def:locBVScat} and Lemma \ref{lem:regR}. In order to prove Theorem \ref{thm.dichotomy}, our goal is to show that such a spacetime is in fact BV scattering, i.e., (1), (2) and (3) in Definition \ref{def:locBVScat} hold and moreover (3) holds with $R=\infty$.

The main step will be to show that there exists a constant $C_{\Lambda}$ such that for every $\epsilon > 0$, there exists $U$ such that for every $u\geq U$, we have
\begin{equation}\label{dichotomy.goal}
\int_{C_u}|\rd_v^2(r\phi)| + \int_{C_u}|\rd_v \lambda| \leq C_{\Lambda} \epsilon.
\end{equation}
This will be achieved in a sequence of Lemmas and Propositions below.

Before we proceed, we first prove a preliminary bound on $\dvr$:
\begin{proposition}
There exists $0<\Lmb<\infty$ such that
\begin{equation} \label{dic.bnd4dvr}
	\Lmb^{-1} \leq \dvr(u,v) \leq \frac{1}{2}.
\end{equation}
\end{proposition}
\begin{proof}
By $(1)$ in Definition \ref{def.blow.up.infty}, there exists $0 < \Lmb < \infty$ such that $\sup \dvr_{\Gmm}^{-1} \leq \Lmb$. As $\lim_{u \to v-} \dvr_{\Gmm}(u) = \lim_{u \to v-} \dvr(u', v)$ (see \cite[Section 7]{Christodoulou:1993bt}), it follows from Lemma \ref{lem:basicEst4dr} that for every $(u,v) \in \PD$, we have the estimate \eqref{dic.bnd4dvr}.
\end{proof}

We now proceed to show \eqref{dichotomy.goal}. The first step is to show that for $u$ sufficiently large, the integrals along $C_u$ of $|F_1|$ and $|F_2|$ are small. Here, we recall the notation in the proof of Proposition \ref{prop:decay2:nullStr}, i.e.,
\begin{align*}
F_{1} := & \rd_{v}^{2} (r \phi) - (\rd_{v} \dvr) \phi, \\
	F_{2} := & \rd_{v} \log \dvr - \frac{\dvr}{(1-\mu)} \frac{\mu}{r} + \rd_{v} \phi \bb( \dvr^{-1} \rd_{v} (r \phi) - \dur^{-1} \rd_{u} ( r \phi) \bb).
\end{align*}
Once we obtain the desired bounds for $F_1$ and $F_2$, we then derive \eqref{dichotomy.goal} from these bounds. This will be the most technical part (see discussions in Remark \ref{technical.rmk.dichotomy}).

First, we bound the integrals of $F_1$ and $F_2$ in the following proposition:
\begin{proposition}\label{8.1.1}
For every $\ep>0$, there exists $V$ sufficiently large such that the following bound holds for $u\geq V$:
\begin{equation} \label{dichotomy.beginning}
\int_{C_u} (|F_1|+|F_2|)(u,v)  \leq 3\ep.
\end{equation}
\end{proposition}
\begin{proof}
By $(2)$ and $(3)$ in Definition \ref{def.blow.up.infty}, we have
$$\int_1^{\infty}\int_u^{\infty}|\rd_v\lambda\rd_u\phi-  \rd_u\lambda\rd_v\phi  | \ud v \ud u <\infty,$$
and
$$\int_1^{\infty}\int_u^{\infty}|\rd_u\phi\rd_v(\nu^{-1}\rd_u(r\phi))-\rd_v\phi\rd_u(\nu^{-1}\rd_u(r\phi))| \ud v \ud u <\infty.$$

Thus, by choosing $V$ sufficiently large, we have
\be
\int_V^{\infty}\int_u^{\infty}|\rd_v\lambda\rd_u\phi-\rd_v\lambda\rd_v\phi| \ud v \ud u <\ep,\label{non.blowup.1}
\ee
and
\be
\int_V^{\infty}\int_u^{\infty}|\rd_u\phi\rd_v(\nu^{-1}\rd_u(r\phi))-\rd_v\phi\rd_u(\nu^{-1}\rd_u(r\phi))| \ud v \ud u <\ep.\label{non.blowup.2}
\ee

From the initial conditions, we easily see that $F_{1}(1, \cdot)$, $F_{2}(1, \cdot)$ obey $\int_{C_{1}} \abs{F_{1}} + \abs{F_{2}} < \infty$. 
Thus, by choosing $V$ larger if necessary, we have
\be
\int_V^{\infty} (|F_1|+|F_2|)(1,v) \ud v\leq \ep.\label{data.BV.bd}
\ee
Notice that by equations \eqref{eq:decay2:nullStr:pf:1} and \eqref{eq:decay2:nullStr:pf:2}, the estimates \eqref{non.blowup.1} and \eqref{non.blowup.2} control $\iint |\rd_u F_1| \ud u \ud v$ and $\iint |\rd_u F_2| \ud u \ud v$. Thus, we have
$$\int_{\max\set{u, V}}^{\infty} (|F_1|+|F_2|)(u,v) \ud v\leq 3\ep.$$
for every $u\geq 1$. In particular, for $u\geq V$, we have
\begin{equation*}
\int_{C_u} (|F_1|+|F_2|)(u,v)  \leq 3\ep,
\end{equation*}
as desired.
\end{proof}

The inequality \eqref{dichotomy.beginning} is the starting point for our proof of \eqref{dichotomy.goal}. More precisely, our basic strategy is to use a continuous induction on $v$, beginning from the axis, to remove the quadratic and higher terms from \eqref{dichotomy.beginning} and infer \eqref{dichotomy.goal}. 

\begin{remark} \label{technical.rmk.dichotomy}
Before beginning the proof in earnest, we would like to point out two technical nuisances that we confront: 
First, in order to estimate the scalar field $\phi$ itself from $F_{1}$ and $F_{2}$, we need to integrate essentially from null infinity\footnote{More precisely, $\phi$ is determined from $\rd_{v} (r \phi)$, which in turn can be determined from $\int \abs{\rd_{v}^{2}(r \phi)}$ by integrating from $v = \infty$. Another conceptual reason why information near $v = \infty$ is relevant for estimating $\phi$ is that the initial condition $\lim_{v \to \infty} \phi(1, v) = 0$ implies that $\lim_{v \to \infty} \phi(u, v)= 0$ for every $u \geq 1$. See the discussion before \eqref{dic.0.2}.}, which is opposite to the direction of our method of continuity. Second, as $\rd_{v}(r \phi)$ is only assumed to be BV, the left-hand side of \eqref{dichotomy.goal} is not continuous in $v$ in general. To overcome the first, we make use of the invariance of \eqref{eq:SSESF} and $F_{1}$, $F_{2}$ under the change $\phi \mapsto \phi + c$. To take care of the second, we carefully keep track of the evolution of discontinuities of $\rd_{v}(r \phi)$.
\end{remark}

Notice that in order to obtain \eqref{dichotomy.goal} from \eqref{dichotomy.beginning}, we only need to integrate on a \emph{fixed} hypersurface $C_u$. We now fix $u_0\geq V$ and define a new function $\overline{\phi}_{u_{0}}$ by
\begin{equation} \label{dic.0.0}
	\overline{\phi}_{u_{0}}(u,v) := \phi(u,v) - \lim_{v' \to u_{0} +} \phi(u_{0}, v').
\end{equation}

As remarked before, note that \eqref{eq:SSESF} is invariant under the change $(\phi, r, m) \mapsto (\overline{\phi}_{u_{0}}, r, m)$, i.e., $(\overline{\phi}_{u_{0}}, r, m)$ is still a solution to \eqref{eq:SSESF}. Moreover, it is easy to check that $F_{1}$ and $F_{2}$ are also invariant under this change, i.e.,
\begin{equation} \label{eq:inv4F12}
\begin{aligned}
	F_{1} =& \rd_{v}^{2} (r \overline{\phi}_{u_{0}}) - (\rd_{v} \dvr) \overline{\phi}_{u_{0}} \\
	F_{2} =& \rd_{v} \log \dvr - \frac{\dvr}{(1-\mu)} \frac{\mu}{r} +  \rd_{v} \overline{\phi}_{u_{0}} \bb( \dvr^{-1} \rd_{v} (r \overline{\phi}_{u_{0}}) - \dur^{-1} \rd_{u} ( r \overline{\phi}_{u_{0}}) \bb).
\end{aligned}
\end{equation}

The new scalar field has been chosen so that $\overline{\phi}_{u_{0}}(u_{0}, \cdot)$ and $\rd_{v}(r \overline{\phi}_{u_{0}})(u_{0}, \cdot)$ vanish at the axis, i.e.,
\begin{equation} \label{dic.0.1}
\lim_{v \to u_{0}+}\overline{\phi}_{u_{0}}(u_{0},v) = \lim_{v \to u_{0}+}\rd_{v}(r \overline{\phi}_{u_{0}}) (u_{0},v) = \lim_{v \to u_{0}+} \rd_{u}(r \overline{\phi}_{u_{0}})(u_{0},v) = 0.
\end{equation}

We claim that the original scalar field $\phi(u, v)$ obeys the condition 
\begin{equation}\label{claim.vanishing}
\lim_{v \to \infty} \phi(u_{0}, v) = 0
\end{equation} 
for every $u_{0} \geq 1$.
Therefore, by the definition given in \eqref{dic.0.0}, we see that $\phi$ and $\overline{\phi}_{u_0}$ are also related by
\begin{equation} \label{dic.0.2}
	\phi(u,v) = \overline{\phi}_{u_{0}}(u,v) - \lim_{v' \to \infty} \overline{\phi}_{u_{0}}(u, v').
\end{equation}

To establish the claim \eqref{claim.vanishing}, we proceed as in the proof of Lemma~\ref{lem:decay1:cptu:0}, but work with $\phi$ rather than $r \phi$. 
Fix $u_{0} > 1$ and let $r_{1} > 0$ be a large number to be determined. For each $u \geq 1$, let $v^{\star}_{1}(u)$ be the unique $v$-value such that $r(u, v^{\star}_{1}(u)) = r_{1}$. Consider $(u, v) \in \set{1 \leq u \leq u_{0}} \cap \set{r \geq r_{1}}$. Using the uniform bound of $m$ and $\frac{\dvr}{1-\mu}$ in terms of the data at $u= 1$ (which holds thanks to monotonicity), we may integrate \eqref{eq:SSESF:dphi} along the incoming direction to estimate
\begin{equation*}
	\abs{\rd_{v} (r \phi) (u, v) - \rd_{v} (r \phi)(1, v)} 
	\leq \frac{C_{0}}{r(u, v)} \sup_{u' \in [1, u]} \abs{\phi(u', v)},
\end{equation*}
where $C_{0}$ depends only on the data at $u=1$. Integrating both sides in the outgoing direction from $v_{1}^{\star}(u)$ to $v$ (using Lemma~\ref{dic.bnd4dvr} for the right-hand side) and dividing by $r = r(u, v)$, we obtain
\begin{equation} \label{claim.vanishing.key}
\begin{aligned}
	\abs{\phi(u, v)} \leq & \frac{r_{1}}{r} \abs{\phi(u, v^{\star}_{1}(u))} + \frac{r(1, v^{\star}_{1}(u))}{r} \abs{\phi(1, v^{\star}_{1}(u))} \\
	& + \frac{r(1, v)}{r} \abs{\phi(1, v)}		+ \frac{C_{0} \Lmb}{r} \log \bb( \frac{r}{r_{1}}\bb) \sup_{1 \leq u' \leq u, \, v \geq v^{\ast}_{1}(u)} \abs{\phi}.
\end{aligned}
\end{equation}
Now the idea is to use \eqref{claim.vanishing.key} to first show that $\phi$ is bounded on the region $\set{1 \leq u \leq u_{0}}$, and then use \eqref{claim.vanishing.key} again with the additional boundedness of $\phi$ to conclude that \eqref{claim.vanishing} holds.
To begin with, observe that $\phi$ is bounded on each set compact subset of $\PD$, since it is a BV solution in the sense of Definition~\ref{def:BVsolution}. Combined with the hypothesis that $\phi(1, v) \to 0$ as $v \to \infty$, we see that the first three terms are bounded by a constant that depends on $r_{1}$. On the other hand, by taking $r_{1}$ sufficiently large, the coefficient $(C_{0} \Lmb / r) \log ( r / r_{1} )$ of the last term can be made arbitrarily small for $r \geq r_{1}$. This smallness allows us to absorb the last term to the left-hand side, and conclude the desired boundedness of $\phi$ on the region $\set{1 \leq u \leq u_{0}}$.
Then plugging in $u = u_{0}$ and the uniform bound for $\phi$ into \eqref{claim.vanishing.key}, the claim \eqref{claim.vanishing} follows from the hypothesis $\lim_{v \to \infty} \phi(1, v)  = 0$.

Let
\beaa
I_1(u, v)&:=&\int_{u}^v |\rd_v^2(r \overline\phi_{u_{0}})| (u,v')dv'  ,\quad I_2(u, v):=\int_{u}^v |\rd_v\lambda | (u,v')dv' .
\eeaa

In the following two lemmas, we will show that
\begin{align} 
I_1(u_0, v) \leq & 3\epsilon+ C_{\Lambda} I_{1}(u_0,v) I_2(u_0,v) , \label{main.ineq.dichotomy.1} \\
I_2(u_0, v)\leq & 3\ep+ C_{\Lmb} I_{1}(u_0,v)^{2} (1+I_{1}(u_0,v))^{2} (1+I_{2}(u_0,v))^{2} e^{C_{\Lmb} I_{1}(u_0,v)^{2} (1+I_{2}(u_0,v))} \label{main.ineq.dichotomy.2}
\end{align}
for every $V \leq u_0 \leq v$, with $C_{\Lmb}$ independent of $u_0$ and $v$.

\begin{lemma} \label{lem.dic.1}
There exists a constant $C_{\Lmb} > 0$ such that for every $V \leq u_0 \leq v$,
\begin{equation*}
I_1(u_0,v) \leq 3\epsilon+ C_{\Lambda} I_{1}(u_0,v) I_2(u_0,v).
\end{equation*}
\end{lemma}
\begin{proof}
In this proof, we fix $u_0 \geq V$ and use the abbreviations
\begin{equation} \label{eq:abbrev-overline-phi}
\overline{\phi} := \overline{\phi}_{u_{0}} , \quad
\rd_{v}(r \overline{\phi}) := \rd_{v}(r \overline{\phi}_{u_{0}}) 
\ \hbox{ and } \ 
\rd_{v}^{2}(r \overline{\phi}) := \rd_{v}^{2}(r \overline{\phi}_{u_{0}}).
\end{equation}
By Lemma \ref{lem:est4phi}, we have
\begin{equation} \label{phi.est}
	\abs{\overline{\phi}(u_0,v)} 
	\leq \frac{1}{r} \int_{u_0}^{v} \rd_{v}(r \overline{\phi}) (u_0, v') \, \ud v' 
	\leq \Lmb \sup_{u_0 \leq v' \leq v} \abs{\rd_{v}(r \overline{\phi})(u_0,v')}.
\end{equation}

By the fundamental theorem of calculus and \eqref{dic.0.1}, note that
\begin{equation} \label{dic.1.0}
\sup_{u_0 \leq v' \leq v} |\rd_{v}(r \overline{\phi})(u_0,v')|\leq I_1(u_0,v).
\end{equation}

Thus, recalling the definition of $F_1$ in \eqref{eq:inv4F12}, we have
\begin{equation*}
\begin{split}
I_1(u_0,v) \leq &\int_{u_0}^v |F_1(u_0,v')|dv'+\int_{u_0}^v|\rd_v\dvr||\overline{\phi}|(u_0,v')dv' \\
\leq &\int_{u_0}^v |F_1(u_0,v')|dv'+ \Lmb\,  I_1(u_0,v) I_2(u_0,v) 
		\leq  3\ep+C_{\Lambda} I_{1}(u_0,v) I_2(u_0,v). \qedhere
\end{split}
\end{equation*}
\end{proof}

We now move on to estimate $I_2(u_0,v)$.
\begin{lemma} \label{lem.dic.2}
There exists a constant $C_{\Lmb} > 0$ such that for every $V \leq u_0 \leq v$,
\begin{equation*} 
I_2(u_0,v)\leq 3\ep+ C_{\Lmb} I_{1}(u_0,v)^{2} (1+I_{1}(u_0,v))^{2} (1+I_{2}(u_0,v))^{2} e^{C_{\Lmb} I_{1}(u_0,v)^{2} (1+I_{2}(u_0,v))}. 
\end{equation*}
\end{lemma}
\begin{proof}
Again, we fix $u_0 \geq V$ and use the abbreviation \eqref{eq:abbrev-overline-phi}, as well as
\begin{equation} \label{eq:abbrev-du-overline-phi}
	\rd_{u} (r \overline{\phi}) (u,v) := \rd_{u} (r \overline{\phi}_{u_{0}})(u,v) .
\end{equation}
Recalling the equation for $F_2$ in \eqref{eq:inv4F12}, in order to control $I_2(u_0,v)$ from $F_2$, we need to estimate $\int_{u_0}^{v} (\frac{\lambda}{1-\mu} \frac{\mu}{r})(u_0,v') \ud v'$ and $\int_{u_0}^v \rd_{v} \overline{\phi} ( \dvr^{-1} \rd_{v} (r \overline{\phi}) - \dur^{-1} \rd_{u} ( r \overline{\phi})) (u_0,v') \ud v'.$
By Lemma \ref{lem:auxEqs},
$$\int_{u_0}^{v} (\frac{\lambda}{1-\mu} \frac{\mu}{r})(u_0,v') \ud v' = \log (1-\mu(u_0,v))+\int_{u_0}^v \frac{r(\rd_v \overline{\phi})^2}{\lambda}(u_0,v') \ud v'.$$
Since $\PD =\regR$, the integrand on the left-hand side is non-negative. Notice furthermore that since $\mu \geq 0$, $\log (1-\mu(u_0,v)) <0$. Thus, 
\beaa
\int_{u_0}^{v} (\frac{\lambda}{1-\mu} \frac{\mu}{r})(u_0,v') \ud v'
&\leq & |\int_{u_0}^v \frac{r(\rd_v \overline{\phi})^2}{\lambda}({u_0},v') \ud v'|\\
&\leq & \int_{u_0}^v |\rd_v \overline{\phi}({u_0},v')||(\frac{\rd_v(r\overline{\phi})}{\lambda}- \overline{\phi})({u_0},v')| \ud v'\\
&\leq & 2 \Lmb I_{1}({u_0},v)  \int_{u_0}^v |\rd_v \overline{\phi}({u_0},v')| \, \ud v' ,
\eeaa
where we have used \eqref{phi.est} and \eqref{dic.1.0} on the last line. Using Lemma \ref{lem:est4dvphi}, we estimate the integral on the last line by
\begin{equation} \label{dic.2.0}
\int_{u_0}^v |\rd_v \overline{\phi}({u_0},v')|dv' \leq
\int_{u_0}^v|\rd_v(\lambda^{-1}\rd_v(r \overline{\phi}))({u_0},v')| \ud v' ,
\end{equation}
and the right hand side can in turn be estimated using \eqref{dic.1.0} by
\bea
\int_{u_0}^v|\rd_v(\lambda^{-1}\rd_v(r \overline{\phi}))({u_0},v')| \ud v' 
&\leq &\int_{u_0}^v \lambda^{-1}|\rd_v^2(r \overline{\phi})({u_0},v')| \ud v'+\int_{u_0}^v \lambda^{-2}|\rd_v\lambda\rd_v(r \overline{\phi})({u_0},v')| \ud v'\notag\\
&\leq &\Lambda I_1({u_0},v)+ \Lambda^2 I_{1}({u_0},v) I_2({u_0},v). \notag
\eea

Therefore, we have
\be\label{dic.2.1}
\int_{u_0}^{v} (\frac{\lambda}{1-\mu} \frac{\mu}{r})({u_0},v') \ud v'\leq  C_{\Lambda} I_1({u_0},v)^{2} (1+ I_{2}({u_0},v)).
\ee

We now move on to bound $\int_{u_0}^v \rd_{v} \overline{\phi} \, \dvr^{-1} \rd_{v} (r \overline{\phi}) ({u_0}, v') \ud v'$. Using \eqref{dic.1.0} and \eqref{dic.2.0}, we easily estimate
\begin{align}
\int_{u_0}^v \abs{\rd_{v} \overline{\phi} \, \dvr^{-1} \rd_{v} ( r \overline{\phi}) ({u_0},v')} \ud v' \notag
\leq & \Lmb \int_{{u_0}}^{v} \abs{\rd_{v} \overline{\phi}}({u_0}, v') \, \ud v' \sup_{{u_0} \leq v' \leq v} \abs{\rd_{v} ( r \overline{\phi}) ({u_0},v')} \\
\leq & C_{\Lmb} I_{1}({u_0},v)^{2} (1+I_{2}({u_0},v)). \label{dic.2.2}
\end{align}

Finally, we are only left to bound $- \int_{u_0}^v \rd_{v} \overline{\phi} \, \dur^{-1} \rd_{u} ( r \overline{\phi}) ({u_0},v') \ud v'$. As before, we begin by estimating
\begin{align}
\int_{u_0}^v \abs{\rd_{v} \overline{\phi} \, \dur^{-1} \rd_{u} ( r \overline{\phi}) ({u_0},v')} \ud v'
\leq & \int_{{u_0}}^{v} \abs{\rd_{v} \overline{\phi}}({u_0}, v') \, \ud v' \sup_{{u_0} \leq v' \leq v} \abs{\dur^{-1} \rd_{u} ( r \overline{\phi}) ({u_0},v')} \notag \\
\leq & C_{\Lmb} I_{1}({u_0},v) (1+I_{2}({u_0},v)) \sup_{{u_0} \leq v' \leq v} \abs{\dur^{-1} \rd_{u} ( r \overline{\phi}) ({u_0},v')}. \label{dic.2.3}
\end{align}

In this case, we do not wish to pull out $\dur$ as we have not assumed any bound on it. Instead, we consider $\dur^{-1} \rd_{u} ( r \overline{\phi})$ as a whole and note that
\begin{equation}\label{inv.wave.eqn}
\rd_{v}(\dur^{-1} \rd_{u} ( r \overline{\phi}) ) 
= 	- \bb( \frac{\dvr}{1-\mu} \frac{\mu}{r} \bb) \dur^{-1} \rd_{u} ( r \overline{\phi}) 
	+ \bb( \frac{\dvr}{1-\mu} \frac{\mu}{r} \bb) \overline{\phi}.
\end{equation}
The equation \eqref{inv.wave.eqn} holds since by \eqref{eq:SSESF:dr} and \eqref{eq:SSESF:dphi}, we have 
\begin{equation*}
\rd_{v}(\dur^{-1} \rd_{u} ( r {\phi}) ) 
= 	- \bb( \frac{\dvr}{1-\mu} \frac{\mu}{r} \bb) \dur^{-1} \rd_{u} ( r {\phi}) 
	+ \bb( \frac{\dvr}{1-\mu} \frac{\mu}{r} \bb) {\phi}
\end{equation*}
and moreover both the left hand side and the right hand side of the equation are invariant under the transformation $\phi\to \phi+c$.

Therefore, by the variation of constants formula and \eqref{dic.0.1}, we have
\begin{equation*}
	\dur^{-1} \rd_{u} (r \overline{\phi})(u_0, v)
	= e^{-J({u_0},v)} \int_{{u_0}}^{v} e^{J({u_0},v')}  \frac{\dvr}{1-\mu} \frac{\mu}{r} \,  \overline{\phi} ({u_0}, v') \, \ud v',
\end{equation*}
where
\begin{equation*}
	J({u_0},v) := \int_{{u_0}}^{v} \frac{\dvr}{1-\mu} \frac{\mu}{r} ({u_0}, v') \, \ud v'.
\end{equation*}

By \eqref{dic.1.0} and \eqref{dic.2.1}, we have
\begin{equation*}
	\sup_{{u_0} \leq v' \leq v}\abs{\dur^{-1} \rd_{u} (r \overline{\phi}) ({u_0},v')} 
	\leq C_{\Lmb} I_{1}({u_0},v)^{3} (1+I_{2}({u_0},v)) e^{C_{\Lmb} I_{1}({u_0},v)^{2} (1+I_{2}({u_0},v))}.
\end{equation*}

Then by \eqref{dic.2.3}, we conclude that 
\begin{equation} \label{dic.2.4}
\int_{u_0}^v \abs{\rd_{v} \overline{\phi} \, \dur^{-1} \rd_{u} ( r \overline{\phi}) ({u_0},v')} \ud v'
	\leq C_{\Lmb} I_{1}({u_0},v)^{4} (1+I_{2}({u_0},v))^{2} e^{C_{\Lmb} I_{1}({u_0},v)^{2} (1+I_{2}({u_0},v))}.
\end{equation}

Combining \eqref{dic.2.1}, \eqref{dic.2.2} and \eqref{dic.2.4}, we conclude that \eqref{main.ineq.dichotomy.2} holds. \qedhere
\end{proof}

Next, we apply \eqref{main.ineq.dichotomy.1} and \eqref{main.ineq.dichotomy.2} to show that 
\begin{proposition} \label{dic.main.prop}
For ${u_0}$ sufficiently large and $v\geq {u_0}$, we have
$$I_1({u_0},v)+I_2({u_0},v)\leq C_{\Lambda} \epsilon.$$
\end{proposition}
\begin{remark}
If it is the case that $\int_{u_0}^v(|\rd_v^2(r\overline{\phi})| + |\rd_v \lambda|)({u_0},v')dv'$ is continuous in $v$ for each fixed ${u_0}$, then the desired conclusion follows from \eqref{main.ineq.dichotomy.1} and \eqref{main.ineq.dichotomy.2} via a simple continuity argument in $v$. In particular, the conclusion follows in the case where the initial data of $\rd_v(r \phi)$ are in $W^{1,1}$ or $C^{1}$. The only remaining difficulty is therefore to control the size of the delta function singularities in $\rd_v^2(r\overline{\phi})$ in the general case where we only have a BV solution.
\end{remark}
\begin{proof}
We begin by studying the propagation of discontinuities for a BV solution to \eqref{eq:SSESF}. In the general case where $\rd_v(r \phi)(1, \cdot)$ is only in BV and contains jump discontinuities (at which $\rd_{v}(r \phi)(1, \cdot)$ is assumed to be right-continuous), notice that the jump discontinuities for a BV function are discrete, i.e., they occur only at a (possibly infinite) sequence of points $V<v_1<v_2< v_3 < \cdots$. On the other hand, note that by the initial condition $r = 2v$ on $C_{1}$, we have $\dvr(1,v) = \frac{1}{2}$; in particular, $\dvr$ is continuous initially.

Thanks to the initial condition, it follows that $\dvr$ does not possess any discontinuities outside $\Gmm$. Indeed, from the definition of a BV solution, $m$ and $r$ are continuous. Then, by \eqref{eq:SSESF:dr}, we see that that $\dur$ is Lipschitz in the $v$ direction outside of $\Gmm$, with bounded Lipschitz constant on each compact interval of $u$. Looking back at \eqref{eq:SSESF:dr} and recalling that $\dvr(1,v) = \frac{1}{2}$, we then see that $\dvr$ does not possess any discontinuities outside $\Gmm$ as desired. Since $\dvr$ is a priori in BV, it follows that $\int_{u_0}^{v} \abs{\rd_{v} \dvr({u_0},v')} \, \ud v'$ is continuous in $v \in ({u_0}, \infty)$ with $\int_{{u_0}}^{v} \abs{\rd_{v} \dvr({u_0},v')} \, \ud v' \to 0$ as $v \to {u_0}+$. 

By the above regularity statements and equation \eqref{eq:SSESF:dphi}, as well as the fact that $\phi$ is continuous outside $\Gmm$ by the definition of a BV solution, it now follows that the jump discontinuities of $\rd_{v}(r \phi)$ are propagated along constant $v_{i}$ curves. Therefore, for ${u_0} \geq V$, we see that $\rd_v (r \phi)({u_0},v)$ is a right-continuous BV function on $({u_0}, \infty)$ with jump discontinuities at ${u_0} < v_1 < v_2 < v_3 < \cdots$ with the same sizes as $\rd_{v} (r \phi)(1,v)$. 
From \eqref{dic.0.0}, notice that, using the abbreviations in \eqref{eq:abbrev-overline-phi} and \eqref{eq:abbrev-du-overline-phi},
\begin{equation*}
	\int_{{u_0}}^{v} \abs{\rd_{v}^{2} (r \overline{\phi})({u_0},v')} \, \ud v' = \int_{{u_0}}^{v} \abs{(\rd_{v}^{2} (r \phi) - c \rd_{v} \dvr)({u_0},v')} \, \ud v',
\end{equation*}
for the constant $c = \lim_{v \to {u_0}+} \phi({u_0},v)$, which is independent of $v$. Thanks to the continuity property of $\dvr({u_0}, \cdot)$, we see that the integral of $\abs{\rd_{v}^{2}(r \overline{\phi})({u_0}, \cdot)}$ has the same jump discontinuities as $\abs{\rd_{v}^{2} (r \phi)({u_0}, \cdot)}$. In particular, by \eqref{data.BV.bd}, each jump of $I_{1}({u_0},v)$ is at most of size $\eps$.

Fix $C_{\Lambda} > 1$ to be larger than the maximum of the constants from \eqref{main.ineq.dichotomy.1} and \eqref{main.ineq.dichotomy.2}. First, a standard continuity argument using \eqref{main.ineq.dichotomy.1} and \eqref{main.ineq.dichotomy.2} implies that if 
$$\lim_{v\to v_i+}\int_{u_0}^v |\rd_v^2(r\overline{\phi}) ({u_0},v')| \ud v'\leq 5 C_{\Lambda}\ep
\hbox{ and }
\lim_{v\to v_i+}\int_{u_0}^v |\rd_v \lambda({u_0},v')| \ud v'\leq 5\ep,$$
(with the convention $v_{0} := {u_0}$) then
$$\int_{u_0}^v |\rd_v^2(r\overline{\phi}) ({u_0},v')| \ud v'\leq 4 C_{\Lambda}\ep
\hbox{ and }
\int_{u_0}^v |\rd_v \lambda({u_0},v')| \ud v'\leq 4 \ep,$$
for $v_i< v< v_{i+1}$.

Assume, for the sake of contradiction, that the conclusion of the proposition is not satisfied. Recall that the integral of $\abs{\rd_v\lambda}$ is continuous. Thus, we have that for some $v_i$ with $i > 0$
$$\lim_{v\to v_i-}\int_{u_0}^v |\rd_v^2(r\overline{\phi}) ({u_0},v')| \ud v'\leq 4 C_{\Lambda}\ep,$$
holds, but at the same time
$$\lim_{v\to v_i+}\int_{u_0}^v |\rd_v^2(r\overline{\phi}) ({u_0},v')| \ud v' > 5 C_{\Lambda}\ep.$$

However, we have seen that the size of the jump in $I_1({u_0},v)$ is bounded by $\ep$, which is smaller than $C_{\Lambda} \ep$ if $C_{\Lambda}>1$. This leads to a contradiction and thus the conclusion of the proposition holds.
\end{proof}

We are now ready to conclude the proof of Theorem \ref{thm.dichotomy}.
\begin{proof} [Conclusion of Proof of Theorem \ref{thm.dichotomy}]
We first establish \eqref{dichotomy.goal}. In what follows, we use the abbreviations in \eqref{eq:abbrev-overline-phi} and \eqref{eq:abbrev-du-overline-phi}, such as $\overline{\phi} = \overline{\phi}_{u_{0}}$ etc. The idea is to transform back to $(\overline{\phi}, r, m) \mapsto (\phi, r, m)$ using \eqref{dic.0.2}. Note that $\abs{\rd_{v} \dvr}$ remains the same under this change, so it suffices to estimate $\abs{\rd_{v}^{2} (r \phi)}$. By \eqref{dic.2.0} and Proposition \ref{dic.main.prop}, for sufficiently large ${u_0}$, the limit $\overline{\phi}({u_0}, \infty) := \lim_{v \to \infty} \overline{\phi}({u_0},v)$ exists and satisfies
\begin{equation*}
	\abs{\overline{\phi}({u_0}, \infty)} \leq C_{\Lmb} \eps,
\end{equation*}
where we note that $C_{\Lmb}$ is independent of ${u_0}$.

By \eqref{dic.0.2}, we have $\phi(u,v) = \overline{\phi}(u,v) - \overline{\phi}(u, \infty)$ for all $u$. Thus, using Proposition \ref{dic.main.prop}, we estimate
\begin{align*}
	\int_{{u_0}}^{\infty} \abs{\rd_{v}^{2} ( r\phi)({u_0},v)} \, \ud v
	= & \int_{{u_0}}^{\infty} \abs{\rd_{v}^{2} ( r \overline{\phi}({u_0},v) - r \overline{\phi}({u_0}, \infty))} \, \ud v \\
	\leq & \int_{{u_0}}^{\infty} \abs{\rd_{v}^{2} (r \overline{\phi})({u_0},v)} \, \ud v 
		+ \abs{\overline{\phi}({u_0},\infty)} \int_{{u_0}}^{\infty} \abs{\rd_{v} \dvr ({u_0},v)} \, \ud v \\
	\leq & C_{\Lmb} (\eps + \eps^{2}).
\end{align*}
Since $u_0\geq V$ is arbitrary, this proves \eqref{dichotomy.goal}.

Finally, we prove that the conditions $(2)$ and $(3)$ of Definition \ref{def:locBVScat} hold.
Indeed, since $\rd_{v} \log \dvr = \dvr^{-1} \rd_{v} \dvr$, $(3)$ in Definition \ref{def:locBVScat} follows from \eqref{dichotomy.goal} and \eqref{dic.bnd4dvr}. In fact, $(3)$ in Definition \ref{def:locBVScat} holds with arbitrarily large $R > 0$. Next, by \eqref{eq:SSESF:dm}, non-negativity of $1-\mu$ and $\mu$ (by Lemma \ref{lem:mntn4r}) and the fact that $m$ is invariant under $\phi \mapsto \overline{\phi}$,
\begin{equation*}
	m(u_0, v) 
	\leq \frac{1}{2} \sup_{u_0 \leq v' \leq v} \abs{(\dvr^{-1} \rd_{v}(r \overline{\phi}) - \overline{\phi})(u_0, v')} \int_{u_0}^{v} \abs{\rd_{v} \overline{\phi}(u_0, v')} \, \ud v'
\end{equation*}
where the right-hand side is $\leq C_{\eps, \Lmb} \eps$ (with $C_{\eps, \Lmb}$ non-decreasing in $\eps$) by the estimates proved so far. Therefore, $(2)$ of Definition \ref{def:locBVScat} follows.
This concludes the proof of Theorem \ref{thm.dichotomy}. \qedhere
\end{proof}

\section{Refinement in the small data case} \label{sec:smallData}
In this section, we sketch a proof of Theorem \ref{thm:smallData}. The idea is to revisit the proofs of the main theorems (Theorems \ref{main.thm.1} and \ref{main.thm.2}), and notice that all the required smallness can be obtained by taking initial total variation of $\rd_{v}(r \phi)$ small. Key to this idea is the following lemma.

\begin{lemma} \label{lem:smallData}
There exist universal constants $\eps_{0}$ and $C_{0}$ such that for $\eps < \eps_{0}$, the following holds. Suppose that $\dvr(1, \cdot) = \frac{1}{2}$ and $\rd_{v}(r \phi)(1, \cdot)$ is of bounded variation with
\begin{equation} \label{eq:smallData:hyp}
	\int_{C_{1}} \abs{\rd_{v}^{2} (r \phi)} < \eps.
\end{equation}

Suppose furthermore that $\lim_{v \to \infty} \phi(1, v) = 0$. Then the maximal development  $(\phi, r, m)$ satisfies condition $(1)$ of Definition \ref{def:locBVScat} (future completeness of radial null geodesics) and obeys
\begin{gather}
	\sup_{v \in [1, \infty)} \int_{\uC_{v}} \abs{\frac{\mu }{(1-\mu)} \frac{\dur}{r}} 
	+ \sup_{u \in [1, \infty)} \int_{C_{u}} \abs{\frac{\mu }{(1-\mu)} \frac{\dvr}{r}} \leq C_{0} \eps^{2},	\label{eq:smallData:smallPtnl} \\
	\sup_{v \in [1, \infty)} \int_{\uC_{v}} \abs{\rd_{u} \phi}
	+ \sup_{u \in [1, \infty)} \int_{C_{u}} \abs{\rd_{v} \phi} \leq C_{0} \eps,				 \label{eq:smallData:smallDphi} \\
	\sup_{v \in [1, \infty)} \int_{\uC_{v}} (\abs{\rd_{u}^{2} (r \phi)} + \rd_{u} \log \dur)
	+ \sup_{u \in [1, \infty)} \int_{C_{u}} (\abs{\rd_{v}^{2} (r \phi)} + \rd_{v} \log \dvr) \leq C_{0} \eps.				 \label{eq:smallData:smallTV} 
\end{gather}

Moreover, the bounds in Proposition \ref{prop:geomLocBVScat} hold with
\begin{equation}  \label{eq:smallData:geom}
	K + \Lmb \leq C_{0}, \quad
	\Psi \leq C_{0} \eps.
\end{equation}
\end{lemma}

\begin{proof} 
This lemma is an easy consequence of Theorem \ref{Chr.BV.Thm} and Lemma \ref{lem:auxEqs} once we show 
\begin{equation*}
\sup_{\PD} \abs{\rd_{v} (r \phi)} \leq C_{0} \eps,
\end{equation*}
using the additional condition $\lim_{v \to \infty} \phi(1, v) = 0$. By Lemma \ref{lem:est4dvphi}, note that $\int_{C_{1}} \abs{\rd_{v} \phi} \leq C \eps$; therefore, integrating from $v = \infty$, we have $\lim_{v \to 1+} \abs{\phi(1, v)} \leq C \eps$. Then using \eqref{eq:smallData:hyp} to integrate from $v = 1$, where we note that $\lim_{v \to 1+}\phi(1, v) = \lim_{v \to 1+} \rd_{v}(r \phi)(1, v)$, we obtain
\begin{equation*}
	\sup_{C_{1}} \abs{\rd_{v} (r \phi)} \leq C \eps.
\end{equation*}

Using \eqref{eq:SSESF:dphi'}, $\rd_{u} \dvr \leq 0$, Lemma \ref{lem:est4phi} (to control $\abs{\phi}$ from $\abs{\rd_{v}(r \phi)}$) and $\frac{1}{3} \leq \dvr \leq \frac{1}{2}$ (by Theorem \ref{Chr.BV.Thm}), it follows that
\begin{align*}
	\sup_{\calD(1, v)}\abs{\rd_{v}(r \phi)} 
	\leq & \sup_{1 \leq v' \leq v} \abs{\rd_{v} (r \phi)(1, v')} + \sup_{(u,v) \in \calD(1,v)} \sup_{1 \leq u' \leq u} \abs{\phi(u', v)} \int_{1}^{u} (- \rd_{u} \dvr)(u', v) \, \ud u' \\
	\leq & C \eps + \frac{1}{2} \sup_{\calD(1,v)} \abs{\rd_{v}(r \phi)},
\end{align*} 
which proves $\sup_{\PD} \abs{\rd_{v} (r \phi)} \leq C_{0} \eps$, as desired. \qedhere
\end{proof}

Equipped with Lemma \ref{lem:smallData}, we now proceed to outline the proof of Theorem \ref{thm:smallData}.
\begin{proof} [Proof of $(1)$ in Theorem \ref{thm:smallData}]
That $(\phi, r, m)$ is globally BV scattering follows from Theorem \ref{thm.dichotomy} and the fact that initial data with small total variation cannot lead to a development which blows up at infinity; the latter fact follows from the results proved by Christodoulou \cite[Section 4, Theorem 6.2]{Christodoulou:1993bt}. 

It remains to prove that \eqref{eq:decay1:1}--\eqref{eq:decay1:3} hold with $A_{1} \leq C_{\calI_{1}} (\calI_{1} + \eps)$ if $\eps > 0$ is sufficiently small. By \eqref{eq:decay1:H1}, it follows Lemma \ref{lem:decay1:cptu:0} holds with $H_{1} \leq C_{\calI_{1}} (\calI_{1} + \eps)$, and \eqref{eq:decay1:extr} in Lemma \ref{lem:decay1:extr} becomes�
\begin{equation*} \tag{\ref{eq:decay1:extr}$'$}
	\sup_{C_{u}} r^{\omg} \abs{\rd_{v} (r \phi)} \leq C_{\calI_{1}} u_{1} (\calI_{1} + \eps) + C M_{i} u_{1}^{-1} \calB_{1}(U),
\end{equation*}

Note that $M_{i} \leq C \calI_{1}^{2}$. Then repeating the arguments in \S \ref{sec.full.decay.1}, we see that \eqref{eq:decay1:intr:pf:key} becomes
\begin{equation*} \tag{\ref{eq:decay1:intr:pf:key}$'$}
	\calB_{1}(U) \leq C_{\calI_{1}} u_{1} (\calI_{1} + \eps) + C (\calI_{1}^{2} u_{1}^{-1} + \eps^{2}) \calB_{1}(U). 
\end{equation*}

It is important to note that the constant $C$ in the last term does not depend on $\calI_{1}$. Take $u_{1} = 1000 C (1+\calI_{15})^{2}$. Then for $\eps > 0$ sufficiently small (independent of $\calI_{1}$), we derive
\begin{equation*}
	\calB_{1}(U) \leq C_{\calI_{1}} (\calI_{1} + \eps)
\end{equation*}

It then follows that \eqref{eq:decay1:1} and \eqref{eq:decay1:2} hold with $A_{1} \leq C_{\calI_{1}} (\calI_{1} + \eps)$. Applying Lemma \ref{lem:decay1:uDecay4durphi}, we conclude that \eqref{eq:decay1:3} holds with $A_{1} \leq C_{\calI_{1}} (\calI_{1} + \eps)$ as well. \qedhere
\end{proof}

\begin{proof} [Proof of $(2)$ in Theorem \ref{thm:smallData}]
We need to prove that \eqref{eq:decay2:1}--\eqref{eq:decay2:4} hold with $A_{2} \leq C_{\calI_{2}} (\calI_{2} + \eps)$. The key is to show that Proposition \ref{prop:decay2:nullStr} holds with
\begin{equation} \label{eq:smallData:pf2:key}
	A_{2}' \leq C_{\calI_{2}} (\calI_{2} + \eps)
\end{equation}

Indeed, by the explicit bounds on the constants (in particular, \eqref{eq:decay2:rDecay:1}, \eqref{eq:decay2:rDecay:2}, \eqref{eq:decay2:uDecayInExtr:1}, \eqref{eq:decay2:uDecayInExtr:2}, \eqref{eq:decay2:final:aux} and \eqref{eq:decay2:A2''}), the desired conclusion easily follows once \eqref{eq:smallData:pf2:key} is established.

Note that $\calI_{1} \leq \calI_{2}$ by definition, and thus $A_{1} \leq C_{\calI_{2}} (\calI_{2} + \eps)$ by the preceding proof. We furthermore claim that the following statements hold:
\begin{itemize}
\item Lemma \ref{lem:decay2:key4nullStr} holds with
\begin{equation} \label{eq:smallData:pf2:1}
	 \eps(u_{2}) \leq C \eps,
\end{equation}
for every $u_{2} \geq 1$.

\item We have
\begin{equation} \label{eq:smallData:pf2:2}
	 H_{2}'(1) \leq C_{\calI_{2}} (\calI_{2} + \eps),
\end{equation}
where we remind the reader that
\begin{equation*}
H_{2}'(1) = \sup_{\set{(u,v) : u \in [1, 3] \, v \in [u, 3u]}} u^{\omg} \bb( \abs{\rd_{v}^{2}(r \phi)} + \abs{\rd_{u}^{2}(r \phi)} + \abs{\rd_{v} \dvr} + \abs{\rd_{u} \dur}\bb)
\end{equation*}
according to \eqref{eq:decay2:def4H'2}.
\end{itemize}

The first claim follows easily from Lemma \ref{lem:smallData} and \eqref{eq:decay2:eps}. For the second claim, since $1 \leq u \leq 3$, it suffices to prove\footnote{Recall that $\calD(1,9) = \set{(u,v) : u \in [1, 3] \, v \in [u, 3u]}$ is the domain of dependence of $C_{1} \cap \set{1 \leq v \leq 9}$.}
\begin{equation*}
	 \sup_{\calD(1,9)} \bb( \abs{\rd_{v}^{2}(r \phi)} + \abs{\rd_{u}^{2}(r \phi)} + \abs{\rd_{v} \dvr} + \abs{\rd_{u} \dur}\bb)\leq C (\calI_{2} + \eps),
\end{equation*}
which follows from a persistence of regularity argument, similar to our proof of Lemma \ref{lem:decay2:key4nullStr}. 

To conclude the proof, recall that we had
\begin{equation*}
	\calB_{2}(U) \leq H_{2}''(u_2) + \eps''(u_{2}) \calB_{2}(U)
\end{equation*}
where $\calB_{2}(U)$ was defined in \eqref{eq:decay2:def4B2}, and $H_{2}''(u_2)$, $\eps''(u_{2})$ obey the bounds in \eqref{eq:decay2:H2''} and \eqref{eq:decay2:eps''} respectively. Thanks to \eqref{eq:smallData:pf2:1}, it follows that we may take $u_{2} = 1$ and $\eps''(1) \leq C\eps$, where $C$ does not depend on $\calI_{2}$. Next, since $u_{2} = 1$, we see that $H_{2}''(1) \leq C_{\calI_{2}} (\calI_{2} + \eps)$ by \eqref{eq:smallData:pf2:2}. Therefore, for $\eps > 0$ sufficiently small (independent of $\calI_{2}$), we conclude that
\begin{equation*}
	\calB_{2}(U) \leq C_{\calI_{2}} (\calI_{2} + \eps),
\end{equation*}
which proves that Proposition \ref{prop:decay2:nullStr} holds with \eqref{eq:smallData:pf2:key}, as desired.  \qedhere
\end{proof}

\section{Optimality of the decay rates} \label{sec.opt}
In this section, we show the optimality of the decay rates obtained above, i.e., we prove Theorems \ref{thm.opt.1} and \ref{thm.opt.2}. 

\subsection{Optimality of the decay rates, in the case $1 < \omg' < 3$}  \label{subsec.opt.1}
In this subsection, we prove Theorem \ref{thm.opt.1}. More precisely, we will demonstrate that the proof of the upper bounds for $\phi$ and its derivatives can in fact be sharpened to give also lower bounds for $\rd_v(r\phi)$ and $\rd_u(r\phi)$ if the initial data satisfy appropriate lower bounds for $\om<3$.

\begin{proof}[Proof of Theorem \ref{thm.opt.1}]

We first prove the lower bound for $\rd_v(r\phi)$. We split the spacetime into the exterior region $\extr$ and interior region $\intr$ as before. Notice that in the exterior region, $u\lesssim r$ and it suffices to prove a lower bound for $r^\om\rd_v(r\phi)$. Similarly, in the interior region, $r\lesssim u$ and it suffices to prove a lower bound for $u^\om\rd_v(r\phi)$.

Revisiting the proof of Lemma \ref{lem:decay1:extr}, we note that instead of controlling $\rd_v(r\phi)$ by the initial data and error terms, we can bound the difference between $\rd_v(r\phi)(u,v)$ and the corresponding initial value of $\rd_v(r\phi)(1,v)$. More precisely, from the proof of Lemma \ref{lem:decay1:extr}, we have

\begin{align*}
	\abs{\rd_{v} (r \phi)(u,v)-\rd_{v} (r \phi)(1, v)}
	\leq  \frac{u_{1} K M_{i}}{r^2(u,v)(1+r(u,v))} H_{1} + \frac{K M(u_{1})}{u_{1} r^{\om}(u,v)} \calB_{1}(U)
\end{align*}
in the case $2<\om<3$ and
\begin{align*}
	\abs{\rd_{v} (r \phi)(u,v)-\rd_{v} (r \phi)(1, v)}
		\leq \frac{\om K M_{i}}{r(u,v)(1+r(u,v))} H_{1} + \frac{\om K M(u_{1})}{u_{1} r^{\om}(u,v)} \calB_{1}(U).
\end{align*}
in the case $1<\om\leq 2$. By the decay results proved in \S \ref{sec.full.decay.1}, we have
$$\sup_u(H_{1}+ \calB_{1}(u))\leq A$$
for some constant $A$. Therefore, by choosing $u_1$ sufficiently large, we have in the region $3u\leq v$,
$$r^{\om}\abs{\rd_{v} (r \phi)(u,v)-\rd_{v} (r \phi)(1, v)}\leq \frac{L}{4},$$
as long as $u\geq u_1$.
We now apply the assumption on the lower bound for the initial data $r^{\omg} \rd_{v} (r \phi)(1, v)\geq L$ for $v\geq V$. Choosing $u$ larger if necessary, we can assume that $u\geq V$. Then, we derive that in $3u\leq v$,
$$r^{\om}\rd_{v} (r \phi)(u,v)\geq \frac{L}{2}.$$

We now move to the interior region where $3u \geq v$. To this end, we improve the bounds in \eqref{eq:decay1:intr:pf:1}. First, notice that the lower bound in the exterior region implies that there exists $L'$ such that 
\bea
u^{\om}\rd_{v} (r \phi)(u,v)\geq L'\label{lower.bd.L1}
\eea
for $3u \leq v$. Then, integrating \eqref{eq:SSESF:dphi} along the incoming direction from $(u/3, v)$ to $(u, v)$, we get
\begin{equation*}
\begin{aligned}
	\abs{\rd_{v} (r \phi) (u,v)-\rd_{v} (r \phi) (u/3,v)} 
	\leq  \frac{1}{2} (\sup_{u' \in [u/3, u]} \sup_{C_{u'}} \abs{\phi}) \int_{u/3}^{u} \abs{\frac{2m \dur}{(1-\mu) r^{2}} (u', v)} \, \ud u'.
\end{aligned}
\end{equation*} 
By Theorem \ref{main.thm.1}, we have 
$$\sup_{C_u}|\phi|\leq A_{1} u^{-\om}$$
for some $A_{1} > 0$. Lemma \ref{lem:smallPtnl} implies that 
$$\int_{u/3}^{u} \abs{\frac{2m \dur}{(1-\mu) r^{2}} (u', v)} \, \ud u'\to 0$$
as $u\to \infty$. Thus the right hand side can be bounded by $\frac{L' u^{\om}}{2}$ after choosing $u$ to be sufficiently large.
Combining this with the lower bound \eqref{lower.bd.L1}, we have
$$u^{\om}\rd_{v} (r \phi) (u,v) \geq \frac{L'}{2}$$
for $3u \leq v$ and $u$ sufficiently large.

We now proceed to obtain the lower bound for $\rd_u(r\phi)$ by revisiting the proof of Lemma \ref{lem:decay1:uDecay4durphi}. Integrating \eqref{eq:SSESF:dphi} along the outgoing direction from $(u, u)$ to $(u, v)$, we have
\begin{align}\label{est.durphi.diff}
	\abs{\rd_{u} (r \phi)(u,v) - \lim_{v' \to u+} \rd_{u} (r \phi)(u, v')}
	\leq & \int_{C_{u}} \abs{\frac{\mu\dvr\dur}{(1-\mu)r}\phi}.
\end{align}
As before, we Theorem \ref{main.thm.1}, i.e.,
$$\sup_{C_u}|\phi|\leq A_{1} u^{-\om}$$
for some $A_{1} > 0$. By Lemma \ref{lem:smallPtnl} and the upper bound \eqref{eq:bnd4dur} for $|\nu|$, we have
$$\int_{C_{u}} \abs{\frac{\mu\dvr\dur}{(1-\mu)r}} \to 0, \quad\mbox{as }u\to\infty.$$
Therefore, we can choose $u$ sufficiently large such that
$$u^{\omg} \int_{C_{u}} \abs{\frac{\mu\dvr\dur}{(1-\mu)r}\phi}\leq \frac{L'}{4}.$$
Returning to \eqref{est.durphi.diff} and recalling that for $u$ large,
$$-\lim_{v' \to u+} \rd_{u} (r \phi)(u, v')=\lim_{v' \to u+} \rd_{v} (r \phi)(u, v')\geq \frac{L'}{2}u^{-\om},$$
we have
$$-\rd_{u} (r \phi)(u,v) \geq \frac{L'}{4}u^{-\om}$$
for $u$ sufficiently large, as desired. \qedhere

\end{proof}

\subsection{Key lower bound lemma}
The goal of the remainder of this section is to prove Theorem \ref{thm.opt.2}. In this subsection we establish the following result, which provides a sufficient condition for the desired lower bounds on the decay of $\phi$ in terms of a number (called $\mathfrak{L}$) computed on $\calI^{+}$. This will be an important ingredient for our proof of Theorem \ref{thm.opt.2} in the next subsection.

\begin{lemma}[Key lower bound lemma] \label{lem:LB}
Let $(\phi, r, m)$ be a $C^{1}$ solution to \eqref{eq:SSESF} which is locally BV scattering and asymptotically flat initial data of order $\omg = 3$ in $C^{1}$. 
Suppose furthermore that 
\begin{equation*}
	\mathfrak{L} := \lim_{v \to \infty} r^{3} \rd_{v} (r \phi)(1, v) + \int_{1}^{\infty} (M \dur_{\infty} \Phi)(u)  \, \ud u \neq 0,
\end{equation*}
where $M(u) := \lim_{v \to \infty} m(u, v)$, $\dur_{\infty}(u) := \lim_{v \to \infty} \dur(u,v)$ and $\Phi(u) := \lim_{v \to \infty} r \phi(u,v)$. 
Then there exist constants $U, L_{3} > 0$ such that the following lower bounds for the decay of $\rd_{v}(r \phi)$, $\rd_{u} (r \phi)$ hold on $\set{(u, v)  : u \geq U}$.
\begin{align} 
	\abs{\rd_{v}(r \phi)(u, v)} \geq & L_{3} \min \set{r(u,v)^{-3}, u^{-3}}, \label{eq:LB:1} \\
	\abs{\rd_{u}(r \phi)(u, v)} \geq & L_{3} u^{-3}. \label{eq:LB:2}
\end{align}
\end{lemma}

\begin{remark} 
Note that $\rd_{v}(r \phi)$ and $\rd_{u}(r \phi)$ have definite signs by \eqref{eq:LB:1}, \eqref{eq:LB:2}. In fact, the proof below shows that the signs of $\rd_{v} (r \phi)$ and $-\rd_{u} (r \phi)$ agree with that of $\mathfrak{L}$.
\end{remark}

\begin{proof} 
Without loss of generality, assume that $\mathfrak{L} > 0$. For $0 < \eta \leq 1$, define the \emph{$\eta$-exterior region} by
\begin{equation*}
	\extr^{\eta} := \set{(u,v) \in \PD : u \leq \eta v}.
\end{equation*}

\pfstep{Step 1} In the first step, we make precise the relation between $r$ and $v$ in $\extr^{\eta}$ for small $\eta$. We claim that $r \aeq v/2$ in this region; more precisely,
\begin{equation} \label{eq:LB:pf:1}
\abs{\frac{r(u,v)}{v} - \frac{1}{2}} \leq \eta C_{A_{1}, A_{2}, K, \Lmb}.
\end{equation}

Integrating by parts, we have
\begin{align*}
r(u,v) 
= \int_{u}^{v} \dvr(u, v') \, \ud v' 
= - \int_{u}^{v} \rd_{v} \dvr(u, v') v' \, \ud v' + v \dvr(u, v) - u \dvr(u, u).
\end{align*}

To make the leading term $v \dvr(u,v)$ and small number $\frac{u}{v}$ explicit, we rewrite the last expression as follows:
\begin{equation*}
	r(u,v) = v \bb[ \dvr(u,v)  - \frac{u}{v} \bb( \dvr(u,u) + \int_{u}^{v} \rd_{v} \dvr(u, v') \frac{v'}{u} \, \ud v' \bb)  \bb].
\end{equation*}

Recall that $\dvr$ is uniformly bounded from the above and below on $\PD$, i.e., $\Lmb^{-1} \leq \dvr \leq 1/2$. Moreover, by the decay estimates for $\rd_{v} \dvr$ proved in Theorem \ref{main.thm.2}, we have
\begin{equation*}
\sup_{(u,v) \in \PD} \int_{u}^{v} \abs{\rd_{v} \dvr(u, v')} \frac{v'}{u} \, \ud v' \leq C_{A_{2}}.
\end{equation*}

As a consequence,
\begin{equation*}
\abs{\frac{r(u,v)}{v} - \dvr(u,v)} \leq \eta C_{A_{2}, \Lmb}.
\end{equation*}

Thus \eqref{eq:LB:pf:1} will follow once we establish
\begin{equation} \label{eq:LB:pf:1:1}
	\abs{\dvr(u,v) - \frac{1}{2}} \leq \eta^{2} C_{A_{1}, A_{2}, K, \Lmb}.
\end{equation}

This inequality is proved by integrating the decay estimate \eqref{eq:decay2:10} for $\rd_{u} \dvr = \rd_{u} \rd_{v} r$ along the incoming direction, starting from the normalization $\dvr(1, v) = 1/2$. 
Here, we use the easy geometric fact that if $(u,v)$ lies in $\extr^{\eta}$, then so does the incoming null curve from $(1,v)$ to $(u,v)$. 

\pfstep{Step 2} We claim that for $U_{1} \geq 1$ sufficiently large and $0 < \eta \leq 1$ suitably small, we have
\begin{equation} \label{eq:LB:pf:2}
	\rd_{v} (r \phi)(u,v) \geq \frac{\mathfrak{L}}{2} \bb( \frac{v}{2} \bb)^{-3}
\end{equation}
for $(u, v) \in \extr^{\eta} \cap \set{u \geq U_{1}}$.

We begin with
\begin{equation} \label{eq:LB:pf:2:1}
	\bb( \frac{v}{2} \bb)^{3} \rd_{v}(r \phi)(u,v) = \bb( \frac{v}{2} \bb)^{3} \rd_{v}(r \phi)(1, v) + \bb( \frac{v}{2} \bb)^{3} \int_{1}^{u} \frac{2m \dvr \dur}{(1-\mu) r^{3}} r \phi(u', v) \, \ud u',
\end{equation}
obtained by integrating the $\rd_{u} \rd_{v} (r \phi)$ equation and multiplying by $(v/2)^{3}$. To prove \eqref{eq:LB:pf:2}, it suffices to show that the right-hand side of \eqref{eq:LB:pf:2:1} is bounded from below by $\mathfrak{L}/2$ for $(u, v) \in \extr^{\eta} \cap \set{u \geq U_{1}}$ with sufficiently large $U_{1} \geq 1$ and small $0 < \eta \leq 1$.

Note that $r = \frac{v-1}{2}$ on $C_{1}$, and $v \geq \eta^{-1}$ if $(u,v) \in \extr^{\eta}$. Thus for $(u,v) \in \extr^{\eta}$ and $0 < \eta \leq 1$ sufficiently small, we have
\begin{equation*}
\abs{\bb( \frac{v}{2} \bb)^{3} \rd_{v} (r \phi)(1, v) - \lim_{v \to \infty} r^{3} \rd_{v} (r \phi)(1, v)} < \frac{\mathfrak{L}}{8}.
\end{equation*}

In order to proceed, it is useful to keep in mind the following technical point: For $U_{1} \geq 1$, by the decay estimates \eqref{eq:decay1:1} and \eqref{eq:decay1:6}, we have
\begin{equation} \label{eq:LB:pf:2:2}
	\sup_{v \geq U_{1}} \int_{U_{1}}^{v} \abs{\frac{2m \dvr \dur}{1-\mu} r \phi (u', v)} \, \ud u' \leq U_{1}^{-6} C_{A_{1}, \Lmb}.
\end{equation}

In what follows, let $(u, v) \in \extr^{\eta} \cap \set{u \geq U_{1}}$. Using \eqref{eq:LB:pf:1}, \eqref{eq:LB:pf:2:2} and the fact that the null segment from $(1, v)$ to $(u, v)$ lies in $\extr^{\eta}$, we get
\begin{equation*}
\bb\vert \bb( \frac{v}{2} \bb)^{3} \int_{1}^{u} \frac{2m \dvr \dur}{(1-\mu) r^{3}} r \phi(u', v) \, \ud u' - \int_{1}^{u} \frac{2m \dvr \dur}{1-\mu} r \phi(u', v) \, \ud u'\bb\vert
\leq \eta C_{A_{1}, A_{2}, K, \Lmb}.
\end{equation*}

Taking $U_{1} \geq 1$ large enough and using \eqref{eq:LB:pf:2:2}, we may arrange 
\begin{equation*}
	\sup_{v \geq U_{1}} \int_{U_{1}}^{v} \abs{\frac{2m \dvr \dur}{1-\mu} r \phi (u', v)} \, \ud u' + \int_{U_{1}}^{\infty} \abs{M \dur_{\infty} \Phi (u')} \, \ud u' < \frac{\mathfrak{L}}{8}.
\end{equation*} 

On the other hand, note that $2m \dvr \dur (1-\mu)^{-1} r \phi (u, v) \to M \dur_{\infty} \Phi(u)$ for each $u \geq 1$ as $v \to \infty$. Therefore, by the dominated convergence theorem, for $0 < \eta \leq 1$ sufficiently small (so that $v$ is large), we have
\begin{align*}
	\abs{\int_{1}^{U_{1}} \frac{2m \dvr \dur}{1-\mu} r \phi (u', v) \, \ud u' - \int_{1}^{U_{1}} M \dur_{\infty} \Phi(u')  \, \ud u'} < \frac{\mathfrak{L}}{8}.
\end{align*}

Putting these together and taking $0 < \eta \leq 1$ sufficiently small, we conclude \eqref{eq:LB:pf:2}.

\pfstep{Step 3} Next, we claim that there exists $U_{2} = U_{2}(U_{1}, A_{2}, \Lmb, K, \eta) \geq 1$ such that $U_{2} \geq U_{1}$ and for $(u,v) \in (\PD \setminus \extr^{\eta}) \cap \set{u \geq U_{2}}$, we have
\begin{align} 
	\rd_{v}(r \phi)(u, v) \geq 2 \eta^{3} \mathfrak{L} \, u^{-3}.		\label{eq:LB:pf:3}
\end{align}

Combined with \eqref{eq:LB:pf:2} (keeping in mind that $r \aeq v/2$ in $\extr^{\eta}$ by \eqref{eq:LB:pf:1}), this would establish \eqref{eq:LB:1}.

Take $U_{2} \geq \eta^{-1} U_{1}$, and consider $(u, v) \in (\PD \setminus \extr^{\eta}) \cap \set{u \geq U_{2}}$. Integrating \eqref{eq:SSESF:dphi}, we have
\begin{equation*}
	\rd_{v} (r \phi)(u,v) = \rd_{v} (r \phi)(\eta u, v) + \int_{\eta u}^{u} \frac{2m \dvr \dur}{r^{2}} \phi(u', v) \, \ud u'.
\end{equation*}
 
Note that $(\eta u, v) \in \extr^{\eta} \cap \set{u \geq U_{1}}$ since $v \geq u$ and $\eta u \geq \eta U_{2} \geq  U_{1}$. Therefore, by \eqref{eq:LB:pf:2} and the fact that $\eta^{-1} u > v$ (as $(u,v) \in \PD \setminus \extr^{\eta}$), the first term on the right-hand side obeys the lower bound
\begin{equation*}
	\rd_{v} (r \phi)(\eta u, v) \geq \bb( \frac{\mathfrak{L}}{2} \bb) \bb( \frac{v}{2} \bb)^{-3} > 4 \eta^{3} \mathfrak{L} \, u^{-3}.
\end{equation*}

On the other hand, using \eqref{eq:decay1:1} and \eqref{eq:decay2:11}, we have
\begin{equation*}
	\abs{\int_{\eta u}^{u} \frac{2m \dvr \dur}{r^{2}} \phi(u', v) \, \ud u'}
	\leq C_{A_{1}, A_{2}, \Lmb, K} \int_{\eta u}^{u} \frac{1}{(u')^{10}} \, \ud u' 
	\leq C_{A_{1}, A_{2}, \Lmb, K, \eta} \, u^{-9}.
\end{equation*}

Taking $U_{2}$ large enough, we conclude that \eqref{eq:LB:pf:3} holds.

\pfstep{Step 4} Finally, we claim that there exists $U = U(U_{2}, A_{2}, \Lmb, K, \eta) \geq 1$ such that $U \geq U_{2} \geq U_{1}$ and for $(u,v) \in \set{u \geq U}$, we have
\begin{align} 
	- \rd_{u}(r \phi)(u, v) \geq \eta^{3} \mathfrak{L} \, u^{-3}.		\label{eq:LB:pf:4}
\end{align}

This would prove \eqref{eq:LB:2}, thereby completing the proof of Lemma \ref{lem:LB}.

Our argument will be very similar to the previous step. Take $U \geq U_{2}$ and consider $(u,v) \in \set{u \geq U}$. Integrating \eqref{eq:SSESF:dphi} along the outgoing direction, we have
\begin{equation*}
	- \rd_{u}(r \phi)(u,v) = - \rd_{u}(r \phi)(u, u) - \int_{u}^{v} \frac{2 m \dvr \dur}{r^{2}} \phi(u, v') \, \ud v'.
\end{equation*}

Recall that $\lim_{v \to u+} \rd_{u} (r \phi)(u, v) = - \lim_{v \to u+} \rd_{v}(r \phi)(u,v)$. By \eqref{eq:LB:pf:3} and the fact that $u \geq U \geq U_{2}$, we see that the first term on the right-hand side obeys the lower bound
\begin{equation*}
	- \rd_{u} (r \phi)(u,u) \geq 2 \eta^{3} \mathfrak{L} \, u^{-3}.
\end{equation*}

On the other hand, using \eqref{eq:decay1:1} and \eqref{eq:decay2:11}, we have
\begin{equation*}
	\abs{\int_{u}^{v} \frac{2 m \dvr \dur}{r^{2}} \phi(u, v') \, \ud v'}
	\leq C_{A_{1}, A_{2}, K} \int_{u}^{v} \min \set{u^{-10}, r^{-2} u^{-8}} \, \dvr \, \ud v' 
	\leq C_{A_{1}, A_{2}, K} \, u^{-9}.
\end{equation*}

Taking $U$ sufficiently large, we conclude that \eqref{eq:LB:pf:4} holds. \qedhere
\end{proof}

\subsection{Optimality of the decay rates, in the case $\omg' \geq 3$}  \label{subsec.opt.2}
In this subsection, we prove Theorem \ref{thm.opt.2} by studying the solution to \eqref{eq:SSESF} arising from the initial value
\begin{equation*}
	\rd_{v}(r \phi)(1, v) = \eps \widetilde{\chi}\bb( \frac{v - v_{0}}{N} \bb),
\end{equation*}
where $\widetilde{\chi} : (-\infty, \infty) \to [0, \infty)$ is a smooth function such that
\begin{equation*}
\mathrm{supp} \, \widetilde{\chi} \subset (-1/2, 1/2), \quad
\int_{\bbR} \widetilde{\chi} = 1.
\end{equation*}

We also require that $v_{0} \geq 2$ and $N \leq v_{0}$. With such data, the initial total variation is of size $\leq C\eps$, i.e.,
\begin{equation*}
	\int_{1}^{\infty} \abs{\rd_{v}^{2} (r \phi) (1, v)} \, \ud v \leq \eps \int_{-\infty}^{\infty} \abs{ \widetilde{\chi}\,' \bb( \frac{v - v_{0}}{N} \bb) } \, \frac{\ud v}{N} \leq C \eps.
\end{equation*}

We also see that $\calI_{1} \leq C \eps v_{0}^{3}$ and $\calI_{2} \leq C \eps v_{0}^{4} / N$ with $\omg' = 3$, as
\begin{equation*}
	\sup_{v \in [1, \infty)} (1+r)^{3} \abs{\rd_{v} (r \phi)}(1,v) \leq C \eps v_{0}^{3}, \quad
	\sup_{v \in [1, \infty)} (1+r)^{4} \abs{\rd_{v}^{2} (r \phi)}(1,v) \leq C \eps \frac{v_{0}^{4}}{N}.
\end{equation*}

We are now ready to give a proof of Theorem \ref{thm.opt.2}. The idea is to compute $\mathfrak{L}$ to the leading order (which turns out to be $- c\eps^{3}$ for some $c > 0$), and then control the lower order terms by taking $\eps > 0$ sufficiently small and applying Theorem \ref{thm:smallData}.

\begin{proof} [Proof of Theorem \ref{thm.opt.2}]
For this proof, we fix $v_{0} = 4$ and $N = 1$. We use the shorthand
\begin{equation*}
	\chi(v) := \widetilde{\chi}(v-4).
\end{equation*}

By the preceding discussion on the size of initial data, we see that Theorem \ref{thm:smallData} applies when $\eps > 0$ is sufficiently small. Therefore, there exists a constant $C > 0$ independent of $\eps > 0$ such that Theorems \ref{main.thm.1}, \ref{main.thm.2} and Proposition \ref{prop:geomLocBVScat} hold with
\begin{equation} \label{eq:opt2:eps}
	A_{1}, A_{2} \leq C \eps, \quad K, \Lmb \leq C.
\end{equation}

We begin by showing
\begin{equation} \label{eq:opt2:dvrphi}
	\rd_{v}(r \phi)(u,v) = \eps \chi(v) + \Err_{1}(u,v),
\end{equation}
where
\begin{equation} \label{eq:opt2:Err1}
	\abs{\Err_{1}(u,v)} \leq C \eps^{3} \min \set{u^{-3}, r(u,v)^{-3}}.
\end{equation}

The argument is similar to the proof of Theorem \ref{thm.opt.1}, but this time we rely on Theorem \ref{thm:smallData} to make the dependence of $\Err_{1}$ on $\eps$ explicit. Indeed, by \eqref{eq:SSESF:dphi}, we have
\begin{equation*}
	\abs{\Err_{1}(u,v)}
	\leq \int_{1}^{u} \abs{\frac{\mu \dvr \dur}{(1-\mu)r}  \phi} (u', v) \, \ud u'.
\end{equation*}

Then estimating the right-hand side using Theorem \ref{main.thm.1}, Proposition \ref{prop:geomLocBVScat} and Corollary \ref{cor:decay2}, and using \eqref{eq:opt2:eps} make the $\eps$-dependence explicit, \eqref{eq:opt2:Err1} follows.

Integrating \eqref{eq:opt2:dvrphi}, we also have
\begin{align*} 
	r\phi(u,v) 	= & \int_{u}^{v} \rd_{v}(r \phi)(u, v') \, \ud v' \\
			= & \eps \int_{u}^{v} \chi(v') \, \ud v' + \int_{u}^{v} \Err_{1}(u, v') \, \ud v' \\
			= & \eps X(u,v) + \Err_{2}(u,v)
\end{align*}
where $X(u, v) := \int_{u}^{v} \chi(v') \, \ud v'$ and $\Err_{2}(u,v) := \int_{u}^{v} \Err_{1}(u, v') \, \ud v'$. Integrating \eqref{eq:opt2:Err1}, and using the bound $C^{-1} \leq \dvr \leq 1/2$, we easily obtain
\begin{equation} \label{eq:opt2:Err2}
	\abs{\Err_{2}(u,v)} \leq C \eps^{3} \min \set{r u^{-3}, u^{-2}}.	
\end{equation}

In particular, taking $v \to \infty$, we see that
\begin{equation} \label{eq:opt2:Phi}
	\abs{\Phi(u) - \eps X(u, \infty)} \leq C \eps^{3} u^{-2}.
\end{equation}

We now proceed to estimate $M(u)$. We begin with the easy observation
\begin{equation} \label{eq:opt2:UBforM}
	M(u) \leq C \eps^{2} u^{-5},
\end{equation}
which follows from Corollary \ref{cor:decay2} and \eqref{eq:opt2:eps}. On the other hand, recalling the definition of $M(u)$ from \eqref{eq:SSESF:dm} and using the elementary inequality $(a+b)^{2} \geq \frac{1}{2} a^{2} - b^{2}$,
\begin{align*}
	M(u)	= & \frac{1}{2} \int_{u}^{\infty} \frac{1-\mu}{\dvr} [ \rd_{v}(r\phi) - \frac{\dvr}{r} (r \phi) ]^{2} (u,v) \, \ud v  \\
		\geq & \frac{\eps^{2}}{4} \int_{u}^{\infty} \frac{1-\mu}{\dvr}(u,v) [\chi (v) - \frac{\dvr}{r} X (u,v) ]^{2} \, \ud v 
			- \frac{1}{2}\int_{u}^{\infty} \frac{1-\mu}{\dvr} [\Err_{1} - \frac{\dvr}{r} \Err_{2}]^{2} (u,v) \, \ud v .
\end{align*}

By \eqref{eq:opt2:eps}, \eqref{eq:opt2:Err1} and \eqref{eq:opt2:Err2}, we have
\begin{equation*}
\abs{\frac{1}{2} \int_{u}^{\infty} \frac{1-\mu}{\dvr} [\Err_{1} - \frac{\dvr}{r} \Err_{2}]^{2} (u,v) \, \ud v} \leq C \eps^{6}.
\end{equation*}

Furthermore, note that $(1-\mu) \geq (K \Lmb)^{-1} \geq C^{-1} > 0$ by Proposition \ref{prop:geomLocBVScat} and \eqref{eq:opt2:eps}. Also, for $(u,v) \in [1,2] \times [8, \infty)$, note that $\chi(v) = 0$ and $X(u,v) = 1$. Therefore, for $1 \leq u \leq 2$, there exists $c > 0$ (independent of $\eps > 0$) such that
\begin{align*}
\frac{1}{4}\int_{u}^{\infty} \frac{1-\mu}{\dvr}(u,v) [\chi - \frac{\dvr}{r} X]^{2} (u,v) \, \ud v
\geq & (4 C)^{-1} \int_{u}^{\infty}  [\chi - \frac{\dvr}{r} X]^{2} (u,v) \, \dvr^{-1} (u,v)  \ud v \\
\geq & (4 C)^{-1} \int_{8}^{\infty} \frac{\dvr}{r^{2}} (u,v) \, \ud v \\
\geq & c\,.
\end{align*}

Therefore, we conclude that
\begin{equation} \label{eq:opt2:LBforM}
	M(u) \geq c \eps^{2} - C \eps^{6} \qquad \hbox{ for } 1 \leq u \leq 2.
\end{equation}

We are now ready to compute $\mathfrak{L}$ and complete the proof. We begin by observing that 
\begin{equation*}
\lim_{v \to \infty} r^{3} \abs{\rd_{v}(r \phi)(1,v)} = 0
\end{equation*}
by our choice of data. Therefore,
\begin{align*}
	- \mathfrak{L}
		 = &\int_{1}^{\infty} M (-\dur_{\infty}) \Phi(u) \, \ud u \\
		= & \eps \int_{1}^{\infty} M (u) (-\dur_{\infty}) (u) X(u,\infty) \, \ud u 
			+ \int_{1}^{\infty} M(u) (-\dur_{\infty})(u) \Err_{2}(u, \infty) \, \ud u. 
\end{align*}

By Proposition \ref{prop:geomLocBVScat}, \eqref{eq:opt2:eps}, \eqref{eq:opt2:Err2} and \eqref{eq:opt2:UBforM}, we have
\begin{equation*}
\abs{\int_{1}^{\infty} M(u) (-\dur_{\infty})(u) \Err_{2}(u, \infty) \, \ud u } \leq C \eps^{5}.
\end{equation*}

On the other hand, by Proposition \ref{prop:geomLocBVScat}, \eqref{eq:opt2:eps} and \eqref{eq:opt2:LBforM}, we have (taking $c > 0$ smaller if necessary)
\begin{align*}
		\eps \int_{1}^{\infty} M(u) (-\dur_{\infty})(u) X(u, \infty) \, \ud u 
		\geq & \eps \int_{1}^{2} M(u) (-\dur_{\infty})(u) X(u, \infty) \, \ud u \\
		\geq & \Lmb^{-1} \eps \int_{1}^{2} M(u) \, \ud u
		\geq c \eps^{3} - C \eps^{7}.
\end{align*}

Therefore, taking $\eps > 0$ sufficiently small, we see that $- \mathfrak{L} > \frac{c}{2} \eps^{3} > 0$. \qedhere
\end{proof}

\bibliographystyle{amsplain}


\end{document}